\documentclass[10pt]{article}

\usepackage{palatino}
\usepackage{float}

\usepackage{url}
\usepackage{amsmath,amsthm,amssymb,amsbsy}
\usepackage{mathdots}
\usepackage{paralist}
\usepackage{xcolor}
\usepackage{color}
\usepackage{graphicx}
\usepackage{algorithm,algpseudocode}
\usepackage{comment}
\usepackage{bbm}
\usepackage{fancyhdr}
\usepackage{cite}
\usepackage{cleveref}
\usepackage{enumerate}

\newtheorem{lemma}{Lemma}

\newtheorem{theorem}{Theorem}

\theoremstyle{remark}

\newcommand{\R}{\mathbb{R}}

\newcommand{\C}{\mathbb{C}}

\newcommand{\Z}{\mathbb{Z}}
\newcommand{\N}{\mathbb{N}}

\renewcommand{\P}{\mathbb{P}}
\newcommand{\E}{\mathbb{E}}
\newcommand{\Var}{\operatorname{Var}}
\newcommand{\Cov}{\operatorname{Cov}}
\newcommand{\Bias}{\operatorname{Bias}}

\renewcommand{\d}[1]{\,\mathrm{d}#1}

\newcommand{\e}{\mathrm{e}}

\newcommand{\vct}[1]{\boldsymbol{#1}}
\newcommand{\mtx}[1]{\boldsymbol{#1}}

\newcommand{\inner}[1]{\left<#1\right>}

\newcommand{\tr}{\operatorname{tr}}

\newcommand{\diag}{\operatorname{diag}}
\newcommand{\set}[1]{\mathcal{#1}}

\DeclareMathOperator*{\minimize}{\text{minimize}}

\newcommand{\floor}[1]{\left\lfloor #1 \right\rfloor}
\newcommand{\ceil}[1]{\left\lceil #1 \right\rceil}
\newcommand{\eps}{\epsilon}

\newcommand{\va}{\vct{a}}
\newcommand{\vb}{\vct{b}}
\newcommand{\vc}{\vct{c}}
\newcommand{\vd}{\vct{d}}
\newcommand{\ve}{\vct{e}}

\newcommand{\vp}{\vct{p}}
\newcommand{\vq}{\vct{q}}

\newcommand{\vs}{\vct{s}}

\newcommand{\vv}{\vct{v}}
\newcommand{\vw}{\vct{w}}
\newcommand{\vx}{\vct{x}}
\newcommand{\vy}{\vct{y}}
\newcommand{\vz}{\vct{z}}

\newcommand{\mA}{\mtx{A}}
\newcommand{\mB}{\mtx{B}}
\newcommand{\mC}{\mtx{C}}
\newcommand{\mD}{\mtx{D}}
\newcommand{\mE}{\mtx{E}}
\newcommand{\mF}{\mtx{F}}

\newcommand{\mP}{\mtx{P}}
\newcommand{\mQ}{\mtx{Q}}
\newcommand{\mR}{\mtx{R}}
\newcommand{\mS}{\mtx{S}}

\newcommand{\mU}{\mtx{U}}
\newcommand{\mV}{\mtx{V}}
\newcommand{\mW}{\mtx{W}}
\newcommand{\mX}{\mtx{X}}

\newcommand{\mZ}{\mtx{Z}}
\newcommand{\mBext}{\mtx{B}_{\text{ext}}}

\newcommand{\mLambda}{\mtx{\Lambda}}

\newcommand{\mGamma}{\mtx{\Gamma}}

\newcommand{\mId}{{\bf I}}

\newcommand{\setI}{\set{I}}

\newcommand{\hatS}{\widehat{S}}
\newcommand{\hatSmt}{\widehat{S}^{\text{mt}}}

\newcommand{\hatSad}{\widehat{S}^{\text{ad}}}

\newcommand{\tildeSmt}{\widetilde{S}^{\text{mt}}}

\newcommand{\tildegamma}{\widetilde{\gamma}}

\newcommand{\mt}{\text{mt}}

\graphicspath{{./figs/}}

\newlength{\imgwidth}
\newboolean{twoColVersion}
\setboolean{twoColVersion}{false}
\newcommand{\twoCol}[2]{\ifthenelse{\boolean{twoColVersion}} {#1} {#2} }

\newcommand\blfootnote[1]{
  \begingroup
  \renewcommand\thefootnote{}\footnote{#1}
  \addtocounter{footnote}{-1}
  \endgroup
}

\begin{document}

\title{Thomson's Multitaper Method Revisited}

\vspace{2mm}
\author{Santhosh Karnik, Justin Romberg, Mark A. Davenport}

\maketitle

\begin{abstract}
Thomson’s multitaper method estimates the power spectrum of a signal from $N$ equally spaced samples by averaging $K$ tapered periodograms. Discrete prolate spheroidal sequences (DPSS) are used as tapers since they provide excellent protection against spectral leakage. Thomson's multitaper method is widely used in applications, but most of the existing theory is qualitative or asymptotic. Furthermore, many practitioners use a DPSS bandwidth $W$ and number of tapers that are smaller than what the theory suggests is optimal because the computational requirements increase with the number of tapers. We revisit Thomson's multitaper method from a linear algebra perspective involving subspace projections. This provides additional insight and helps us establish nonasymptotic bounds on some statistical properties of the multitaper spectral estimate, which are similar to existing asymptotic results. We show using $K=2NW-O(\log(NW))$ tapers instead of the traditional $2NW-O(1)$ tapers better protects against spectral leakage, especially when the power spectrum has a high dynamic range. Our perspective also allows us to derive an $\eps$-approximation to the multitaper spectral estimate which can be evaluated on a grid of frequencies using $O(\log(NW)\log\tfrac{1}{\eps})$ FFTs instead of $K=O(NW)$ FFTs. This is useful in problems where many samples are taken, and thus, using many tapers is desirable.

\end{abstract}

\blfootnote{S. Karnik, J. Romberg, and M. A. Davenport are with the School of Electrical and Computer Engineering, Georgia Institute of Technology, Atlanta, GA, 30332 USA (e-mail: skarnik1337@gatech.edu, jrom@ece.gatech.edu, mdav@gatech.edu). This work was partially supported by NSF grant CCF-1409406, a grant from Lockheed Martin, and a gift from the Alfred P. Sloan Foundation. A preliminary version of this paper highlighting some of the key results also appeared in~\cite{Karnik19}.}

\section{Introduction}
\label{sec:Intro}
Perhaps one of the most fundamental problems in digital signal processing is spectral estimation, i.e.,\ estimating the power spectrum of a signal from a window of $N$ evenly spaced samples. The simplest solution is the periodogram, which simply takes the squared-magnitude of the discrete time Fourier transform (DTFT) of the samples. Obtaining only a finite number of samples is equivalent to multiplying the signal by a rectangular function before sampling. As a result, the DTFT of the samples is the DTFT of the signal convolved with the DTFT of the rectangular function, which is a slowly-decaying sinc function. Hence, narrow frequency components in the true signal appear more spread out in the periodogram. This phenomenon is known as ``spectral leakage''. 

The most common approach to mitigating the spectral leakage phenomenon is to multiply the samples by a taper before computing the periodogram. Since multiplying the signal by the taper is equivalant to convolving the DTFT of the signal with the DTFT of the taper, using a taper whose DTFT is highly concentrated around $f = 0$ will help mitigate the spectral leakage phenomenon. Numerous kinds of tapers have been proposed \cite{McClellan03} which all have DTFTs which are highly concentrated around $f = 0$. The Slepian basis vectors, also known as the discrete prolate spheroidal sequences (DPSSs), are designed such that their DTFTs have a maximal concentration of energy in the frequency band $[-W,W]$ subject to being orthonormal \cite{SlepianV}. The first $\approx 2NW$ of these Slepian basis vectors have DTFTs which are highly concentrated in the frequency band $[-W,W]$. Thus, any of the first $\approx 2NW$ Slepian basis vectors provides a good choice to use as a taper.

In 1982, David Thomson \cite{Thomson82} proposed a multitaper method which computes a tapered periodogram for each of the first $K \approx 2NW$ Slepian tapers, and then averages these periodograms. Due to the spectral concentration properties of the Slepian tapers, Thomson's multitaper method also does an excellent job mitigating spectral leakage. Furthermore, by averaging $K$ tapered periodograms, Thomson's multitaper method is more robust than a single tapered periodogram. As such, Thomson's multitaper method has been used in a wide variety of applications, such as cognitive radio \cite{Haykin05, Farhang08Filterbank, Farhang08Multicarrier, Haykin09, Axell12}, digital audio coding \cite{Hamdy96, Painter00}, as well as to analyze EEG \cite{Delorme04,Delorme07} and other neurological signals \cite{Llinas99, Pesaran02, Csicsvari03, Mitra99, Jones05}, climate data \cite{Bond97, Ghil02, Vautard89, Mann96, Minobe97, Jouzel93, Thomson90}, breeding calls of Ad{\'e}lie penguins \cite{Brunton10} and songs of other birds \cite{Saar08, Hansson11, Leonardo99, Tchernichovski00}, topography of terrestrial planets \cite{Wieczorek07}, solar waves\cite{Claudepierre08}, and gravitational waves \cite{Allen99}.

The existing theoretical results regarding Thomson's multitaper method are either qualitative or asymptotic. Here, we provide a brief summary of these results. See \cite{Thomson82,Percival93,Walden00,Haykin09,Abreu17,Haley17} for additional details. Suppose that the power spectral density $S(f)$ of the signal is ``slowly varying''. Let $\hatSmt_K(f)$ denote the multitaper spectral estimate. Then, the following results are known. 
\begin{itemize}
\item The estimate is approximately unbiased, i.e., $\E\hatSmt_K(f) \approx S(f)$.

\item The variance is roughly $\Var[\hatSmt_K(f)] \approx \tfrac{1}{K}S(f)^2$.

\item For any two frequencies $f_1,f_2$ that are at least $2W$ apart, $\hatSmt_K(f_1)$ and $\hatSmt_K(f_2)$ are approximately uncorrelated.

\item The multitaper spectral estimate $\hatSmt_K(f)$ has a concentration behavior about its mean which is similar to a scaled chi-squared random variable with $2K$ degrees of freedom.

\item If the power spectral density $S(f)$ is twice differentiable, then choosing a bandwidth of $W = O(N^{-1/5})$ and using $K \approx 2NW = O(N^{4/5})$ tapers minimizes the mean squared error of the multitaper spectral estimate $\E|\hatSmt_K(f)-S(f)|^2$.
\end{itemize}

These asymptotic results demonstrate that using more than a small constant number of tapers improves the quality of the multitaper spectral estimate. However, using more tapers increases the computational requirements. As a result, many practitioners often use significantly fewer tapers than what is optimal. 

The main contributions of this work are as follows:
\begin{itemize}
\item We revisit Thomson's multitaper method from a linear algebra based perspective. Specifically, for each frequency $f$, the multitaper spectral estimate computes the $2$-norm of the projection of the vector of samples onto a $K$-dimensional subspace. The subspace chosen can be viewed as the result of performing principle component analysis on a continuum of sinusoids whose frequency is between $f-W$ and $f+W$. 

\item Using this perspective, we establish non-asymptotic bounds on the bias, variance, covariance, and probability tails of the multitaper spectral estimate. These non-asymptotic bounds are comparable to the known asymptotic results which assume that the spectrum is slowly varying. Also, these bounds show that using $K = 2NW-O(\log(NW))$ tapers (instead of the traditional choice of $K = \floor{2NW}-1$ or $\floor{2NW}-2$ tapers) provides better protection against spectral leakage, especially in scenarios where the spectrum has a large dynamic range.

\item We also use this perspective to demonstrate a fast algorithm for evaluating an $\eps$-approximation of the multitaper spectral estimate on a grid of $L$ equally spaced frequencies. The complexity of this algorithm is $O(L\log L \log(NW)\log \tfrac{1}{\eps})$ while the complexity of evaluating the exact multitaper spectral estimate on a grid of $L$ equally spaced frequencies is $O(KL\log L)$. Computing the $\eps$-approximation is faster than the exact multitaper provided $K \gtrsim \log(NW)\log \tfrac{1}{\eps}$ tapers are used. 
\end{itemize}

The rest of this work is organized as follows. In Section~\ref{sec:Multitaper} we formally introduce Thompson's multitaper spectral estimate, both from the traditional view involving tapered periodograms as well as from a linear algebra view involving projections onto subspaces. In Section~\ref{sec:ParameterSelection}, we state non-asymptotic bounds regarding the bias, variance, and probability concentration of the multitaper spectral estimate. These bounds are proved in Appendix~\ref{sec:ParameterSelectionProofs}. In Section~\ref{sec:FastAlgorithms}, we state our fast algorithm for evaluating the multitaper spectral estimate at a grid of frequencies. The proofs of the approximation error and computational requirements are in Appendix~\ref{sec:FastAlgorithmsProofs}. In Section~\ref{sec:Simulations}, we show numerical experiments which demonstrate that using $K = 2NW-O(\log(NW))$ tapers minimizes the impact of spectral leakage on the multitaper spectral estimate, and that our $\eps$-approximation to the multitaper spectral estimate can be evaluated at a grid of evenly spaced frequencies significantly faster than the exact multitaper spectral estimate. We finally conclude the paper in Section~\ref{sec:Conclusions}.

\section{Thomson's Multitaper Method}
\label{sec:Multitaper}

\subsection{Traditional view}

Let $x(n)$, $n \in \Z$ be a stationary, ergodic, zero-mean, Gaussian process. The autocorrelation and power spectral density of $x(n)$ are defined by $$R_n = \E\left[x(m)\overline{x(m+n)}\right] \quad \text{for} \quad m,n \in \Z,$$ and $$S(f) = \sum_{n = -\infty}^{\infty}R_ne^{-j2\pi fn} \quad \text{for} \quad f \in \R$$ respectively. The goal of spectral estimation is to estimate $S(f)$ from the vector $\vx \in \C^N$ of equispaced samples $\vx[n] = x(n)$ for $n \in [N]$. \footnote{We use the notation $[N]$ to denote the set $\{0,\ldots,N-1\}$. We will also use $[K]$ and $[L]$ in a similar manner.} 

One of the earliest, and perhaps the simplest estimator of $S(f)$ is the periodogram \cite{Stokes1879,Schuster1898} $$\hatS(f) = \dfrac{1}{N}\left|\sum_{n = 0}^{N-1}\vx[n]e^{-j2\pi fn}\right|^2.$$ This estimator can be efficiently evaluated at a grid of evenly spaced frequencies via the FFT. However, the periodogram has high variance and suffers from the spectral leakage phenomenon \cite{Harris78}. 

A modification to the periodogram is to pick a data taper $\vw \in \R^N$ with $\|\vw\|_2 = 1$, and then weight the samples by the taper as follows $$\hatS_{\vw}(f) = \left|\sum_{n = 0}^{N-1}\vw[n]\vx[n]e^{-j2\pi fn}\right|^2.$$ If $\vw[n]$ is small near $n = 0$ and $n = N-1$, then this ``smoothes'' the ``edges'' of the sample window. Note that the expectation of the tapered periodogram is given by a convolution of the true spectrum and the spectral window of the taper, $$\E[\hatS_{\vw}(f)] = S(f) \circledast \left|\widetilde{\vw}(f)\right|^2$$ where $$\widetilde{\vw}(f) = \sum_{n = 0}^{N-1}\vw[n]e^{-j2\pi fn}.$$ Hence, a good taper will have its spectral window $\left|\widetilde{\vw}(f)\right|^2$ concentrated around $f = 0$ so that $\E\left[\hatS_{\vw}(f)\right] = S(f) \circledast \left|\widetilde{\vw}(f)\right|^2 \approx S(f)$, i.e.,\ the tapered periodogram will be approximately unbiased.

For a given bandwidth parameter $W > 0$, we can ask ``Which taper maximizes the concentration of its spectral window, $\left|\widetilde{\vw}(f)\right|^2$, in the interval $[-W,W]$?'' Note that we can write $$\int_{-W}^{W}\left|\widetilde{\vw}(f)\right|^2\,df = \vw^*\mB\vw$$ where $\mB$ is the $N \times N$  prolate matrix \cite{Varah93,Bojanczyk95} whose entries are given by $$\mB[m,n] = \begin{cases}\dfrac{\sin[2\pi W(m-n)]}{\pi(m-n)} & \text{if} \ m \neq n, \\ 2W & \text{if} \ m = n.\end{cases}$$ Hence, the taper whose spectral window, $\left|\widetilde{\vw}(f)\right|^2$, is maximally concentrated in $[-W,W]$ is the eigenvector of $\mB$ corresponding to the largest eigenvalue. 

The prolate matrix $\mB$ was first studied extensively by David Slepian \cite{SlepianV}. As such, we will refer to the orthonormal eigenvectors $\vs_0, \vs_1, \ldots, \vs_{N-1} \in \R^N$ of $\mB$ as the Slepian basis vectors, where the corresponding eigenvalues $\lambda_0 \ge \lambda_1 \ge \cdots \ge \lambda_{N-1}$ are sorted in descending order. Slepian showed that all the eigenvalues of $\mB$ are distinct and strictly between $0$ and $1$. Furthermore, the eigenvalues of $\mB$ exhibit a particular clustering behavior. Specifically, the first slightly less than $2NW$ eigenvalues are very close to $1$, and the last slightly less than $N-2NW$ eigenvalues are very close to $0$. 

While $\vs_0$ is the taper whose spectral window is maximally concentrated in $[-W,W]$, any of the Slepian basis vectors $\vs_k$ for which $\vs_k^*\mB\vs_k = \lambda_k \approx 1$ will also have a high energy concentration in $[-W,W]$, and thus, make good tapers. Thomson \cite{Thomson82} proposed a multitaper spectral estimate by using each of the first $K \approx 2NW$ Slepian basis vectors as tapers, and taking an average of the resulting tapered periodograms, i.e., $$\hatSmt_K(f) = \dfrac{1}{K}\sum_{k = 0}^{K-1}\hatS_{k}(f) \quad \text{where} \quad \hatS_k(f) = \left|\sum_{n = 0}^{N-1}\vs_k[n]\vx[n]e^{-j2\pi fn}\right|^2.$$

The expectation of the multitaper spectral estimate satisfies $$\E\left[\hatSmt_K(f)\right] = S(f) \circledast \psi(f)$$ where $$\psi(f) = \dfrac{1}{K}\sum_{k = 0}^{K-1}\left|\sum_{n = 0}^{N-1}\vs_k[n]e^{-j2\pi fn}\right|^2$$ is known as the spectral window of the multitaper spectral estimate. It can be shown that when $K \approx 2NW$, the spectral window $\psi(f)$ approximates $\tfrac{1}{2W}\mathbbm{1}_{[-W,W]}(f)$ on $f \in [-\tfrac{1}{2},\tfrac{1}{2}]$. Thus, the multitaper spectral estimate behaves in expectation like a smoothed version of the true spectrum $S(f)$. 

It can be shown that if the spectrum $S(f)$ is slowly varying around a frequency $f$, then the tapered spectral estimates $\hatS_k(f)$ are approximately uncorrelated, and $\Var[\hatS_k(f)] \approx S(f)^2$. Hence, $\Var[\hatSmt_K(f)] \approx \tfrac{1}{K}S(f)^2$. Thus, Thomson's multitaper method produces a spectral estimate whose variance is a factor of $K \approx 2NW$ smaller than the variance of a single tapered periodogram. 

As we increase $W$, the width of the spectral window $\psi(f)$ increases, which causes the expectation of the multitaper spectral estimate to be further smoothed. However, increasing $W$ also allows us to increase the number of tapers $K \approx 2NW$, which reduces the variance of the multitaper spectral estimate. Intuitively, Thomson's multitaper method introduces a tradeoff between resolution and robustness. 

\subsection{Linear algebra view}
Here we provide an alternate perspective on Thomson's multitaper method which is based on linear algebra and subspace projections. Suppose that for each frequency $f \in \R$, we choose a low-dimensional subspace $\mathcal{S}_f \subset \C^N$, and form a spectral estimate by computing $\|\text{proj}_{\mathcal{S}_f}(\vx)\|_2^2$, i.e., the energy in the projection of $\vx$ onto the subspace $\mathcal{S}_f$. One simple choice is the one-dimensional subspace $\mathcal{S}_f = \text{span}\{\ve_f\}$ where $$\ve_f = \begin{bmatrix}1 & e^{j2\pi f \cdot 1} & e^{j2\pi f \cdot 2} & \cdots & e^{j2\pi f \cdot (N-1)}\end{bmatrix}^T$$ is a vector of equispaced samples from a complex sinusoid with frequency $f$. For this choice of $\mathcal{S}_f$, we have $$\left\|\text{proj}_{\mathcal{S}_f}(\vx)\right\|_2^2 = \dfrac{\left|\inner{\ve_f,\vx}\right|^2}{\|\ve_f\|_2^2} = \dfrac{1}{N}\left|\sum_{n = 0}^{N-1}\vx[n]e^{-j2\pi fn}\right|^2,$$ which is exactly the classic periodogram. 

We can also choose a low-dimensional subspace $\mathcal{S}_f$ which minimizes the average representation error of sinusoids $\ve_{f'}$ with frequency $f' \in [f-W,f+W]$ for some small $W > 0$, i.e., $$\minimize_{\substack{\mathcal{S}_f \subset \C^N \\ \dim(\mathcal{S}_f) = K}}\int_{f-W}^{f+W}\left\|\ve_{f'} - \text{proj}_{\mathcal{S}_f}(\ve_{f'})\right\|_2^2\,df',$$ where the dimension of the subspace $K$ is fixed. Using ideas from the Karhunen-Loeve (KL) transform \cite{Stark86}, it can be shown that the optimal $K$-dimensional subspace is the span of the top $K$ eigenvectors of the covariance matrix $$\mC_f := \dfrac{1}{2W}\int_{f-W}^{f+W}\ve_{f'}\ve_{f'}^*\,df'.$$ The entries of this covariance matrix are 
\begin{align*}
\mC_f[m,n] &= \dfrac{1}{2W}\int_{f-W}^{f+W}\ve_{f'}[m]\overline{\ve_{f'}[n]}\,df' 
\\
&= \dfrac{1}{2W}\int_{f-W}^{f+W}e^{j2\pi f'(m-n)}\,df' 
\\
&= \dfrac{\sin[2\pi W(m-n)]}{2\pi W(m-n)}e^{j2\pi f(m-n)}
\\
&= \dfrac{1}{2W}\ve_f[m]\mB[m,n]\overline{\ve_f[n]},
\end{align*} 
where again $\mB$ is the $N \times N$ prolate matrix. Hence, we can write $$\mC_f = \dfrac{1}{2W}\mE_f\mB\mE_f^*,$$ where $\mE_f = \diag(\ve_f) \in \C^{N \times N}$ is a unitary matrix which modulates vectors by pointwise multiplying them by the sinusoid $\ve_f$. Therefore, the eigenvectors of $\mC_f$ are the modulated Slepian basis vectors $\mE_f\vs_k$ for $k \in [N]$, and the corresponding eigenvalues are $\tfrac{\lambda_k}{2W}$. Hence, we can choose $\mathcal{S}_f = \text{span}\{\mE_f\vs_0,\ldots,\mE_f\vs_{K-1}\}$, i.e., the span of the first $K$ Slepian basis vectors modulated to the frequency $f$. Since $\vs_0,\ldots,\vs_{K-1}$ are orthonormal vectors, and $\mE_f$ is a unitary matrix, $\mE_f\vs_0,\ldots,\mE_f\vs_{K-1}$ are orthonormal vectors. Hence, $\text{proj}_{\mathcal{S}_f}(\vx) = \mE_f\mS_K\mS_K^*\mE_f^*\vx$ where $\mS_K = \begin{bmatrix}\vs_0 & \ldots & \vs_{K-1}\end{bmatrix}$, and thus,
\begin{align*}
\left\|\text{proj}_{\mathcal{S}_f}(\vx)\right\|_2^2 &= \left\|\mE_f\mS_K\mS_K^*\mE_f^*\vx\right\|_2^2
\\
&= \left\|\mS_K^*\mE_f^*\vx\right\|_2^2
\\
&= \sum_{k = 0}^{K-1}\left|(\mS_K^*\mE_f^*\vx)[k]\right|^2
\\
&= \sum_{k = 0}^{K-1}\left|\vs_k^*\mE_f\vx\right|^2
\\
&= \sum_{k = 0}^{K-1}\left|\sum_{n = 0}^{N-1}\vs_k[n]\vx[n]e^{-j2\pi fn}\right|^2.
\end{align*}
Up to a constant scale factor, this is precisely the multitaper spectral estimate. Hence, we can view the the multitaper spectral estimate $\hatSmt_K(f) = \tfrac{1}{K}\|\mS_K^*\mE_f^*\vx\|_2^2$ as the energy in $\vx$ after it is projected onto the $K$-dimensional subspace which best represents the collection of sinusoids $\{\ve_{f'} : f' \in [f-W,f+W]\}$. 

\subsection{Slepian basis eigenvalues}
\label{sec:SlepianBasisEigenvalues}

Before we proceed to our main results, we elaborate on the clustering behavior of the Slepian basis eigenvalues, as they are critical to our analysis in the rest of this paper. For any $\eps \in (0,\tfrac{1}{2})$, slightly fewer than $2NW$ eigenvalues lie in $[1-\eps,1)$, slightly fewer than $N-2NW$ eigenvalues lie in $(0,\eps]$, and very few eigenvalues lie in the so-called ``transition region'' $(\eps,1-\eps)$. In Figure 1, we demonstrate this phenomenon by plotting the first $1000$ Slepian basis eigenvalues for $N = 10000$ and $W = \tfrac{1}{100}$ (so $2NW = 200$). The first $194$ eigenvalues lie in $[0.999,1)$ and the last $9794$ eigenvalues lie in $(0,0.001]$. Only $12$ eigenvalues lie in $(0.001,0.999)$. 

\begin{figure}
\centering
\includegraphics[scale=0.35]{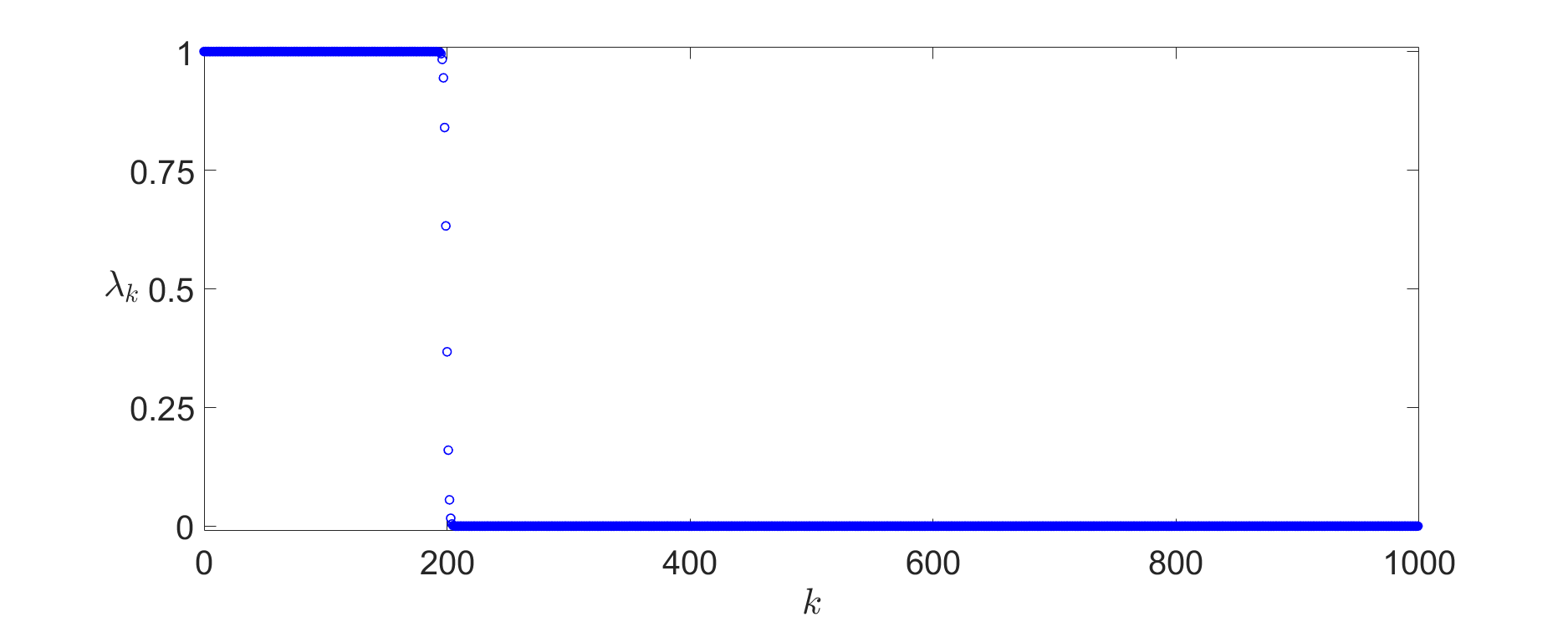}
\caption{A plot of the first $1000$ Slepian basis eigenvalues for $N = 10000$ and $W = \tfrac{1}{100}$. These eigenvalues satisfy $\lambda_{193} \approx 0.9997$ and $\lambda_{206} \approx 0.0003$. Only $12$ of the $10000$ Slepian basis eigenvalues lie in $(0.001,0.999)$.}
\label{fig:EigenvaluePlot}
\end{figure}

Slepian \cite{SlepianV} showed that for any fixed $W \in (0,\tfrac{1}{2})$ and $b \in \R$, $$\lambda_{\floor{2NW+(b/\pi)\log N}} \to \dfrac{1}{1+e^{b\pi}} \quad \text{as} \quad N \to \infty.$$ From this result, it is easy to show that for any fixed $W \in (0,\tfrac{1}{2})$ and $\eps \in (0,\tfrac{1}{2})$, $$\#\{k : \eps < \lambda_k < 1-\eps\} \sim \dfrac{2}{\pi^2}\log N \log\left(\dfrac{1}{\eps}-1\right) \quad \text{as} \quad N \to \infty.$$

In \cite{DPSSEigenvalues}, the authors of this paper derive a non-asymptotic bound on the number of Slepian basis eigenvalues in the transition region. For any $N \in \N$, $W \in (0,\tfrac{1}{2})$, and $\eps \in (0,\tfrac{1}{2})$, $$\#\{k : \eps < \lambda_k < 1-\eps\} \le \dfrac{2}{\pi^2}\log(100NW+25)\log\left(\dfrac{5}{\eps(1-\eps)}\right)+7.$$ This shows that the width of the transition region grows logarithmically with respect to both the time-bandwidth product $2NW$ and the tolerance parameter $\eps$. Note that when $\eps$ and $W$ are small, this result is a significant improvement over the prior non-asymptotic bounds in \cite{ZhuWakin17}, \cite{Boulsane20}, and \cite{KarnikFST}, which scale like $O(\tfrac{\log N}{\eps})$, $O(\tfrac{\log(NW)}{\eps})$, and $O(\log N \log \tfrac{1}{\eps})$ respectively. In Section~\ref{sec:FastAlgorithms}, we will exploit the fact that the number of Slepian basis eigenvalues in the transition region grows like $O(\log(NW)\log\tfrac{1}{\eps})$ to derive a fast algorithm for evaluating $\hatSmt_K(f)$ at a grid of evenly spaced frequencies. 

Furthermore, in \cite{DPSSEigenvalues}, the authors of this paper combine the above bound on the width of the transition region with the fact from \cite{ZhuWakin17} that $\lambda_{\floor{2NW}-1} \ge \tfrac{1}{2} \ge \lambda_{\ceil{2NW}}$ to obtain the following lower bound on the first $\approx 2NW$ eigenvalues $$\lambda_k \ge 1-10\exp\left(-\dfrac{\floor{2NW}-k-7}{\tfrac{2}{\pi^2}\log(100NW+25)}\right) \quad \text{for} \quad 0 \le k \le \floor{2NW}-1.$$ This shows that as $k$ decreases from $\approx 2NW$, the quantity $1-\lambda_k$ decays exponentially to $0$. In particular, this means that the first $2NW-O(\log(NW))$ Slepian basis tapers $\vs_k$ have spectral windows $\widetilde{\vs}_k(f)$ which have very little energy outside $[-W,W]$. However, the spectral windows of the Slepian basis tapers for $k$ near $2NW$ will have a significant amount of energy outside $[-W,W]$. In Section~\ref{sec:ParameterSelection}, we show that the statistical properties of the multitaper spectral estimate are significantly improved when using $K = 2NW-O(\log(NW))$ tapers instead of the traditional $K = \floor{2NW}-1$ tapers.


\section{Statistical Properties and Spectral Leakage}
\label{sec:ParameterSelection}

For a given vector of signal samples $\vx \in \C^N$, using Thomson's multitaper method for spectral estimation requires selecting two parameters: the half-bandwidth $W$ of the Slepian basis tapers and the number of tapers $K$ which are used in the multitaper spectral estimate. The selection of these parameters can greatly impact the accuracy of the multitaper spectral estimate.

In some applications, a relatively small number of samples $N$ are taken, and the desired frequency resolution for the spectral estimate is $O(\tfrac{1}{N})$, i.e., a small multiple of the fundamental Rayleigh resolution limit. In such cases, many practitioners \cite{Llinas99,Mitra99,Bond97,Ghil02,Mann96,Thomson90,Hansson11,Leonardo99} choose the half-bandwidth parameter $W$ such that $2NW$ is between $3$ and $10$, and then choose the number of Slepian basis tapers to be between $K = \floor{2NW}$ and $K = \floor{2NW}-2$. However, in applications where a large number of samples $N$ are taken, and some loss of resolution is acceptable, choosing a larger half-bandwidth parameter $W$ can result in a more accurate spectral estimate. Furthermore, if the power spectral density $S(f)$ has a high dynamic range (that is $\max_f S(f) \gg \min_f S(f)$), we aim to show that choosing $K = 2NW - O(\log(NW)\log\tfrac{1}{\delta})$ tapers for some small $\delta > 0$ (instead of $K = 2NW - O(1)$ tapers) can provide significantly better protection against spectral leakage. 

For all the theorems in this section, we assume that $\vx \in \C^N$ is a vector of samples from a complex Gaussian process whose power spectral density $S(f)$ is bounded and integrable. Note that the analogous results for a real Gaussian process would be similar, but slightly more complicated to state. To state our results, define $$M = \max_{f \in \R}S(f),$$ i.e., the global maximum of the power spectral density, and for each frequency $f \in \R$ we define: $$m_f = \min_{f' \in [f-W,f+W]} S(f'),$$ $$M_f = \max_{f' \in [f-W,f+W]} S(f'),$$ $$A_f = \dfrac{1}{2W}\int_{f-W}^{f+W}S(f')\,df',$$ $$R_f = \sqrt{\dfrac{1}{2W}\int_{f-W}^{f+W}S(f')^2\,df'},$$ i.e. the minimum, maximum, average, and root-mean-squared values of the power spectral density over the interval $[f-W,f+W]$. We also define the quantities $$\Sigma_K^{(1)} = \dfrac{1}{K}\sum_{k = 0}^{K-1}(1-\lambda_k)$$ $$\Sigma_K^{(2)} = \sqrt{\dfrac{1}{K}\sum_{k = 0}^{K-1}(1-\lambda_k)^2}.$$

Before we proceed to our results, we make note of the fact that $m_f$, $M_f$, $A_f$, and $R_f$ are all ``local'' properties of the power spectral density, i.e., they depend only on values of $S(f')$ for $f' \in [f-W,f+W]$, whereas $M$ is a ``global'' property. Note that if the power spectral density is ``slowly varying'' over the interval $[f-W,f+W]$, then $m_f \approx M_f \approx A_f \approx R_f \approx S(f)$. However, $M$ could be several orders of magnitude larger than $m_f$, $M_f$, $A_f$, and $R_f$ if the power spectral density has a high dynamic range. 

By using the bound on the Slepian basis eigenvalues from Section~\ref{sec:SlepianBasisEigenvalues}, we can obtain $\lambda_{K-1} \ge 1-\delta$ for some suitably small $\delta > 0$ by choosing the number of tapers to be $K = 2NW-O(\log(NW)\log\tfrac{1}{\delta})$. This choice of $K$ guarantees that $0 \le \Sigma^{(1)}_K \le \Sigma^{(2)}_K \le 1-\lambda_{K-1} \le \delta$, i.e., $\Sigma^{(1)}_K$, $\Sigma^{(2)}_K$, and $1-\lambda_{K-1}$ are all small, and thus, the global property $M = \max_f S(f)$ of the power spectral density will have a minimal impact on the non-asymptotic results below. In other words, using $K = 2NW-O(\log(NW)\log\tfrac{1}{\delta})$ tapers mitigates the ability for values of the power spectral density $S(f')$ at frequencies $f' \not\in [f-W,f+W]$ to impact the estimate $\hatSmt_K(f)$. However, if $K = 2NW - O(1)$ tapers are used, then the quantities $\Sigma^{(1)}_K$, $\Sigma^{(2)}_K$, and $1-\lambda_{K-1}$ could be large enough for the global property $M = \max_f S(f)$ of the power spectral density to significantly weaken the non-asymptotic results below. In other words, energy in the power spectral density $S(f')$ at frequencies $f' \not\in [f-W,f+W]$ can ``leak'' into the estimate $\hatSmt_K(f)$.

We begin with a bound on the bias of the multitaper spectral estimate under the additional assumption that the power spectral density is twice differentiable. Note this assumption is only used in Theorem~\ref{thm:Bias2Diff}. 
\begin{theorem}
\label{thm:Bias2Diff}
For any frequency $f \in \R$, if $S(f')$ is twice continuously differentiable in $[f-W,f+W]$, then the bias of the multitaper spectral estimate is bounded by $$\Bias\left[\hatSmt_K(f)\right] = \left|\E\hatSmt_K(f)-S(f)\right| \le \dfrac{M''_fNW^3}{3K} + (M+M_f)\Sigma_K^{(1)},$$ where $$M''_f = \max_{f' \in [f-W,f+W]}|S''(f')|.$$
\end{theorem}
If $K = 2NW-O(\log(NW)\log\tfrac{1}{\delta})$ tapers are used for some small $\delta > 0$, then this upper bound is slightly larger than $\tfrac{1}{6}M''_f W^2$, which is similar to the asymptotic results in \cite{Thomson82, Walden00, Lii08, Abreu17, Haley17} which state that the bias is roughly $\tfrac{1}{6}S''(f)W^2$. However, if $K = 2NW - O(1)$ tapers are used, the term $(M+M_f)\Sigma_K^{(1)}$ could dominate this bound, and the bias could be much larger than the asymptotic result.

If the power spectral density is not twice-differentiable, we can still obtain the following bound on the bias of the multitaper spectral estimate.
\begin{theorem}
\label{thm:BiasNotDiff}
For any frequency $f \in \R$, the bias of the multitaper spectral estimate is bounded by $$\Bias\left[\hatSmt_K(f)\right] = \left|\E\hatSmt_K(f)-S(f)\right| \le (M_f-m_f)(1-\Sigma_K^{(1)})+M\Sigma_K^{(1)}.$$
\end{theorem}
If $K = 2NW-O(\log(NW)\log\tfrac{1}{\delta})$ tapers are used for some small $\delta > 0$, then this upper bound is slightly larger than $M_f-m_f$. This guarantees the bias is small when the power spectral density is ``slowly varying'' over $[f-W,f+W]$. However, if $K = 2NW - O(1)$ tapers are used, the term $M\Sigma_K^{(1)}$ could dominate this bound, and the bias could be much larger than the asymptotic result.

Next, we state our bound on the variance of the multitaper spectral estimate. 
\begin{theorem}
\label{thm:Variance}
For any frequency $f \in \R$, the variance of the multitaper spectral estimate is bounded by $$\Var\left[\hatSmt_K(f)\right] \le \dfrac{1}{K}\left(R_f\sqrt{\dfrac{2NW}{K}} + M\Sigma^{(2)}_K\right)^2.$$
\end{theorem}
If $K = 2NW-O(\log(NW)\log\tfrac{1}{\delta})$ tapers are used for some small $\delta > 0$, then this upper bound is slightly larger than $\tfrac{1}{K}R_f^2$, which is similar to the asymptotic results in \cite{Thomson82,Walden00,Lii08,Abreu17,Haley17} which state that the variance is roughly $\tfrac{1}{K}S(f)^2$. However, if $K = 2NW - O(1)$ tapers are used, the term $M\Sigma_K^{(2)}$ could dominate this bound, and the variance could be much larger than the asymptotic result.  

We also note that if the frequencies $f_1,f_2$ are more than $2W$ apart, then the multitaper spectral estimates at those frequencies have a very low covariance.  
\begin{theorem}
\label{thm:Covariance}
For any frequencies $f_1,f_2 \in \R$ such that $2W < |f_1-f_2| < 1-2W$, the covariance of the multitaper spectral estimates at those frequencies is bounded by $$0 \le \Cov\left[\hatSmt_K(f_1),\hatSmt_K(f_2)\right] \le \left((R_{f_1}+R_{f_2})\sqrt{\dfrac{2NW}{K}\Sigma^{(1)}_K} + M\Sigma^{(1)}_K\right)^2.$$
\end{theorem}
If $K = 2NW-O(\log(NW)\log\tfrac{1}{\delta})$ tapers are used for some small $\delta > 0$, then the covariance is guaranteed to be small. However, if $K = 2NW - O(1)$ tapers are used, the upper bound is no longer guaranteed to be small, and the covariance could be large. 

Finally, we also provide a concentration result for the multitaper spectral estimate. 
\begin{theorem}
\label{thm:Concentration}
For any frequency $f \in \R$, the multitaper spectral estimate satisfies the concentration inequalities $$\P\left\{\hatSmt_K(f) \ge \beta\E\hatSmt_K(f)\right\} \le \beta^{-1}\e^{-\kappa_f(\beta-1-\ln \beta)} \quad \text{for} \quad \beta > 1,$$ and $$\P\left\{\hatSmt_K(f) \le \beta\E\hatSmt_K(f)\right\} \le \e^{-\kappa_f(\beta-1-\ln \beta)} \quad \text{for} \quad 0 < \beta < 1,$$ where the frequency dependent constant $\kappa_f$ satisfies $$\kappa_f \ge \dfrac{K\left(1-\Sigma^{(1)}_K\right)M_f - 2NW(M_f-A_f)}{M_f+(M-M_f)(1-\lambda_{K-1})}.$$
\end{theorem}
We note that these are identical to the concentration bounds for a chi-squared random variable with $2\kappa_f$ degrees of freedom. If $K = 2NW-O(\log(NW)\log\tfrac{1}{\delta})$ tapers are used for some small $\delta > 0$ and the power spectral density is ``slowly varying'' over $[f-W,f+W]$, then this lower bound on $\kappa_f$ is slightly less than $K$. Hence, $\hatSmt_K(f)$ has a concentration behavior that is similar to a chi-squared random variable with $2K$ degrees of freedom, as the asymptotic results in \cite{Thomson82,Percival93} suggest. However, if $K = 2NW - O(1)$ tapers are used, then $\kappa_f$ could be much smaller, and thus, the multitaper spectral estimate would have significantly worse concentration about its mean. 

The proofs of Theorems~\ref{thm:Bias2Diff}-\ref{thm:Concentration} are given in Appendix~\ref{sec:ParameterSelectionProofs}. In Section~\ref{sec:Simulations}, we perform simulations demonstrating that using $K = 2NW - O(1)$ tapers results in a multitaper spectral estimate that is vulnerable to spectral leakage, whereas using $K = 2NW-O(\log(NW)\log\tfrac{1}{\delta})$ tapers for a suitably small $\delta > 0$ significantly reduces the impact of spectral leakage on the multitaper spectral estimate.

\section{Fast Algorithms}
\label{sec:FastAlgorithms}
Given a vector of $N$ samples $\vx \in \C^N$, evaluating the multitaper spectral estimate $\hatSmt_K(f)$ at a grid of $L$ evenly spaced frequencies $f \in [L]/L$ (where we assume $L \ge N$) can be done in $O(KL\log L)$ operations and using $O(KL)$ memory via $K$ length-$L$ fast Fourier transforms (FFTs). In applications where the number of samples $N$ is small, the number of tapers $K$ used is usually a small constant, and so, the computational requirements are a small constant factor more than that of an FFT. However in many applications, using a large number of tapers is desirable, but the computational requirements make this impractical. As mentioned in Section~\ref{sec:Intro}, if the power spectrum $S(f)$ is twice-differentiable, then the MSE of the multitaper spectral estimate is minimized when the bandwidth parameter is $W = O(N^{-1/5})$ and $K = O(N^{4/5})$ tapers are used \cite{Abreu17}. For medium to large scale problems, precomputing and storing $O(N^{4/5})$ tapers and/or performing $O(N^{4/5})$ FFTs may be impractical.

In this section, we present an $\eps$-approximation $\tildeSmt_K(f)$ to the multitaper spectral estimate $\hatSmt_K(f)$ which requires $O(L\log L\log(NW)\log\tfrac{1}{\eps})$ operations and $O(L\log(NW)\log\tfrac{1}{\eps})$ memory. This is faster than the exact multitaper spectral estimation provided the number of tapers satisfies $K \gtrsim \log(NW)\log\tfrac{1}{\eps}$.

To construct this approximation, we first fix a tolerance parameter $\eps \in (0,\tfrac{1}{2})$, and suppose that the number of tapers, $K$, is chosen such that $\lambda_{K-1} \ge \tfrac{1}{2}$ and $\lambda_K \le 1-\eps$. Note that this is a very mild assumption as it only forces $K$ to be slightly less than $2NW$. Next, we partition the indices $[N]$ into four sets:
\begin{align*}
\setI_1 &= \{k \in [K] : \lambda_k \ge 1-\eps\}
\\
\setI_2 &= \{k \in [K] : \eps < \lambda_k < 1-\eps\}
\\
\setI_3 &= \{k \in [N] \setminus [K] : \eps < \lambda_k < 1-\eps\}
\\
\setI_4 &= \{k \in [N] \setminus [K] : \lambda_k \le \eps\}
\end{align*}
and define the approximate estimator $$\tildeSmt_K(f) := \dfrac{1}{K}\Psi(f) + \dfrac{1}{K}\sum_{k \in \setI_2}(1-\lambda_k)\hatS_k(f) - \dfrac{1}{K}\sum_{k \in \setI_3}\lambda_k\hatS_k(f),$$ where $$\Psi(f) := \sum_{k = 0}^{N-1}\lambda_k\hatS_k(f).$$

Both $\hatSmt_K(f)$ and $\tildeSmt_K(f)$ are weighted sums of the single taper estimates $\hatS_k(f)$ for $k \in [N]$. Additionally, it can be shown that the weights are similar, i.e., the first $K$ weights are exactly or approximately $\tfrac{1}{K}$, and the last $N-K$ weights are exactly or approximately $0$. Hence, it is reasonable to expect that $\tildeSmt_K(f) \approx \hatSmt_K(f)$. The following theorem shows that is indeed the case. 

\begin{theorem}
\label{thm:ApproxMultitaper}
The approximate multitaper spectral estimate $\tildeSmt_K(f)$ defined above satisfies $$\left|\tildeSmt_K(f) - \hatSmt_K(f)\right| \le \dfrac{\eps}{K}\|\vx\|_2^2 \quad \text{for all} \quad f \in \R.$$ 
\end{theorem}

Furthermore, we show in Lemma~\ref{lem:FastPsi} that $\Psi(f) = \vx^*\mE_f\mB\mE_f^*\vx$. This formula doesn't involve any of the Slepian tapers. By exploiting the fact that the prolate matrix $\mB$ is Toeplitz, we also show in Lemma~\ref{lem:FastPsi} that if $L \ge 2N$, then $$\begin{bmatrix}\Psi(\tfrac{0}{L}) & \Psi(\tfrac{1}{L}) & \cdots & \Psi(\tfrac{L-2}{L}) & \Psi(\tfrac{L-1}{L}) \end{bmatrix}^T = \mF^{-1}\left(\vb \circ \mF\left|\mF\mZ\vx\right|^2\right),$$ where $\mZ \in \R^{L \times N}$ is a matrix which zero-pads length-$N$ vectors to length-$L$, $\mF \in \C^{L \times L}$ is a length-$L$ FFT matrix, $\vb \in \R^{L}$ is the first column of the matrix formed by extending the prolate matrix $\mB$ to an $L \times L$ circulant matrix, $| \cdot |^2$ denotes the pointwise magnitude-squared, and $\circ$ denotes a pointwise multiplication. Hence, $\Psi(f)$ can be evaluated at a grid of $L$ evenly spaced frequencies $f \in [L]/L$ in $O(L \log L)$ operations via three length-$L$ FFTs/inverse FFTs. Evaluating the other $\#(\setI_2\cup\setI_3) = O(\log(NW)\log\tfrac{1}{\eps})$ terms in the expression for $\tildeSmt_K(f)$ at the $L$ grid frequencies can be done in $O(L\log L\log(NW)\log\tfrac{1}{\eps})$ operations via $\#(\setI_2\cup\setI_3)$ length-$L$ FFTs. Using these results, we establish the following theorem which states how quickly $\tildeSmt_K(f)$ can be evaluated at the grid frequencies.

\begin{theorem}
\label{thm:FastMultitaper}
For any vector of samples $\vx \in \C^N$ and any number of grid frequencies $L \ge N$, the approximate multitaper spectral estimate $\tildeSmt_K(f)$ can be evaluated at the $L$ grid frequencies $f \in [L]/L$ in $O(L\log L \log(NW)\log\tfrac{1}{\eps})$ operations and using $O(L\log(NW)\log\tfrac{1}{\eps})$ memory.
\end{theorem}

Note if $N \le L < 2N$, we can apply the method briefly described above to evaluate $\Psi(f)$ at $f \in [2L]/2L$, and then downsample the result. The proofs of Theorems~\ref{thm:ApproxMultitaper} and \ref{thm:FastMultitaper} are given in Appendix~\ref{sec:FastAlgorithmsProofs}. In Section~\ref{sec:Simulations}, we perform simulations comparing the time needed to evaluate $\hatSmt_K(f)$ and $\tildeSmt_K(f)$ at a grid of frequencies.  

\section{Simulations}
\label{sec:Simulations}

In this section, we show simulations to demonstrate three observations. (1) Using $K = 2NW-O(\log(NW))$ tapers instead of the traditional choice of $K = \floor{2NW}-1$ tapers significantly reduce the effects of spectral leakage. (2) Using a larger bandwidth $W$, and thus, more tapers can produce a more robust spectral estimate. (3) As the number of samples $N$ and the number of tapers $K$ grows, our approximation $\tildeSmt_K(f)$ becomes significantly faster to use than the exact multitaper spectral estimate $\hatSmt_K(f)$. 

\subsection{Spectral leakage}
\label{sec:SpectralWindow}
First, we demonstrate that choosing $K = 2NW-O(\log(NW))$ tapers instead of the traditional choice of $K = \floor{2NW}-1$ tapers significantly reduces the effects of spectral leakage. We fix a signal length of $N = 2000$, a bandwidth parameter of $W = \tfrac{1}{100}$ (so $2NW = 40$) and consider four choices for the number of tapers: $K = 39$, $36$, $32$, and $29$. Note that $K = 39 = \floor{2NW}-1$ is the traditional choice as to how many tapers to use, while $36$, $32$, and $29$ are the largest values of $K$ such that $\lambda_{K-1}$ is at least $1-10^{-3}$, $1-10^{-6}$, and $1-10^{-9}$ respectively.

In Figure~\ref{fig:SpectralWindow}, we show three plots of the spectral window $$\psi(f) = \dfrac{1}{K}\sum_{k = 0}^{K-1}\left|\sum_{n = 0}^{N-1}\vs_k[n]e^{-j 2\pi fn}\right|^2$$ of the multitaper spectral estimate for each of those values of $K$. At the top of Figure~\ref{fig:SpectralWindow}, we plot $\psi(f)$ over the entire range $[-\tfrac{1}{2},\tfrac{1}{2}]$ using a logarithmic scale. The lines outside appear thick due to the highly oscillatory behavior of $\psi(f)$. This can be better seen in the middle of Figure~\ref{fig:SpectralWindow}, where we plot $\psi(f)$ over $[-2W,2W]$ using a logarithmic scale. The behavior of $\psi(f)$ inside $[-W,W]$ can be better seen at the bottom of Figure~\ref{fig:SpectralWindow}, where we plot $\psi(f)$ over $[-2W,2W]$ using a linear scale.

All four spectral windows have similar behavior in that $\psi(f)$ is small outside $[-W,W]$ and large near $0$. However, outside of $[-W,W]$ the spectral windows using $K = \floor{2NW}-O(\log(NW))$ tapers are multiple orders of magnitude smaller than the spectral window using $K = \floor{2NW}-1$ tapers. Hence, the amount of spectral leakage can be reduced by multiple orders of magnitude by trimming the number of tapers used from $K = \floor{2NW}-1$ to $K = 2NW-O(\log(NW))$. 

\begin{figure}
\centering
\includegraphics[scale=0.35]{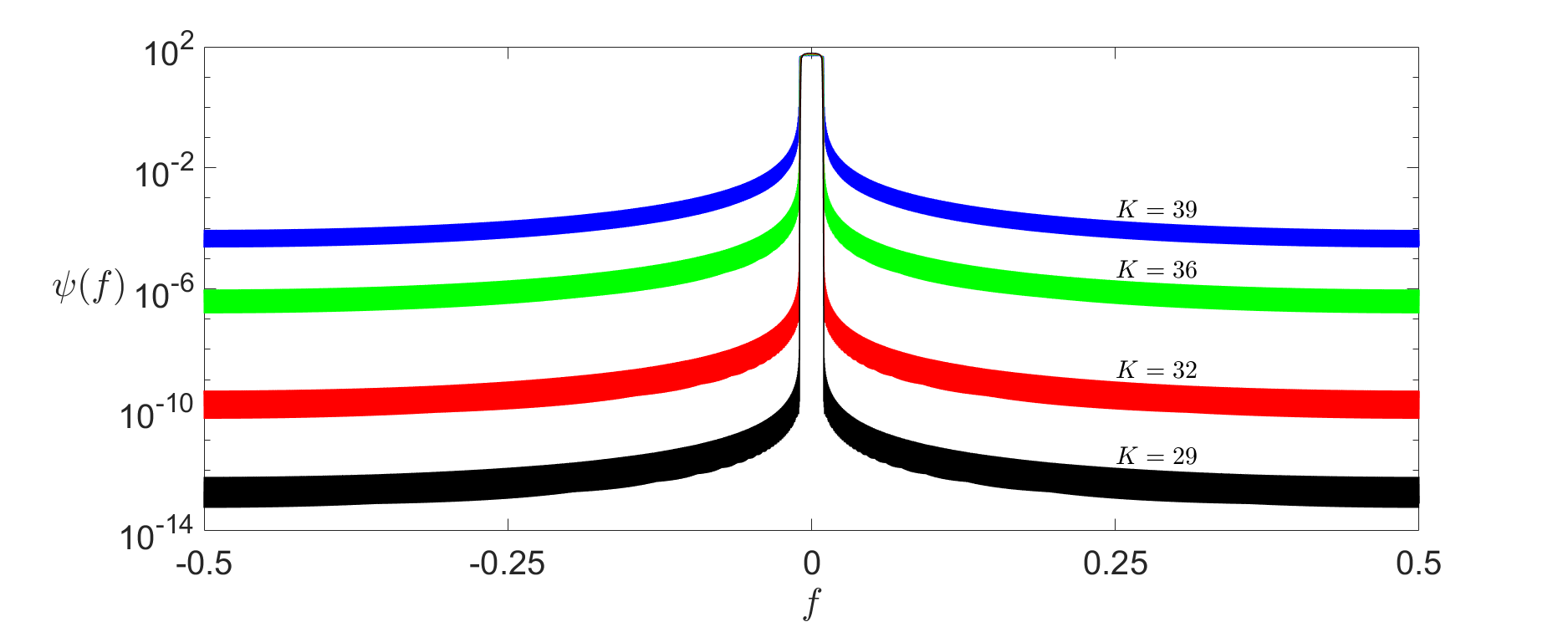}

\includegraphics[scale=0.35]{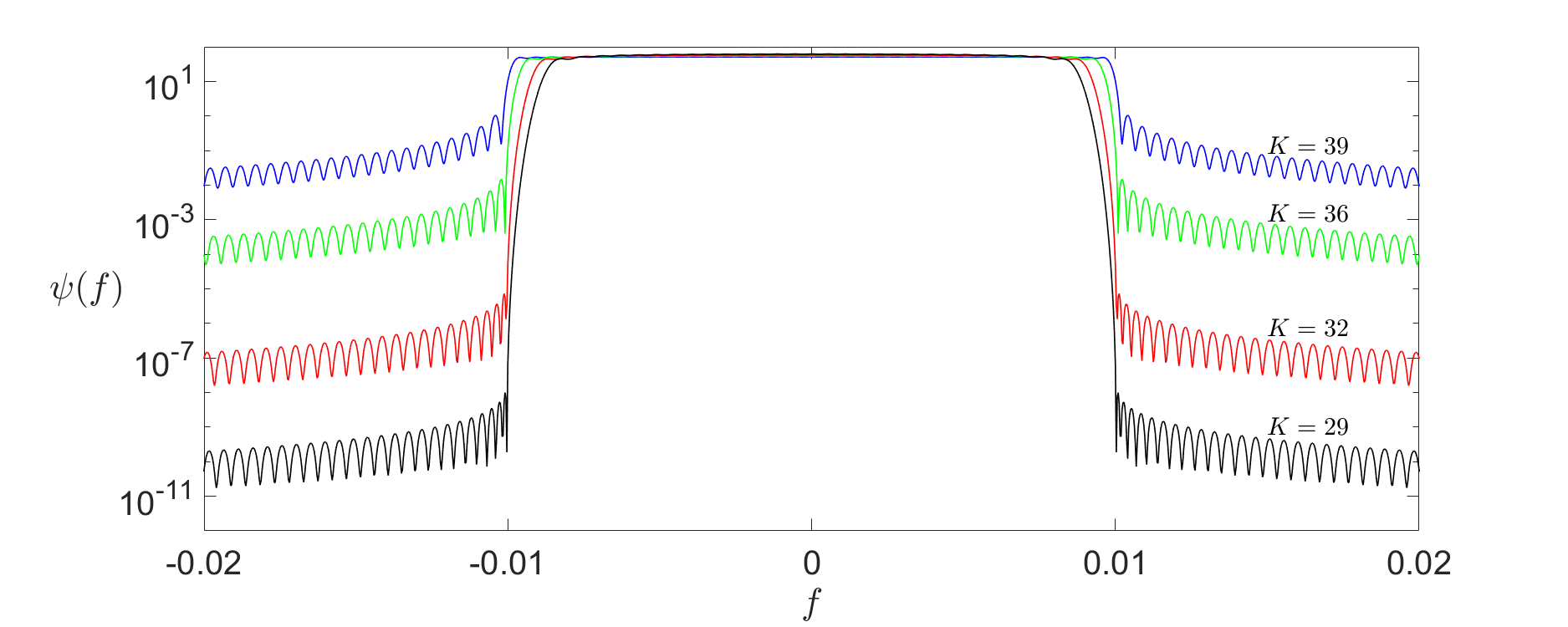}

\includegraphics[scale=0.35]{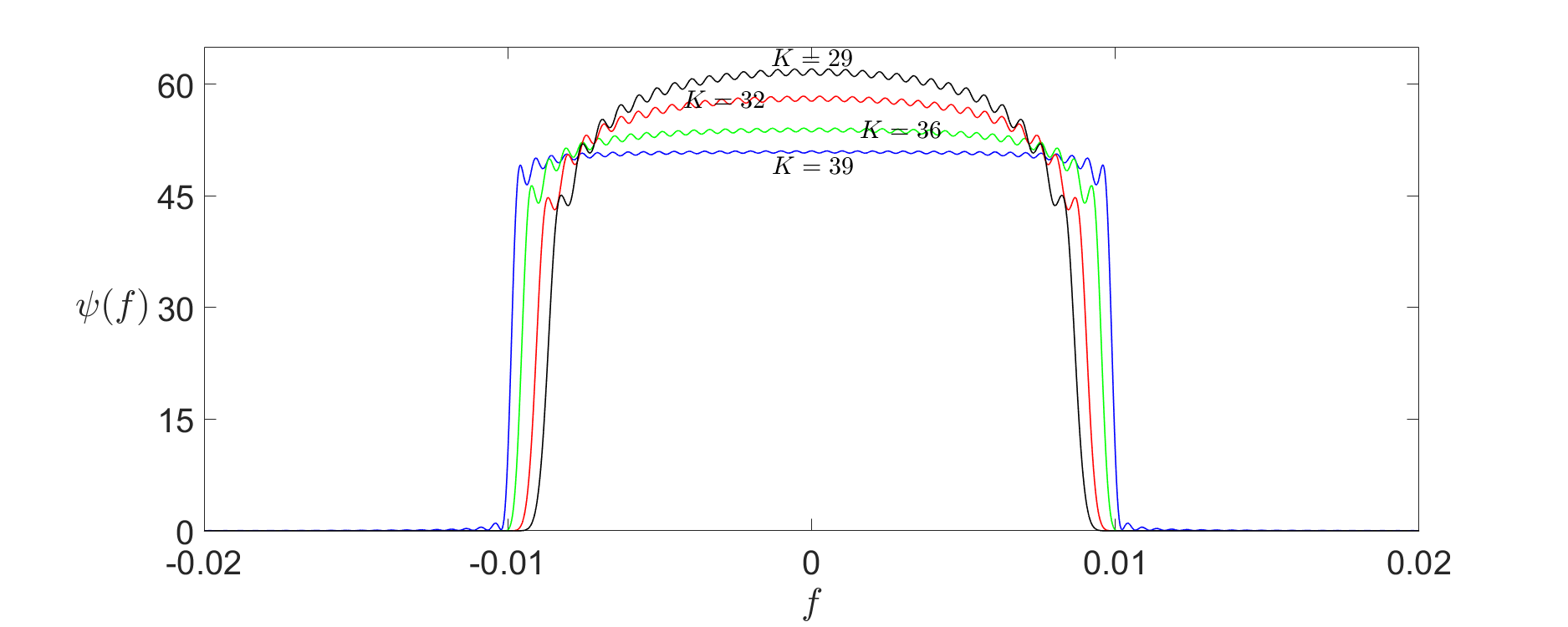}

\caption{Plots of the spectral windows $\psi(f)$ for $N = 2000$, $W = \tfrac{1}{100}$, and $K = 39$, $36$, $32$, and $29$ tapers. (Top) A logarithmic scale plot over $f \in [-\tfrac{1}{2},\tfrac{1}{2}]$. (Middle) A logarithmic scale plot over $f \in [-2W,2W]$. (Bottom) A linear scale plot over $f \in [-2W,2W]$.}
\label{fig:SpectralWindow}
\end{figure}

We further demonstrate the importance of using $K = 2NW-O(\log(NW))$ tapers to reduce spectral leakage by showing a signal detection example. We generate a vector $\vx \in \C^N$ of $N = 2000$ samples of a Gaussian random process with a power spectral density function of $$S(f) = \begin{cases}10^3 & \text{if} \ f \in [0.18,0.22] \\ 10^9 & \text{if} \ f \in [0.28,0.32] \\ 10^2 & \text{if} \ f \in [0.38,0.42] \\ 10^1 & \text{if} \ f \in [0.78,0.82]  \\ 10^0 & \text{else}\end{cases}.$$ This simulates an antenna receiving signals from four narrowband sources with some background noise. Note that one source is significantly stronger than the other three sources. In Figure~\ref{fig:MultibandSpectrum}, we plot: 
\begin{itemize}
\item the periodogram of $\vx$, 

\item the multitaper spectral estimate of $\vx$ with $W = \tfrac{1}{100}$ and $K = \floor{2NW}-1 = 39$ tapers, 

\item the multitaper spectral estimate of $\vx$ with $W = \tfrac{1}{100}$ and $K = 29$ tapers (chosen so $\lambda_{K-1} \ge 1-10^{-9}$).
\end{itemize}

\begin{figure}
\centering
\includegraphics[scale=0.35]{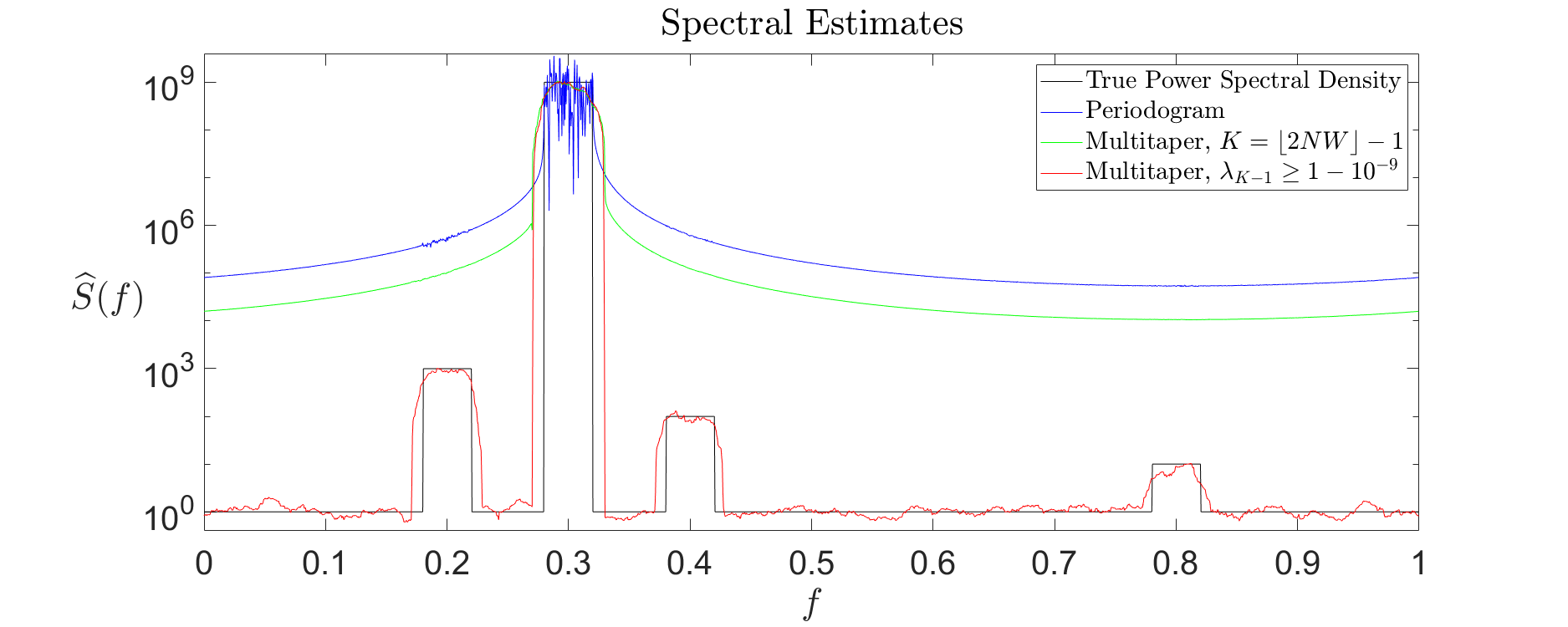}

\caption{Plots of the true power spectral density, the periodogram, the multitaper spectral estimate with $W = \tfrac{1}{100}$ and $K = \floor{2NW}-1 = 39$, and the multitaper spectral estimate with $W = \tfrac{1}{100}$ and $K = 29$ tapers (chosen so $\lambda_{K-1} \ge 1-10^{-9}$).}
\label{fig:MultibandSpectrum}
\end{figure}

We note that all three spectral estimates yield large values in the frequency band $[0.28,0.32]$. However, in the periodogram and the multitaper spectral estimate with $K = 39$ tapers, the energy in the frequency band $[0.28,0.32]$ ``leaks'' into the frequency bands occupied by the smaller three sources. As a result, the smaller three sources are hard to detect using the periodogram or the multitaper spectral estimate with $K = 39$ tapers. However, all four sources are clearly visible when looking at the multitaper spectral estimate with $K = 29$ tapers. For frequencies $f$ not within $W$ of the edges of the frequency band, the multitaper spectral estimate is within a small constant factor of the true power spectral density. 

\subsection{Comparison of spectral estimation methods}
Next, we demonstrate a few key points about selecting the bandwidth parameter $W$ and the number of tapers $K$ used in the multitaper spectral estimate. We compare the following eight methods to estimate the power spectrum of a Gaussian random process from a vector $\vx \in \C^N$ of $N = 2^{18} = 262144$ samples: 
\begin{enumerate}
\item The classic periodogram

\item A tapered periodogram using a single DPSS taper $\vs_0$ with $2NW = 8$ (chosen such that $\lambda_0 \ge 1-10^{-9}$)

\item The exact multitaper spectral estimate $\hatSmt_K(f)$ with a small bandwidth parameter of $W = 1.25 \times 10^{-4}$ and $K = \floor{2NW}-1 = 64$ tapers

\item The exact multitaper spectral estimate $\hatSmt_K(f)$ with a small bandwidth parameter of $W = 1.25 \times 10^{-4}$ and $K = 53$ tapers (chosen such that $\lambda_{K-1} \ge 1-10^{-9} > \lambda_K$)

\item The approximate multitaper spectral estimate $\tildeSmt_K(f)$ with a larger bandwidth parameter of $W = 2.0 \times 10^{-3}$, $K = \floor{2NW}-1 = 1047$ tapers, and a tolerance parameter of $\eps = 10^{-9}$

\item The approximate multitaper spectral estimate $\tildeSmt_K(f)$ with a larger bandwidth parameter of $W = 2.0 \times 10^{-3}$, $K = 1031$ tapers (chosen such that $\lambda_{K-1} \ge 1-10^{-9} > \lambda_K$), and a tolerance parameter of $\eps = 10^{-9}$

\item The exact multitaper spectral estimate with the adaptive weighting scheme suggested by Thomson\cite{Thomson82} with a small bandwidth parameter of $W = 1.25 \times 10^{-4}$ and $K = \floor{2NW}-1 = 64$ tapers

\item The exact multitaper spectral estimate with the adaptive weighting scheme suggested by Thomson\cite{Thomson82} with a larger bandwidth parameter of $W = 2.0 \times 10^{-3}$ and $K = \floor{2NW}-1 = 1047$ tapers
\end{enumerate}

The adaptive weighting scheme computes the single taper periodograms $\hatS_k(f)$ for $k \in [K]$, and then forms a weighted estimate $$\hatSad_K(f) = \dfrac{\sum_{k = 0}^{K-1}\alpha_k(f)\hatS_k(f)}{\sum_{k = 0}^{K-1}\alpha_k(f)}$$ where the frequency dependent weights $\alpha_k(f)$ satisfy $$\alpha_k(f) = \dfrac{\lambda_k\hatSad_K(f)^2}{\left(\lambda_k\hatSad_K(f)+(1-\lambda_k)\sigma^2\right)^2}$$ where $\sigma^2 = \tfrac{1}{N}\|\vx\|_2^2$. Of course, solving for the weights directly is difficult, so this method requires initializing the weights and alternating between updating the estimate $\hatSad_K(f)$ and updating the weights $\alpha_k(f)$. This weighting procedure is designed to keep all the weights large at frequencies where $S(f)$ is large and reduce the weights of the last few tapers at frequencies where $S(f)$ is small. Effectively, this allows the spectrum to be estimated with more tapers at frequencies where $S(f)$ is large while simultaneously reducing the spectral leakage from the last few tapers at frequencies where $S(f)$ is small. The cost is the increased computation time due to setting the weights iteratively. For more details regarding this adaptive scheme, see \cite{Haykin09}. 

In Figure~\ref{fig:SpectralEstimates}, we plot the power spectrum and the eight estimates for a single realization $\vx \in \C^N$ of the Gaussian random process. Additionally, for $1000$ realizations $\vx_i \in \C^N$, $i = 1,\ldots,1000$ of the Gaussian random process, we compute a spectral estimate using each of the above eight methods. In Figure~\ref{fig:EmpiricalMLD}, we plot the empirical mean logarithmic deviation in dB, i.e., $$\dfrac{1}{1000}\sum_{i = 1}^{1000}\left|10\log_{10}\dfrac{\hatS[\vx_i](f)}{S(f)}\right|$$ for each of the eight methods. In Table~\ref{tab:SpectralEstimationResults}, we list the average time needed to precompute the DPSS tapers, the average time to compute the spectral estimate after the tapers are computed, and the average of the empirical mean logarithmic deviation in dB across the frequency spectrum.

We make the following observations:

\begin{itemize}
\item The periodogram and the single taper periodogram (methods 1 and 2) are too noisy to be useful spectral estimates.

\item Methods 3, 4, and 7 yield a noticeably noisier spectral estimate than methods 5, 6, and 8. This is due to the fact that methods 5, 6, and 8 use a larger bandwidth parameter and more tapers.

\item The spectral estimates obtained with methods 1, 3 and 5 suffer from spectral leakage, i.e., the error is large at frequencies $f$ where $S(f)$ is small, as can be seen in Figure~\ref{fig:EmpiricalMLD}. This is due to the fact that they use $K = \floor{2NW}-1$ tapers, and thus, include tapered periodograms $\hatS_k(f)$ for which $\lambda_k$ is not very close to $1$.

\item Methods 4 and 6 avoid using tapered periodograms $\hatS_k(f)$ for which $\lambda_k < 1-10^{-9}$ and methods 7 and 8 use these tapered periodograms but assign a low weight to them at frequencies where $S(f)$ is small. Hence, methods 4, 6, 7, and 8 are able to mitigate the spectral leakage phenomenon.

\item Methods 5 and 6 are slightly faster than methods 3 and 4 due to the fact that our approximate multitaper spectral estimate only needs to compute $\#\{k : \eps < \lambda_k < 1-\eps\} = 36$ tapers and $36$ tapered periodograms. 

\item Method 7 takes noticeably longer than methods 3 and 4, and method 8 takes considerably longer than methods 5 and 6. This is because the iterative method for computing the adaptive weights requires many iterations to converge when the underlying power spectral density has a high dynamic range.

\item Methods 6 and 8 exhibit very similar performance. This is to be expected, as using a weighted average of $1047$ tapered periodograms is similar to using the unweighted average of the first $1031$ tapered periodograms. The empirical mean logarithmic deviation is larger at frequencies where $S(f)$ is rapidly varying and smaller at frequencies where $S(f)$ is slowly varying. This is to be expected as the local bias (caused due to the smoothing effect of the tapers) dominates the variance at these frequencies.   
\end{itemize}

\begin{figure}%
   \centering
  \includegraphics[width = 0.49\textwidth]{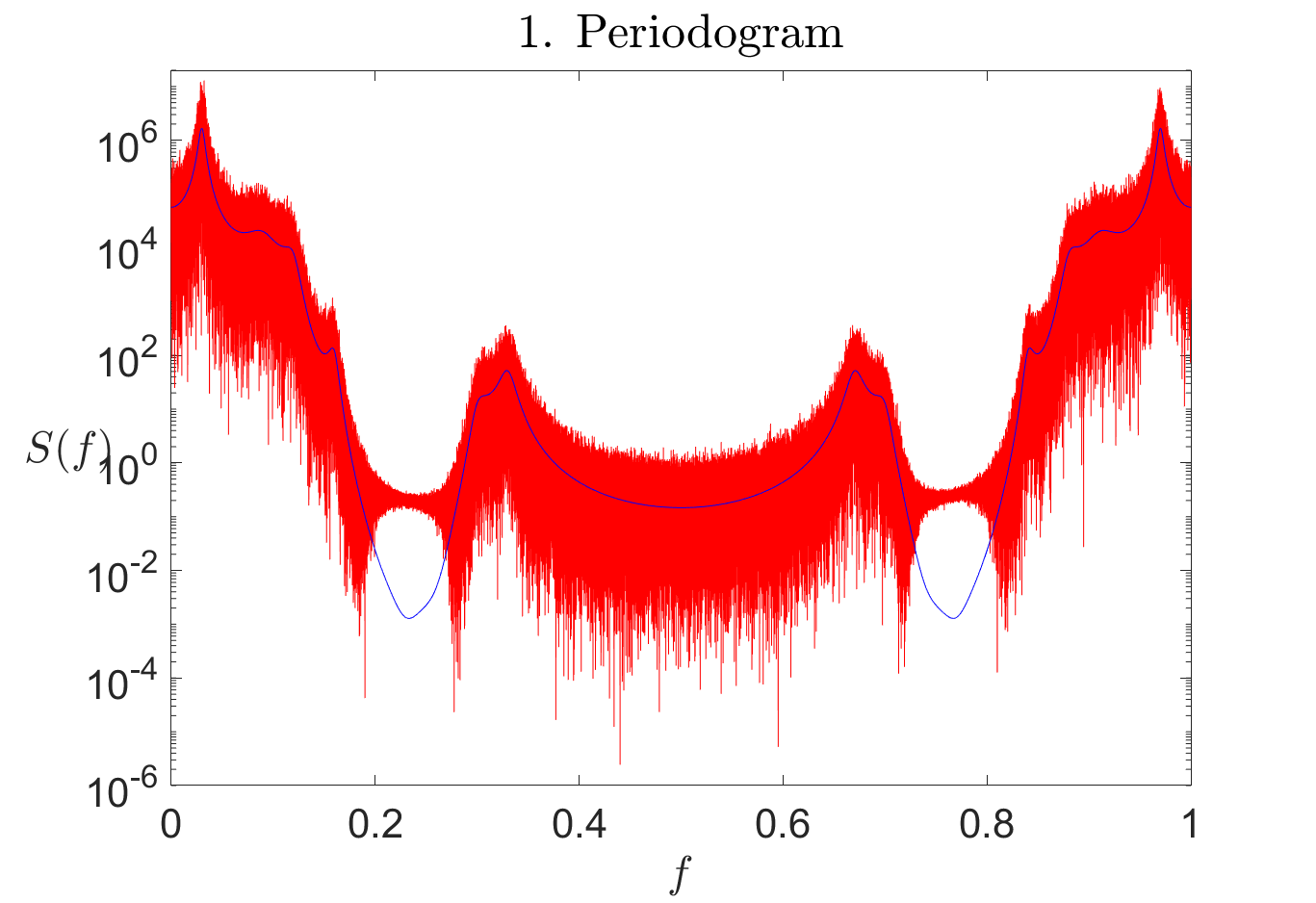} \includegraphics[width = 0.49\textwidth]{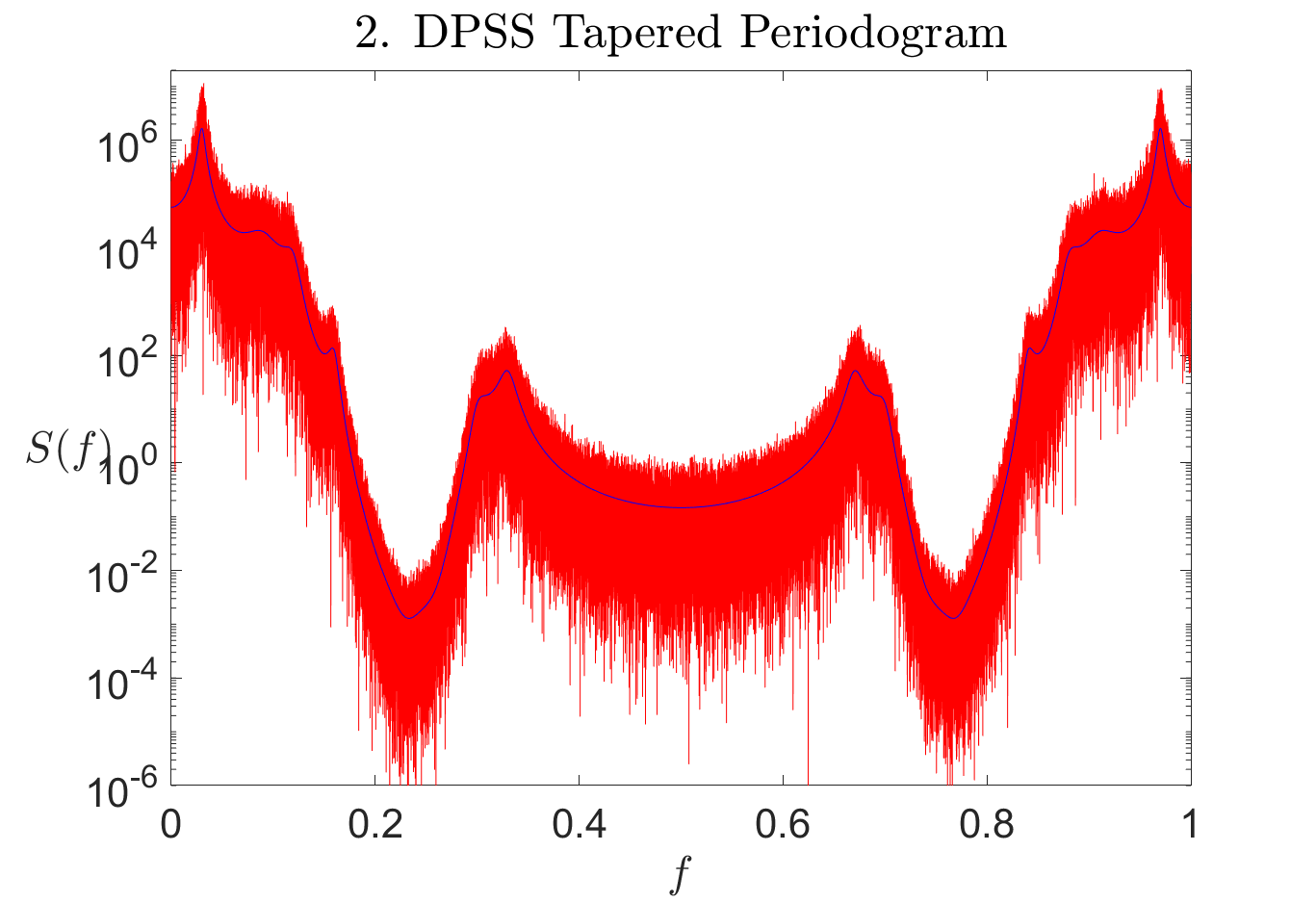}
  
  \includegraphics[width = 0.49\textwidth]{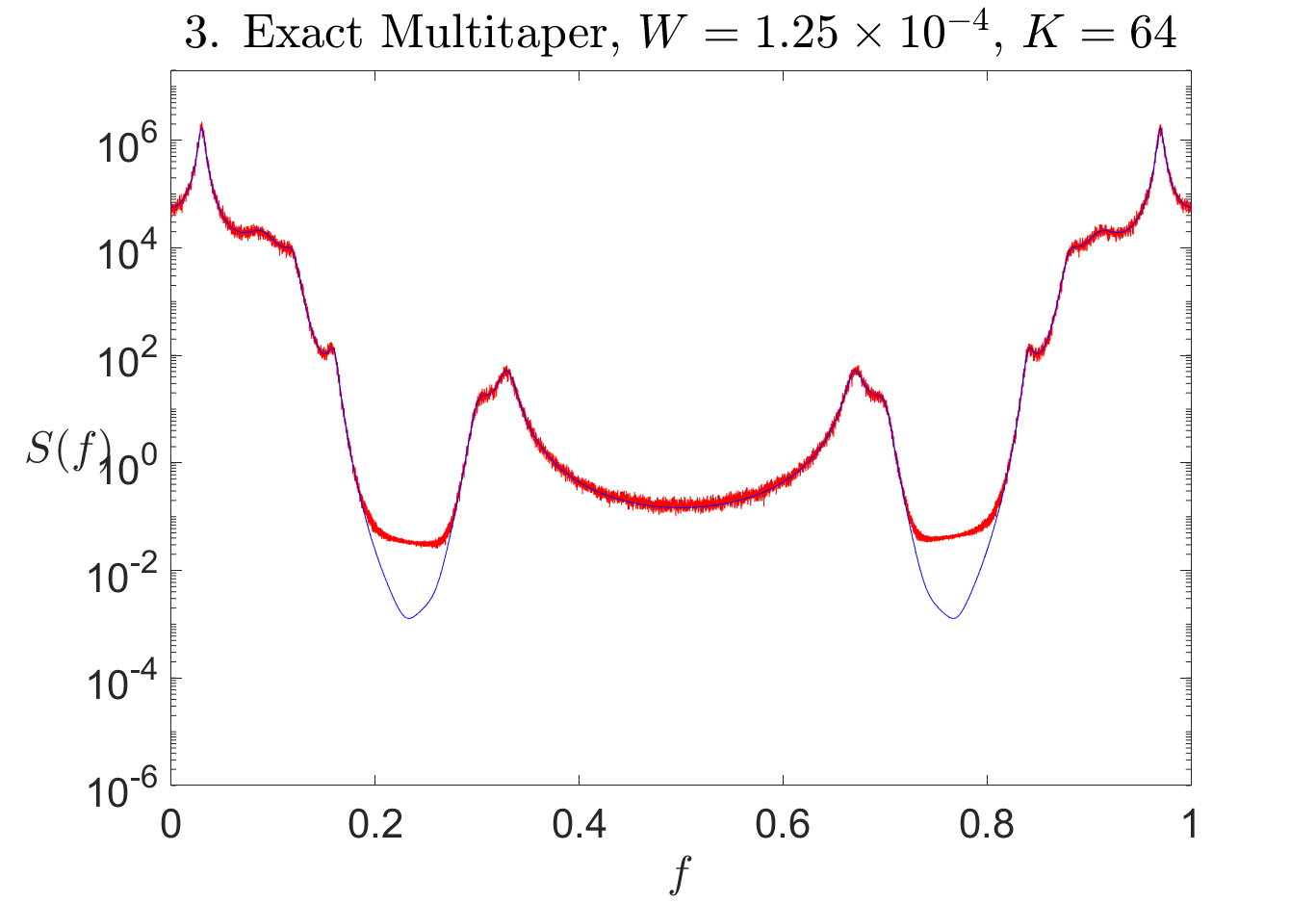} \includegraphics[width = 0.49\textwidth]{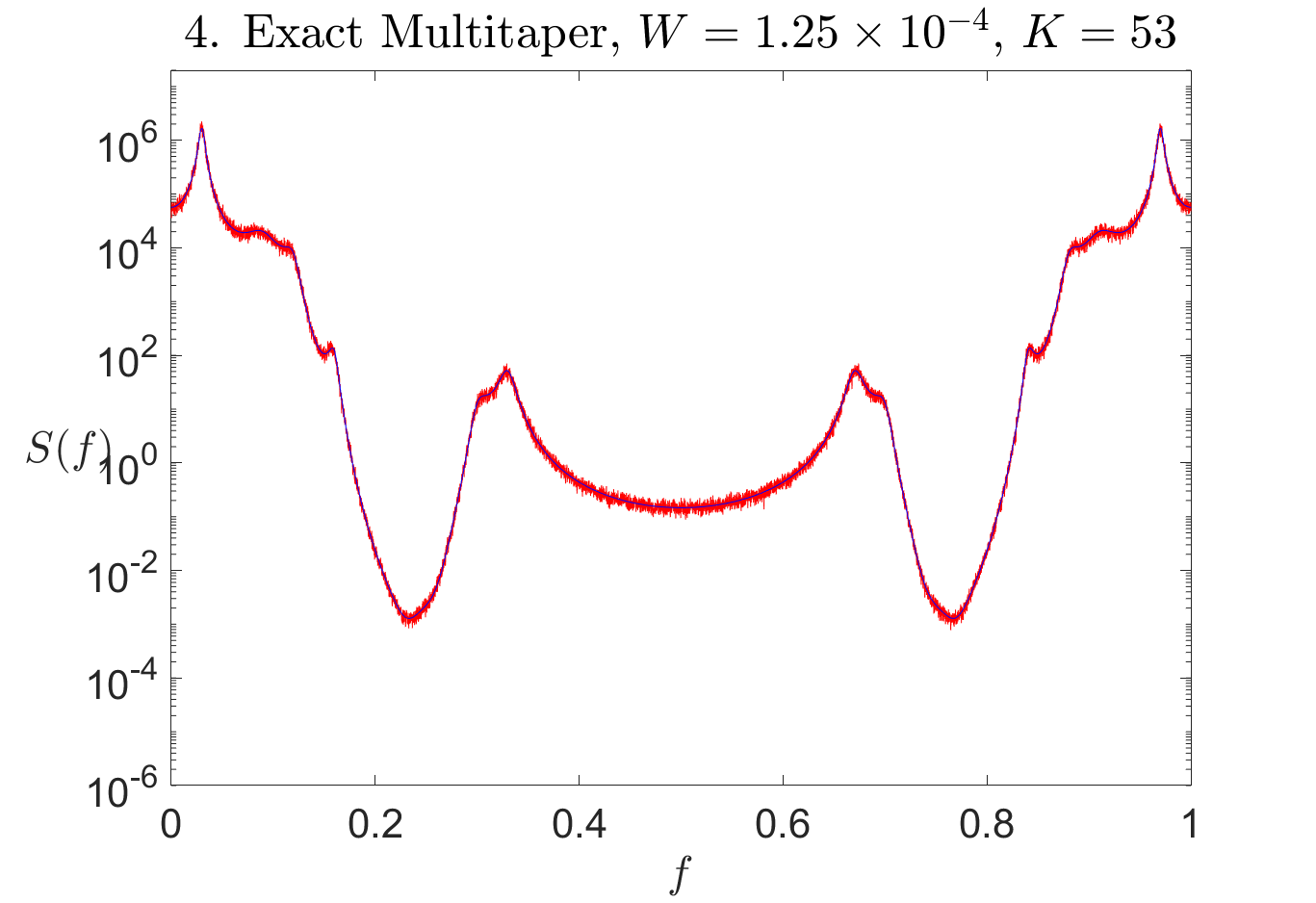}
  
  \includegraphics[width = 0.49\textwidth]{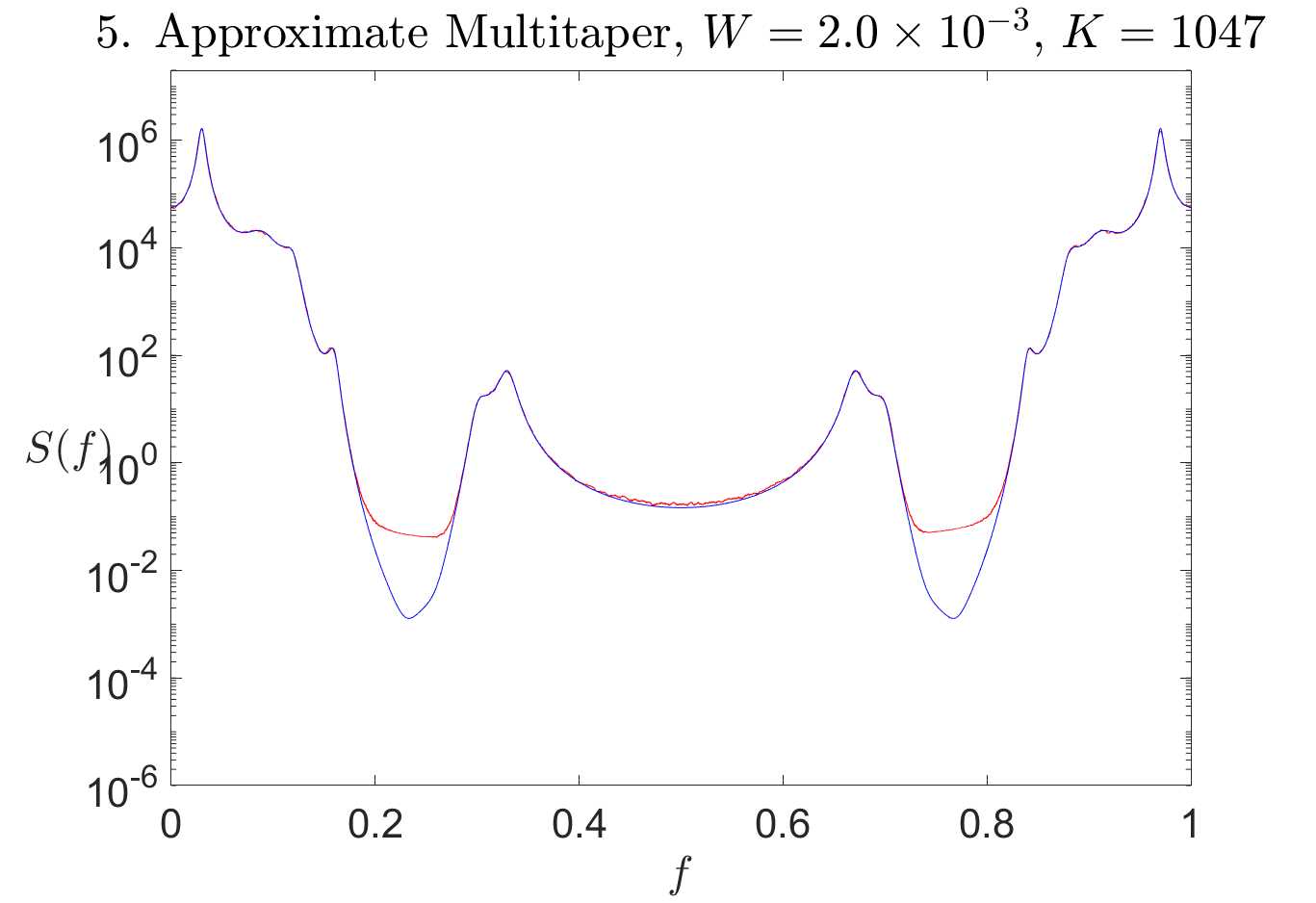} \includegraphics[width = 0.49\textwidth]{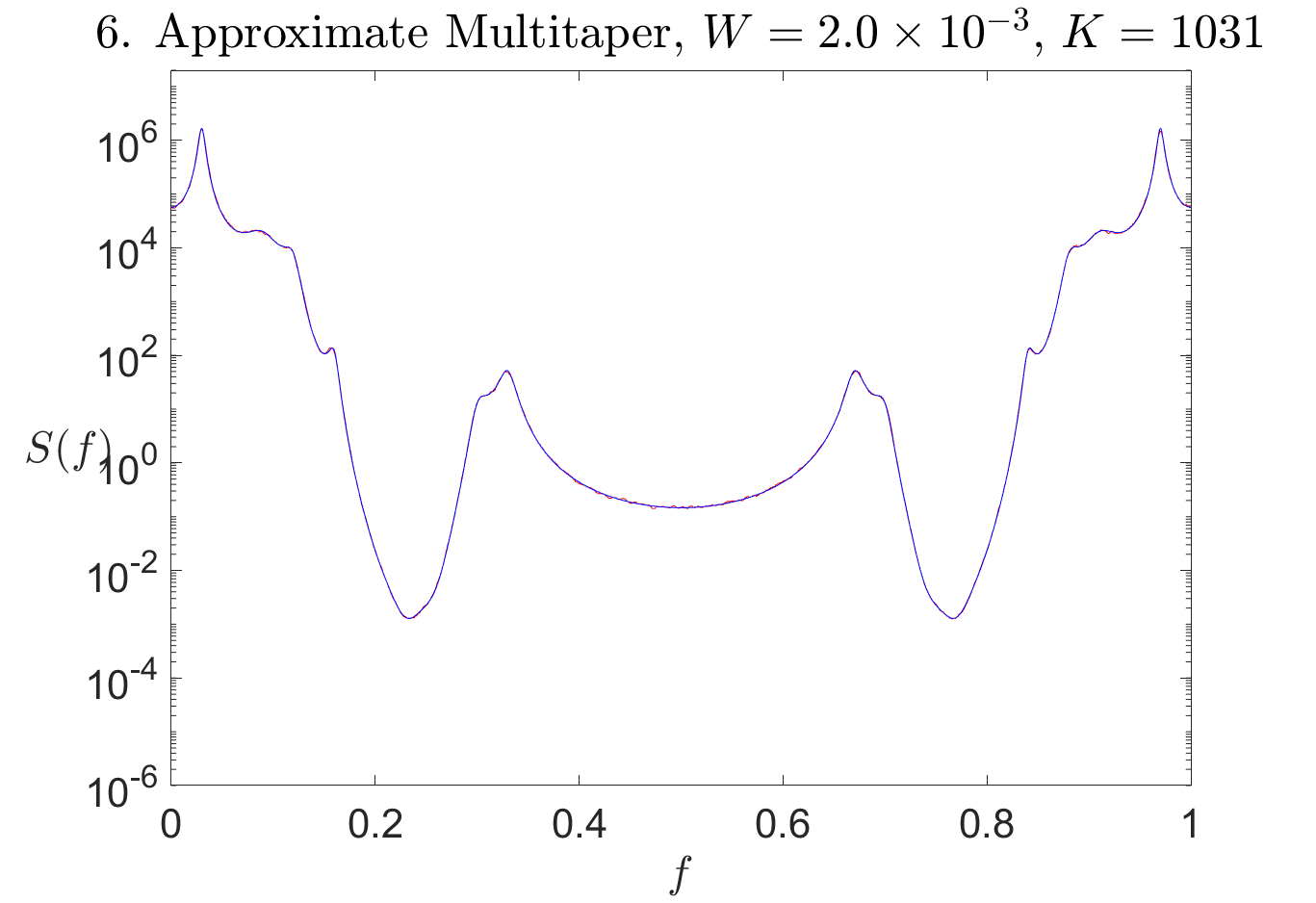}
    
  \includegraphics[width = 0.49\textwidth]{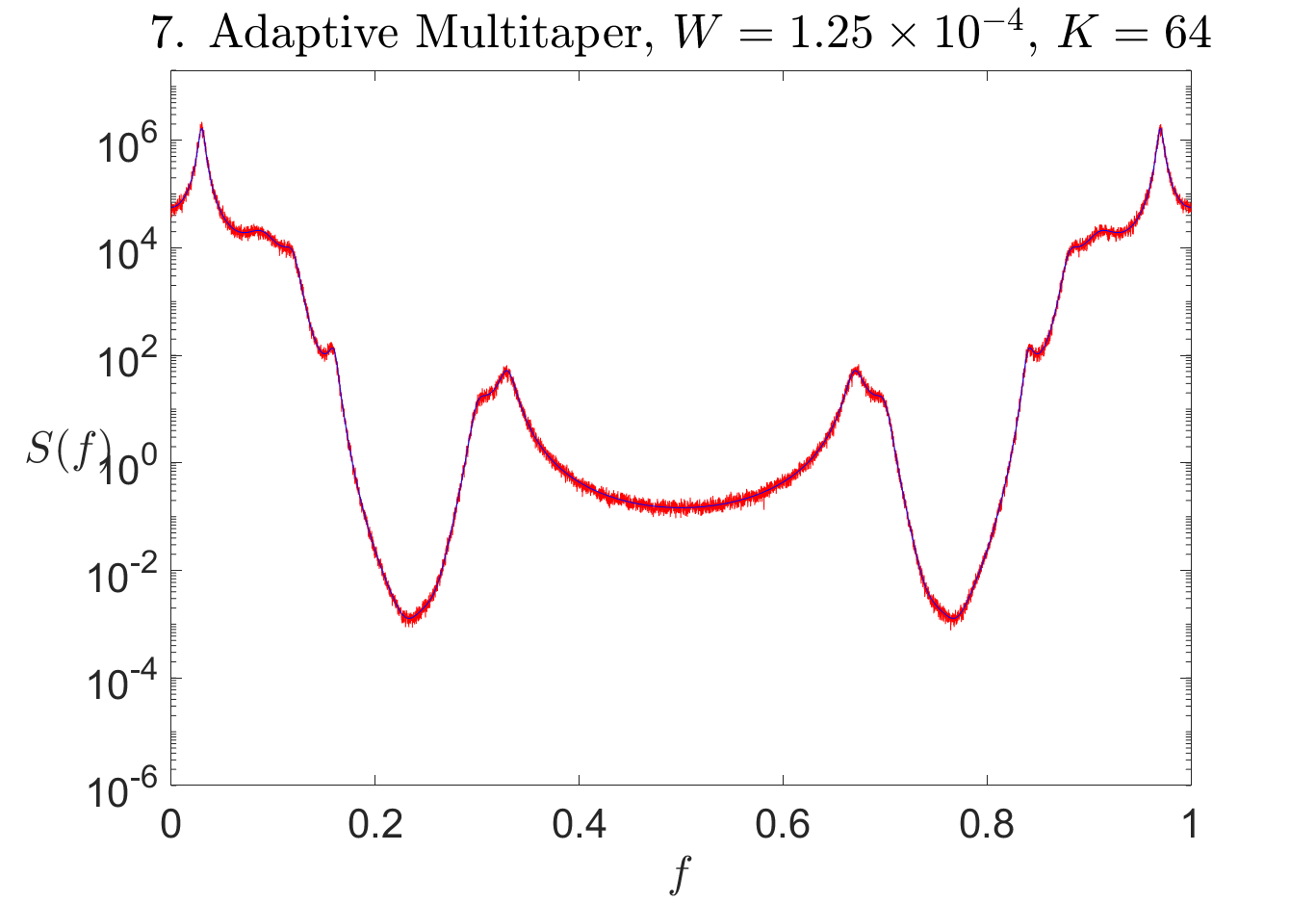} \includegraphics[width = 0.49\textwidth]{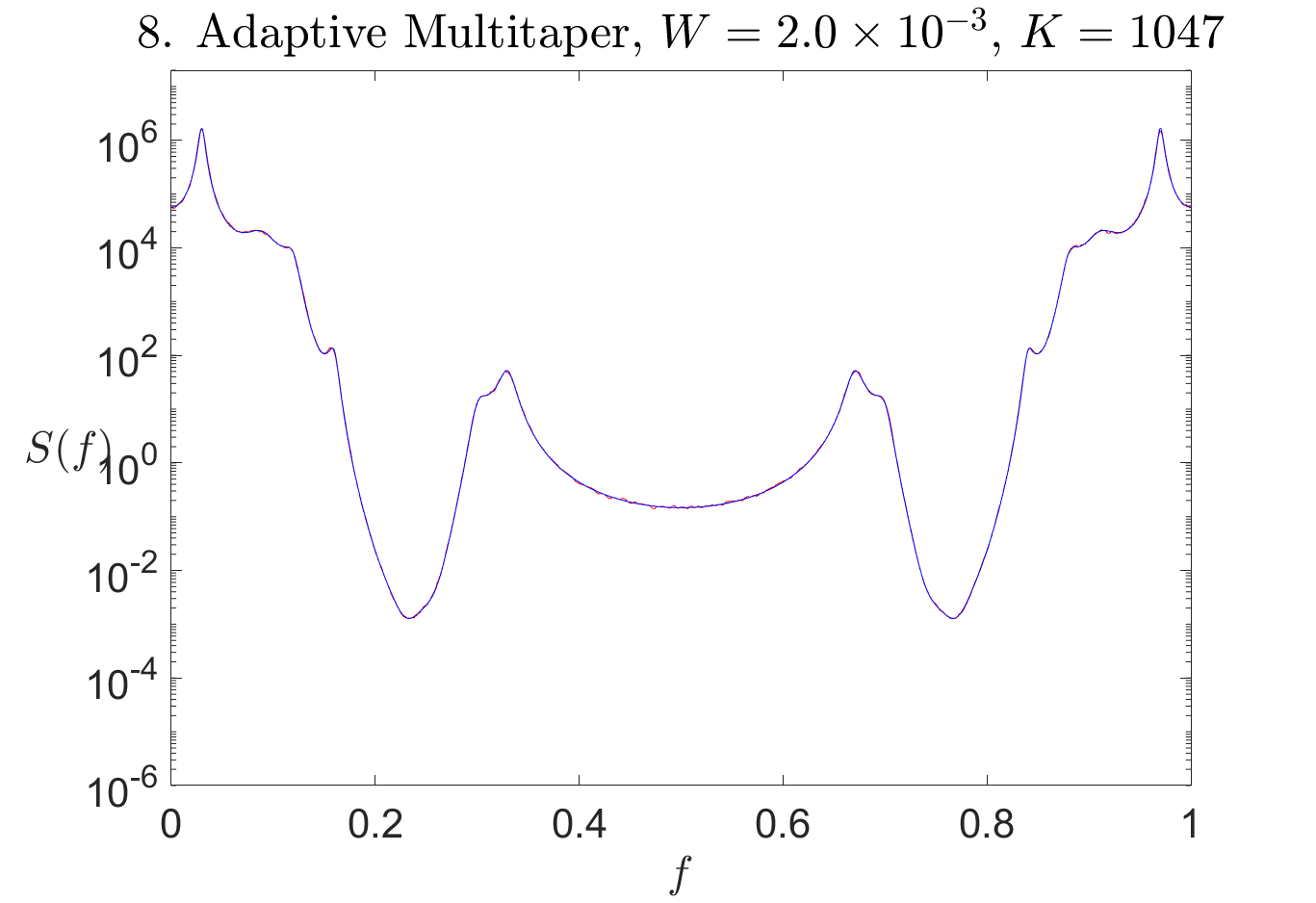}

   \caption{Plots of the true spectrum (blue) and the spectral estimates (red) using each of the eight methods.}
   \label{fig:SpectralEstimates}
\end{figure}

\begin{figure}%
   \centering
  \includegraphics[width = 0.49\textwidth]{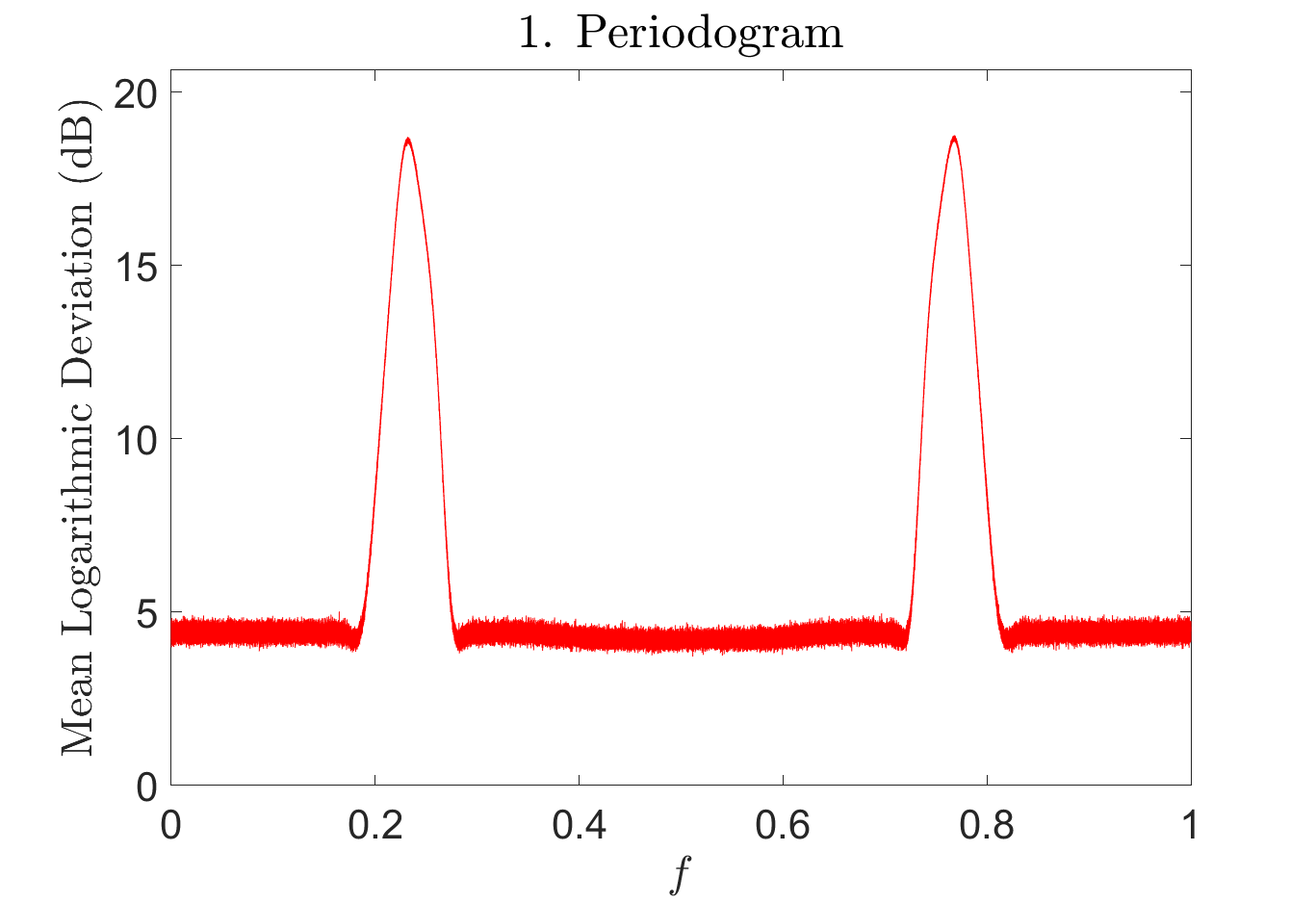} \includegraphics[width = 0.49\textwidth]{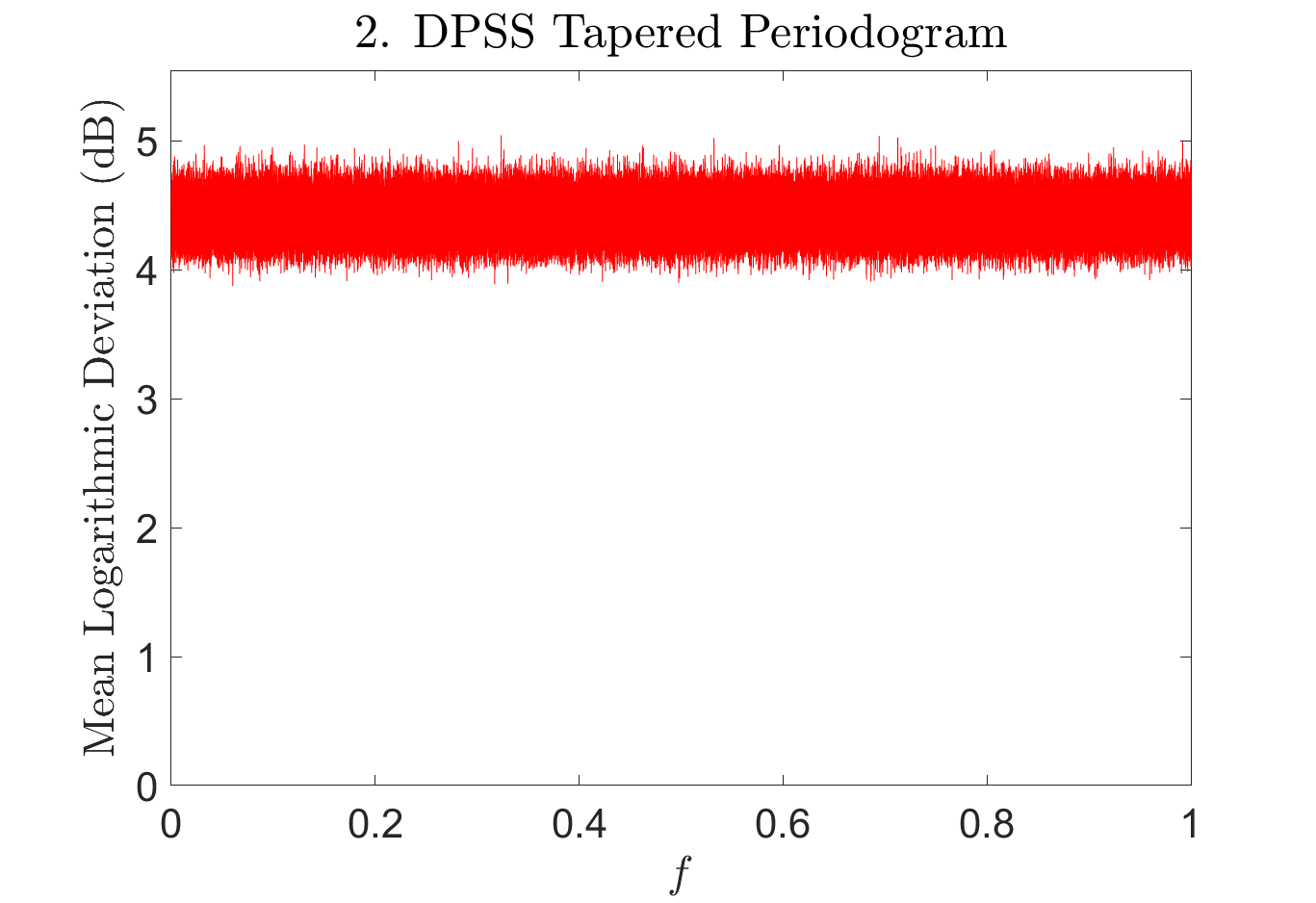}
  
  \includegraphics[width = 0.49\textwidth]{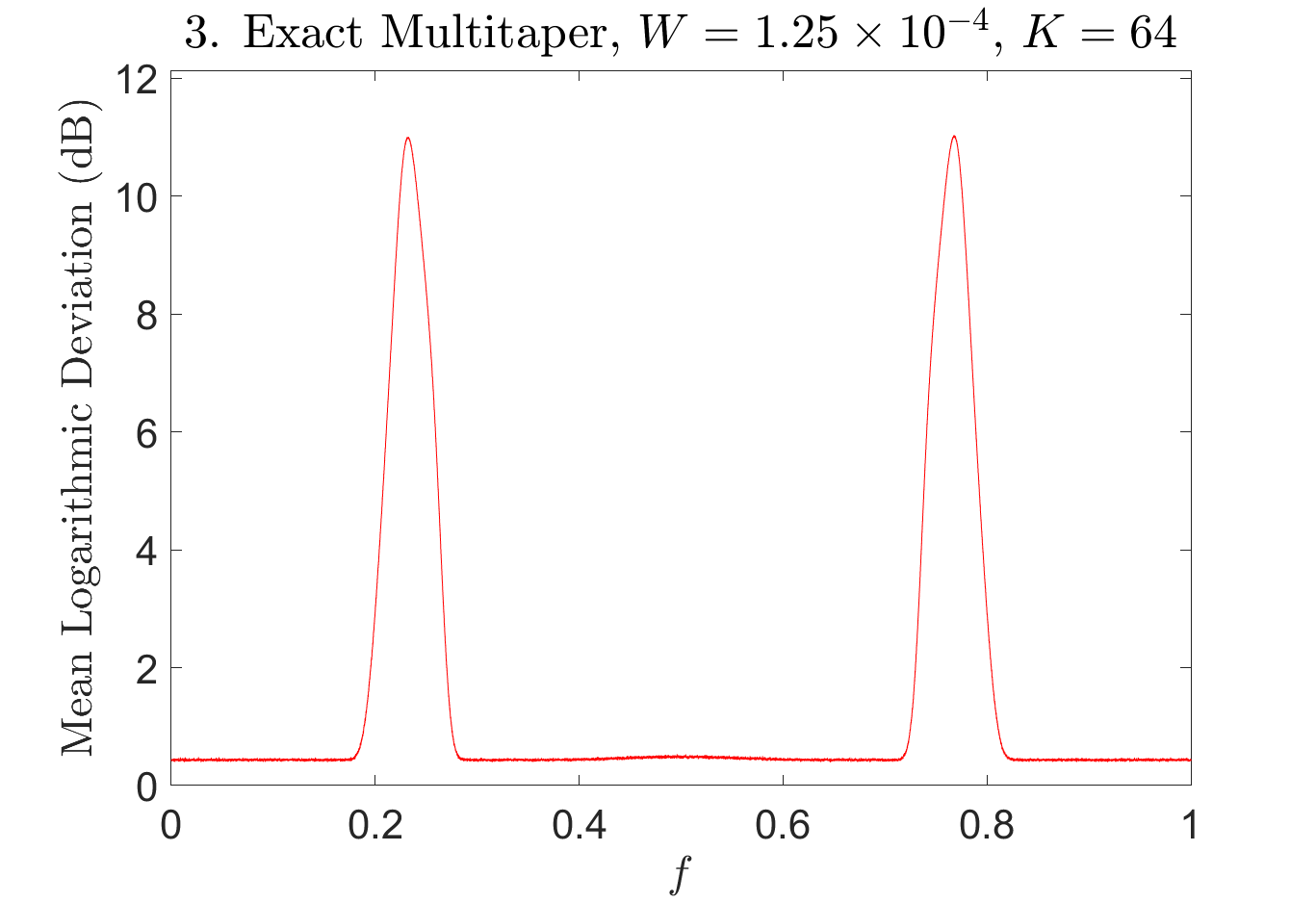} \includegraphics[width = 0.49\textwidth]{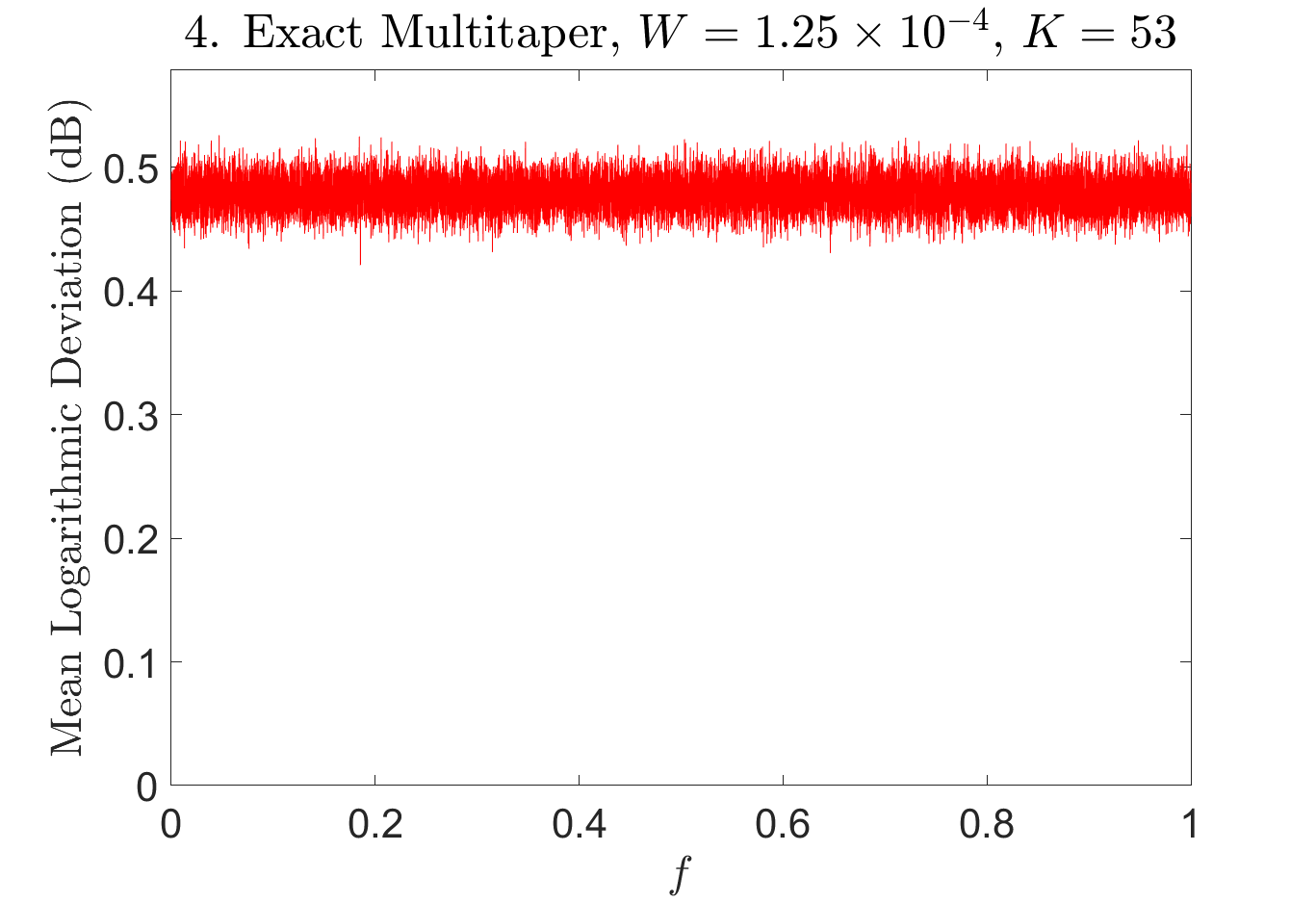}
  
  \includegraphics[width = 0.49\textwidth]{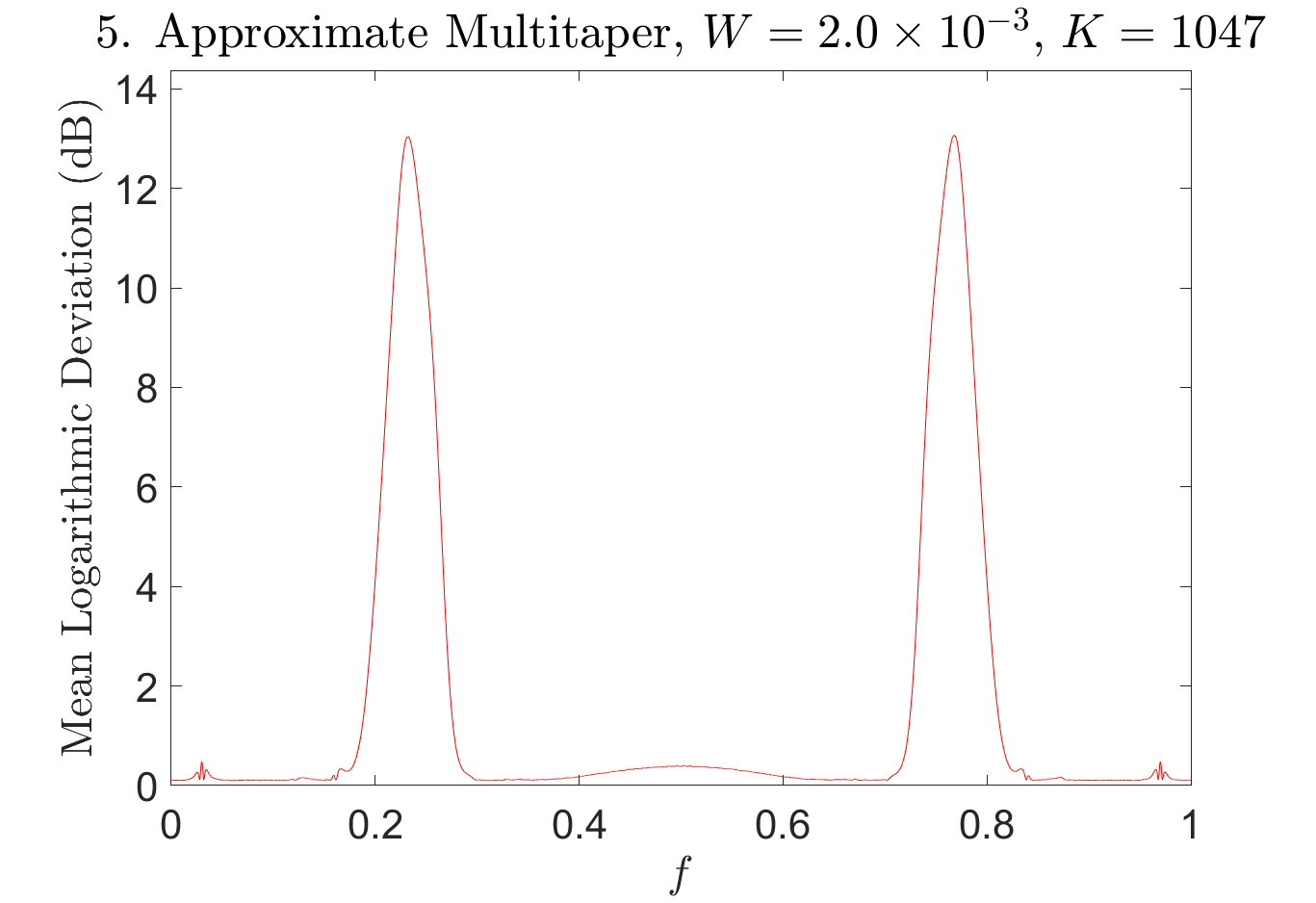} \includegraphics[width = 0.49\textwidth]{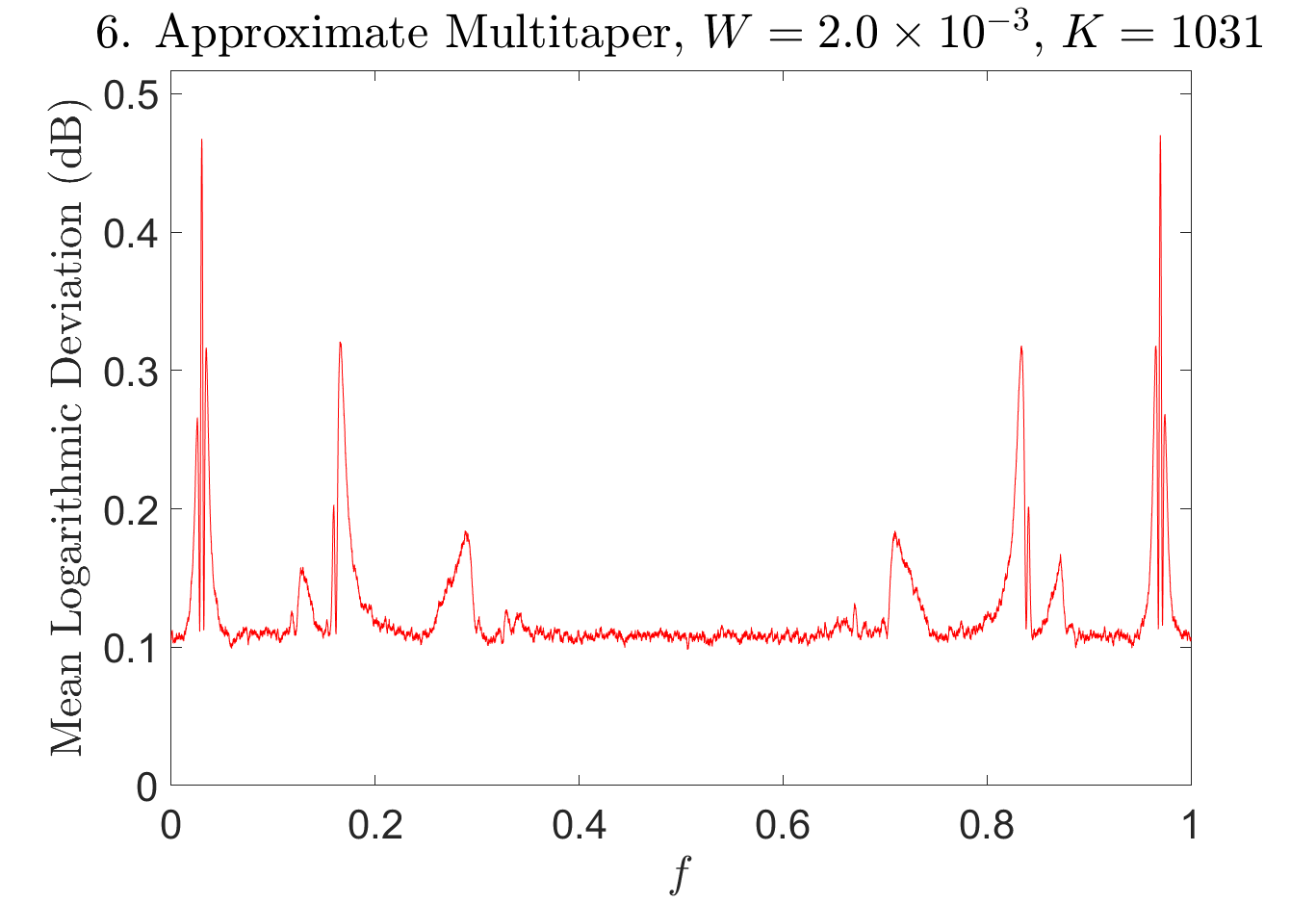}
    
  \includegraphics[width = 0.49\textwidth]{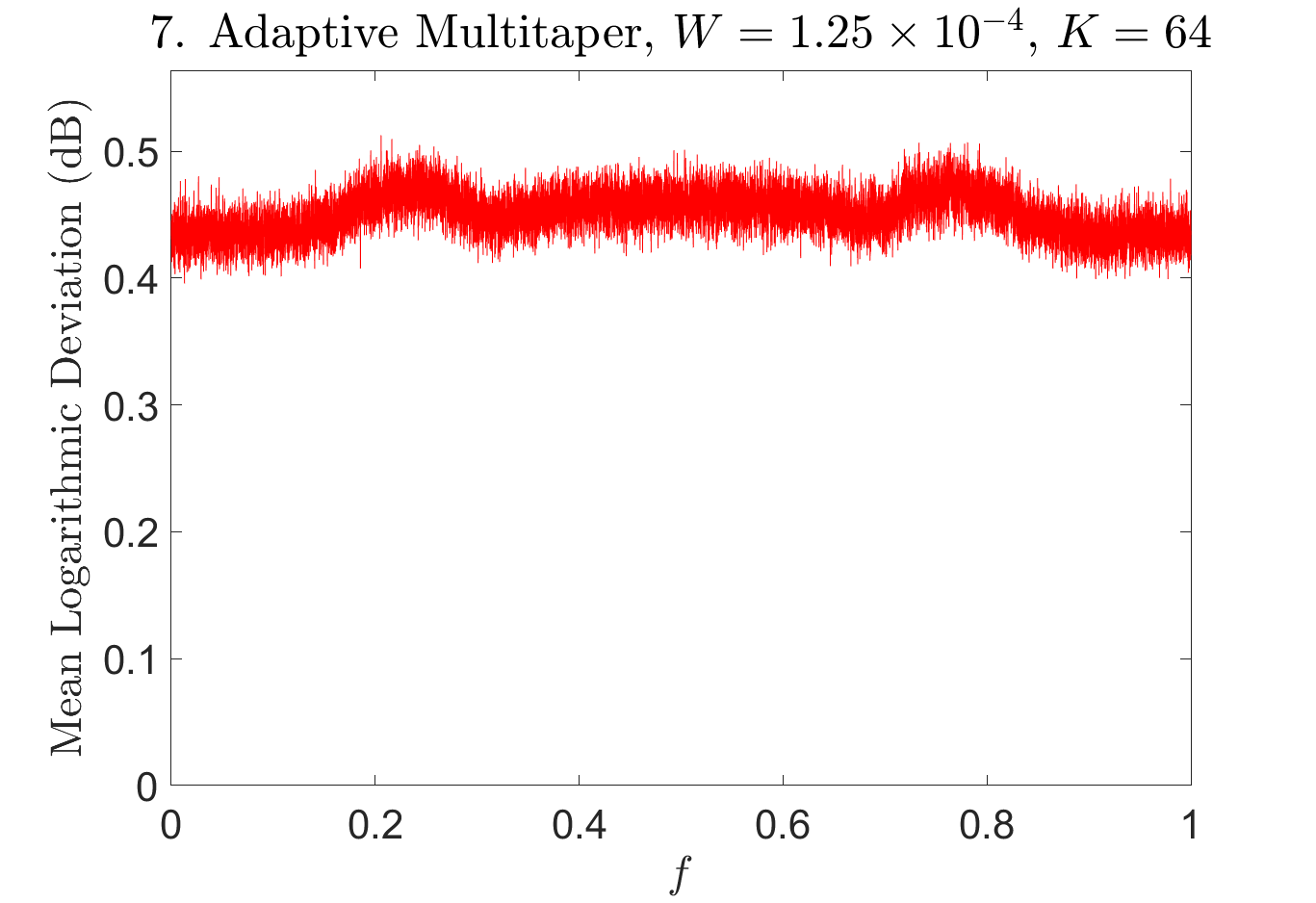} \includegraphics[width = 0.49\textwidth]{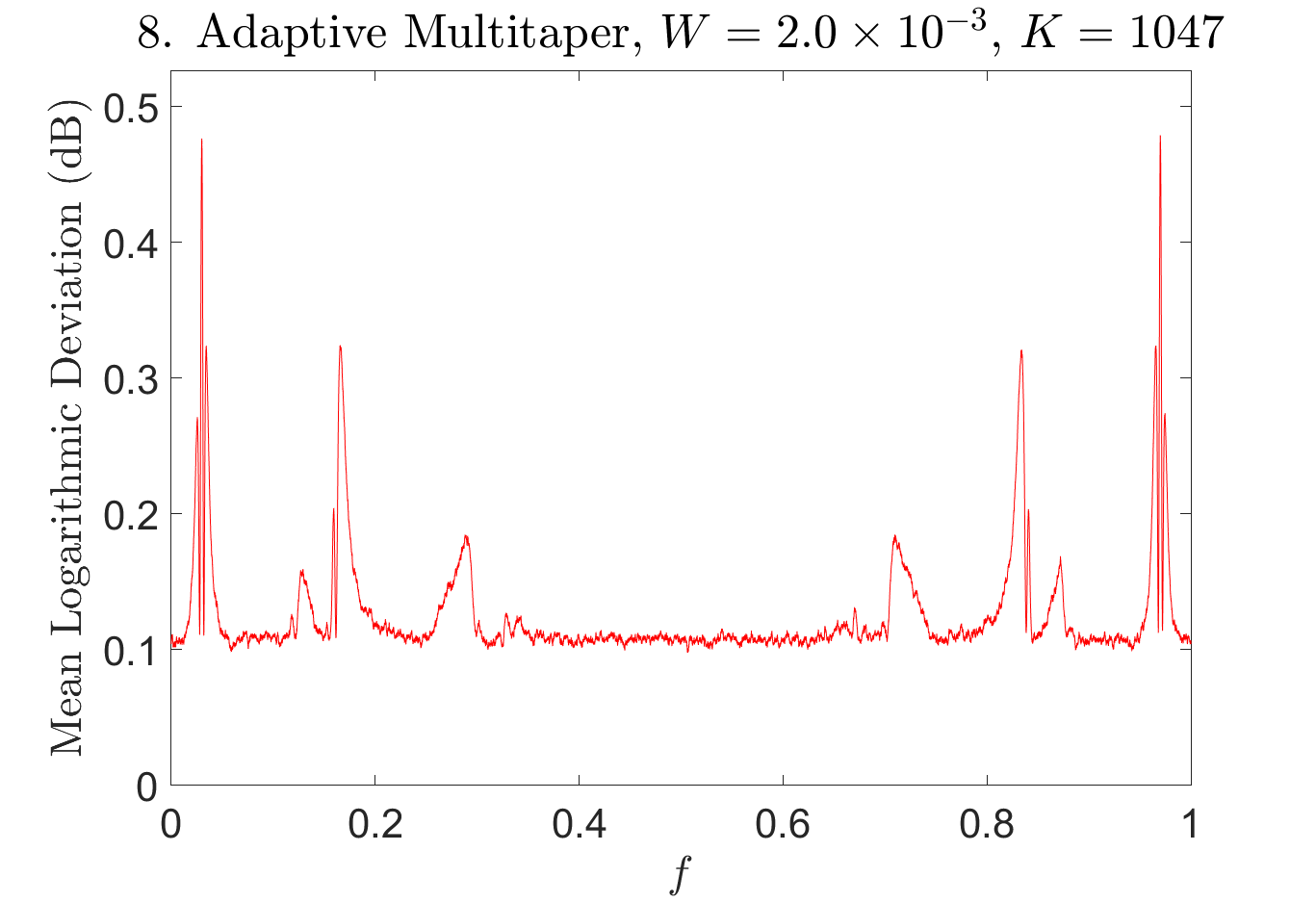}

   \caption{Plots of the empirical mean logarithmic deviation using each of the eight methods.}
   \label{fig:EmpiricalMLD}
\end{figure}
   
\begin{table}
\begin{center}
\begin{tabular}{|c|r|r|r|} \hline
Method & Avg MLD (dB) & Avg Precomputation Time (sec) & Avg Run Time (sec) \\ \hline
1 & 5.83030 &     N/A &  0.0111 \\ \hline
2 & 4.41177 &  0.1144 &  0.0111 \\ \hline
3 & 1.48562 &  4.3893 &  0.3771 \\ \hline
4 & 0.47836 &  3.6274 &  0.3093 \\ \hline
5 & 1.56164 &  3.3303 &  0.2561 \\ \hline
6 & 0.12534 &  3.3594 &  0.2533 \\ \hline
7 & 0.45080 &  4.3721 &  1.4566 \\ \hline
8 & 0.12532 & 78.9168 & 17.2709 \\ \hline
\end{tabular}
\end{center}
\caption{Table of mean logarithmic deviations (averaged across entire frequency spectrum), precomputation times, and computation times for each of the eight spectral estimation methods.}
\label{tab:SpectralEstimationResults}
\end{table}

From this, we can draw several conclusions. First, by slightly trimming the number of tapers from $K = \floor{2NW}-1$ to $2NW-O(\log(NW))$, one can significantly mitigate the amount of spectral leakage in the spectral estimate. Second, using a larger bandwidth parameter and averaging over more tapered periodograms will result in a less noisy spectral estimate. Third, our fast approximate multitaper spectral estimate can yield a more accurate estimate in the same amount of time as an exact multitaper spectral estimate because our fast approximation needs to compute significantly fewer tapers and tapered periodograms. Fourth, a multitaper spectral estimate with $\floor{2NW}-1$ tapers and adaptive weights can yield a slightly lower error than a multitaper spectral estimate with $2NW-O(\log(NW))$ tapers and fixed weights, but as $2NW$ increases, this effect becomes minimal and not worth the increased computational cost. 

\subsection{Fast multitaper spectral estimation}
\label{sec:SpeedTest}
Finally, we demonstrate that the runtime for computing the approximate multitaper spectral estimate $\tildeSmt_K(f)$ scales roughly linearly with the number of samples. We vary the signal length $N$ over several values between $2^{10}$ and $2^{20}$. For each value of $N$, we set the bandwidth parameter to be $W = 0.08 \cdot N^{-1/5}$ as this is similar to what is asymptotically optimal for a twice differentiable spectrum \cite{Abreu17}. We then choose a number of tapers $K$ such that $\lambda_{K-1} \ge 1-10^{-3} > \lambda_K$, and then compute the approximate multitaper spectral estimate $\tildeSmt_K(f)$ at $f \in [N]/N$ for each of the tolerances $\eps = 10^{-4}$, $10^{-8}$, and $10^{-12}$. If $N \le 2^{17}$, we also compute the exact multitaper spectral estimate $\hatSmt_K(f)$ at $f \in [N]/N$. For $N > 2^{17}$, we did not evaluate the exact multitaper spectral estimate due to the excessive computational requirements. 

We split the work needed to produce the spectral estimate into a precomputation stage and a computation stage. The precomputation stage involves computing the Slepian tapers which are required for the spectral estimate. In applications where a multitaper spectral estimate needs to be computed for several signals (using the same parameters $N,W,K$), computing the tapers only needs to be done once. The exact multitaper spectral estimate $\hatSmt_K(f)$ requires computing $\vs_k$ for $k \in [K]$, while the approximate multitaper spectral estimate $\tildeSmt_K(f)$ requires computing $\vs_k$ for $k \in \setI_2 \cup \setI_3 = \{k : \eps < \lambda_k < 1-\eps\}$. The computation stage involves evaluating $\hatSmt_K(f)$ or $\tildeSmt_K(f)$ at $f \in [N]/N$ with the required tapers $\vs_k$ already computed. 

In Figures~\ref{fig:PrecomputationTime} and \ref{fig:ComputationTime}, we plot the precomputation and computation time respectively versus the signal length $N$ for the exact multitaper spectral estimate $\hatSmt_K(f)$ as well as the approximate multitaper spectral estimate $\tildeSmt_K(f)$ for $\eps = 10^{-4}$, $10^{-8}$, and $10^{-12}$. The times were averaged over $100$ trials of the procedure described above. The plots are on a log-log scale. We note that the precomputation and computation times for the approximate multitaper spectral estimate grow roughly linearly with $N$ while the precomputation and computation times for the exact multitaper spectral estimate grow significantly faster. Also, the precomputation and computation times for the approximate multitaper spectral estimate with the very small tolerance $\eps = 10^{-12}$ are only two to three times more than those for the larger tolerance $\eps = 10^{-4}$.  

\begin{figure}
\centering
\includegraphics[scale=0.40]{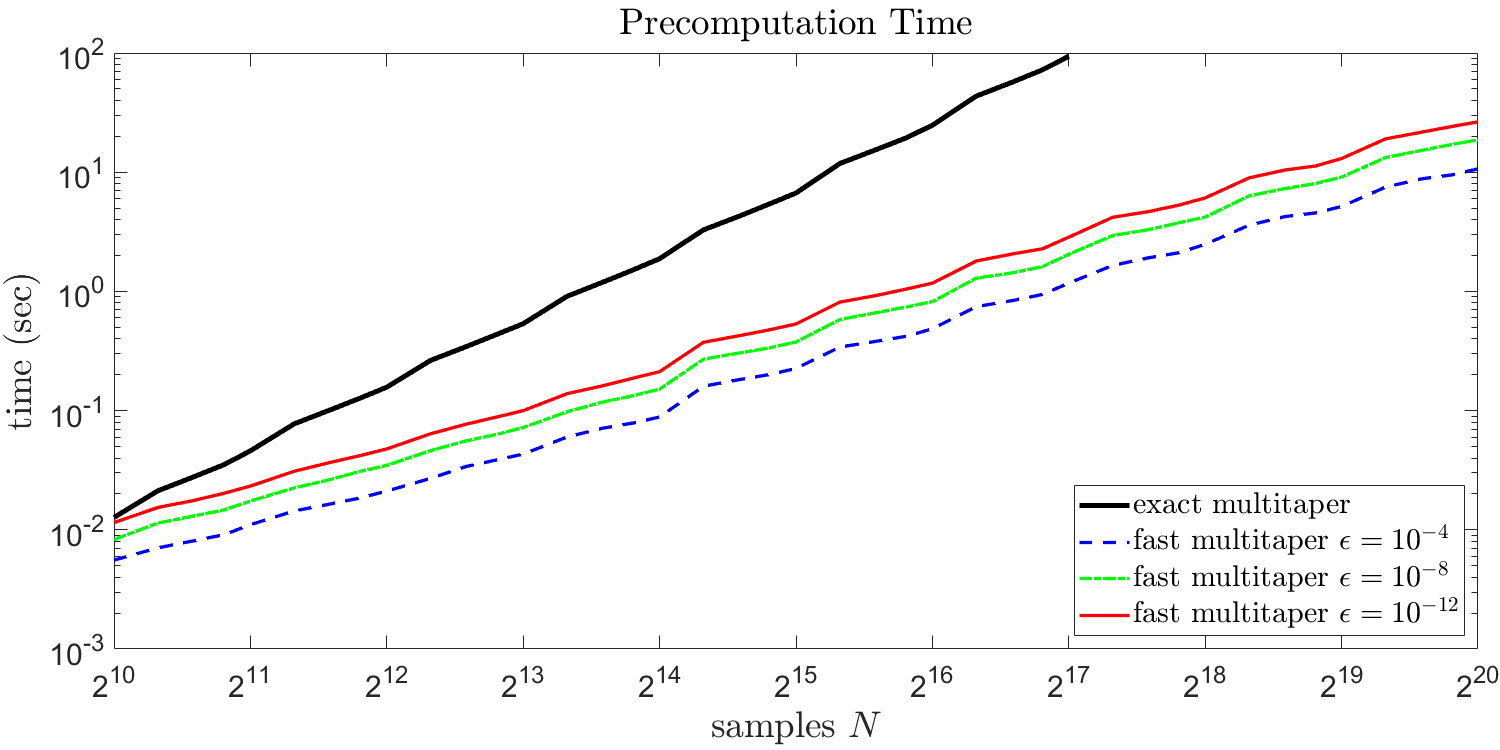}
\caption{Plot of the average precomputation time vs. signal length $N$ for the exact multitaper spectral estimate and our fast approximation for $\eps = 10^{-4}$, $10^{-8}$, and $10^{-12}$.}
\label{fig:PrecomputationTime}
\end{figure}

\begin{figure}
\centering
\includegraphics[scale=0.40]{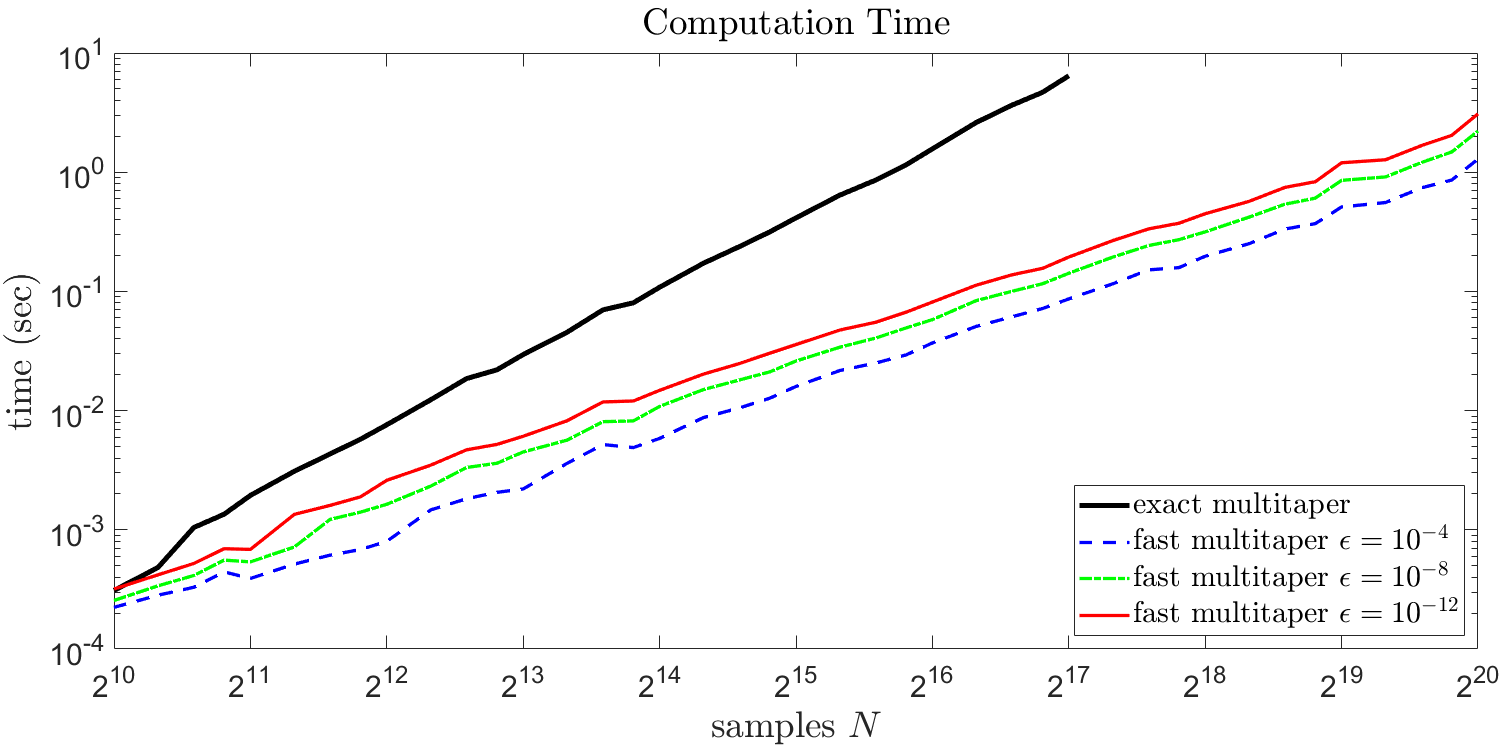}
\caption{Plot of the average computation time vs. signal length $N$ for the exact multitaper spectral estimate and our fast approximation for $\eps = 10^{-4}$, $10^{-8}$, and $10^{-12}$.}
\label{fig:ComputationTime}
\end{figure}

\section{Conclusions}
\label{sec:Conclusions}
Thomson's multitaper method is a widely-used tool for spectral estimation. In this paper, we have presented a linear algebra based perspective of Thomson's multitaper method. Specifically, the multitaper spectral estimate $\hatSmt_K(f)$ can be viewed as the $2$-norm energy of the projection of the signal $\vx \in \C^N$ onto a $K$-dimensional subspace $\mathcal{S}_f \subset \C^N$. The subspace $\mathcal{S}_f$ is chosen to be the optimal $K$-dimensional subspace for representing sinusoids with frequencies between $f-W$ and $f+W$. 

We have also provided non-asymptotic bounds on the bias, variance, covariance, and probability tails of the multitaper spectral estimate. In particular, if the multitaper spectral estimate uses $K = 2NW-O(\log(NW))$ tapers, the effects of spectral leakage are mitigated, and our non-asymptotic bounds are very similar to existing asymptotic results. However, when the traditional choice of $K = \floor{2NW}-1$ tapers is used, the non-asymptotic bounds are significantly worse, and our simulations show the multitaper spectral estimate suffers severely from spectral leakage. 

Finally, we have derived an $\eps$-approximation to the multitaper spectral estimate which can be evaluated at a grid of $L$ evenly spaced frequencies in $O(L\log L \log(NW)\log\tfrac{1}{\eps})$ operations. Since evaluating the exact multitaper spectral estimate at a grid of $L$ evenly spaced frequencies requires $O(KL\log L)$ operations, our $\eps$-approximation is faster when $K \gtrsim \log(NW)\log\tfrac{1}{\eps}$ tapers are required. In problems involving a large number of samples, minimizing the mean-squared error of the multitaper spectral estimate requires using a large number of tapers, and thus, our $\eps$-approximation can drastically reduce the required computation.






\bibliographystyle{unsrt}
\bibliography{citationsThomsonsMethodRevisited}

\appendix
\section{Proof of Results in Section~\ref{sec:ParameterSelection}}
\label{sec:ParameterSelectionProofs}

\subsection{Norms of Gaussian random variables}
In this subsection, we develop results on the $2$-norm of linear combinations of Gaussian random variables, which will be critical to our proofs of the theorems in Section~\ref{sec:ParameterSelection}.

First, we state a result from \cite{JanssenStoica88} regarding the product of four jointly complex Gaussian vectors.
\begin{lemma} \cite{JanssenStoica88}
\label{lem:FourGaussians}
Suppose $\va, \vb, \vc, \vd \in \C^K$ have a joint complex Gaussian distribution. Then $$\E[\va^*\vb\vc^*\vd] = \E[\va^*\vb]\E[\vc^*\vd] + \E[\vc^* \otimes \va^*]\E[\vd \otimes \vb] + \E[\va^*\E[\vb\vc^*]\vd] + 2\E[\va^*]\E[\vb]\E[\vc^*]\E[\vd],$$ where $\otimes$ denotes the Kronecker product.
\end{lemma}
It should be noted that \cite{JanssenStoica88} proves a more general result for the product of four joint complex Gaussian matrices, but this is a bit harder to state. So we give the result for vectors.

We now prove a result regarding the norm-squared of linear combinations of complex Gaussians. 
\begin{lemma}
\label{lem:ExpCovNormGaussians}
Let $\vx \sim \mathcal{CN}(\vct{0},\mR)$ for some positive semidefinite $\mR \in \C^{N \times N}$ and let $\mU,\mV \in \C^{K \times N}$. Then, we have:
$$\E\left[\left\|\mU\vx\right\|_2^2\right] = \tr[\mU\mR\mU^*] \quad \text{and} \quad \Cov\left[\left\|\mU\vx\right\|_2^2,\left\|\mV\vx\right\|_2^2\right] = \|\mU\mR\mV^*\|_F^2.$$
\end{lemma}

\begin{proof}
The expectation of $\|\mU\vx\|_2^2$ can be computed as follows $$\E\left[\left\|\mU\vx\right\|_2^2\right] = \E\left[\tr\left[\mU\vx\vx^*\mU^*\right]\right] = \tr\left[\E\left[\mU\vx\vx^*\mU^*\right]\right] = \tr\left[\mU\E\left[\vx\vx^*\right]\mU^*\right] = \tr\left[\mU\mR\mU^*\right].$$

Since the entries of $\vx$ are jointly Gaussian with mean $0$, the entries of $\mU\vx$ and $\mV\vx$ are also jointly Gaussian with mean $0$. Then, by applying Lemma~\ref{lem:FourGaussians} (with $\va = \vb = \mU\vx$ and $\vc = \vd = \mV\vx$), we find that 
\begin{align*}
\Cov\left[\left\|\mU\vx\right\|_2^2,\left\|\mV\vx\right\|_2^2\right] &= \E\left[\left\|\mU\vx\right\|_2^2 \cdot \left\|\mV\vx\right\|_2^2\right] - \E\left[\left\|\mU\vx\right\|_2^2\right] \E\left[\left\|\mV\vx\right\|_2^2\right]
\\
&= \E\left[\vx^*\mU^*\mU\vx\vx^*\mV^*\mV\vx\right] -  \E\left[\vx^*\mU^*\mU\vx\right]\E\left[\vx^*\mV^*\mV\vx\right]
\\
&= \E\left[\vx^*\mV^* \otimes \vx^*\mU^*\right]\E\left[\mV\vx \otimes \mU\vx\right] + \E\left[\vx^*\mU^*\E\left[\mU\vx\vx^*\mV^*\right]\mV\vx\right]
\\
& \ \ \ \ \ + 2\E[\vx^*\mU^*]\E[\mU\vx]\E[\vx^*\mV^*]\E[\mV\vx].
\end{align*}
We proceed to evaluate each of these three terms. 

Since $\vx \sim \mathcal{CN}(\vct{0},\mR)$, we can write $\vx = \mR^{1/2}\vy$ where $\vy \sim \mathcal{CN}(\vct{0},\mId)$ and $\mR^{1/2}$ is the unique positive semidefinite squareroot of $\mR$. Then, by using the identity $\mX_1\mX_2 \otimes \mX_3\mX_4 = (\mX_1 \otimes \mX_3)(\mX_2 \otimes \mX_4)$ for appropriately sized matrices $\mX_1, \mX_2, \mX_3, \mX_4$, we get $$\E[\mV\vx \otimes \mU\vx] = \E\left[\mV\mR^{1/2}\vy \otimes \mU\mR^{1/2}\vy\right] = \E\left[\left(\mV\mR^{1/2} \otimes \mU\mR^{1/2}\right)(\vy \otimes \vy)\right] = \left(\mV\mR^{1/2} \otimes \mU\mR^{1/2}\right)\E[\vy \otimes \vy]$$ Since the entries of $\vy$ are i.i.d. $\mathcal{CN}(0,1)$, $\E\left[\vy[n]\vy[n']\right] = 0$ for all indices $n,n' = 0,\ldots,N-1$. (Note that $\E[\vy[n]^2] \neq 0$ if the entries of $\vy$ were i.i.d. $\mathcal{N}(0,1)$ instead of $\mathcal{CN}(0,1)$.) Hence, $\E[\vy \otimes \vy] = \vct{0}$, and thus, $\E[\mV\vx \otimes \mU\vx] = (\mV\mR^{1/2} \otimes \mU\mR^{1/2})\E[\vy \otimes \vy] = \vct{0}$. Similarly, $\E[\vx^*\mV^* \otimes \vx^*\mU^*] = \vct{0}^*$, and so, $\E\left[\vx^*\mV^* \otimes \vx^*\mU^*\right]\E\left[\mV\vx \otimes \mU\vx\right] = 0$. 

Using the cyclic property of the trace operator, linearity of the trace and expectation operators, and the fact that $\E[\vx\vx^*] = \mR$, we find that
\begin{align*}
E\left[\vx^*\mU^*\E\left[\mU\vx\vx^*\mV^*\right]\mV\vx\right] &= \E\left[\tr\left[\E\left[\mU\vx\vx^*\mV^*\right]\mV\vx\vx^*\mU^*\right]\right] 
\\
&= \tr\left[\E\left[\E\left[\mU\vx\vx^*\mV^*\right]\mV\vx\vx^*\mU^*\right]\right]
\\
&= \tr\left[\mU\E\left[\vx\vx^*\right]\mV^*\mV\E\left[\vx\vx^*\right]\mU^*\right]
\\
&= \tr\left[\mU\mR\mV^*\mV\mR\mU^*\right] 
\\
&= \left\|\mU\mR\mV^*\right\|_F^2.
\end{align*}

Finally, since $\E[\vx] = \vct{0}$, we have $\E[\mU\vx] = \E[\mV\vx] = \vct{0}$ and $\E[\vx^*\mU^*] = \E[\vx^*\mV^*] = \vct{0}^*$. Hence, $$2\E[\vx^*\mU^*]\E[\mU\vx]\E[\vx^*\mV^*]\E[\mV\vx] = 0.$$

Adding these three terms yields,
\begin{align*}
\Cov\left[\left\|\mU\vx\right\|_2^2,\left\|\mV\vx\right\|_2^2\right] &= \E\left[\vx^*\mV^* \otimes \vx^*\mU^*\right]\E\left[\mV\vx \otimes \mU\vx\right] + \E\left[\vx^*\mU^*\E\left[\mU\vx\vx^*\mV^*\right]\mV\vx\right]
\\
& \ \ \ \ \ + 2\E[\vx^*\mU^*]\E[\mU\vx]\E[\vx^*\mV^*]\E[\mV\vx].
\\
&= \left\|\mU\mR\mV^*\right\|_F^2.
\end{align*}
\end{proof}

\subsection{Concentration of norms of Gaussian random variables}
Next, we state a result from \cite{Janson18} regarding concentration bounds for sums of independent exponential random variables.
\begin{lemma}
\label{lem:ExpChernoff}\cite{Janson18}
Let $Z_0,\ldots,Z_{N-1}$ be independent exponential random variables with $\E[Z_n] = \mu_n$. Then, the sum $$Z := \displaystyle\sum_{n = 0}^{N-1}Z_n$$ satisfies $$\displaystyle \P\left\{Z \ge \beta\E[Z]\right\} \le \beta^{-1}e^{-\kappa(\beta-1-\ln\beta)} \quad \text{for} \quad \beta > 1,$$ and $$\displaystyle \P\left\{Z \le \beta\E[Z]\right\} \le e^{-\kappa(\beta-1-\ln\beta)} \quad \text{for} \quad 0 < \beta < 1,$$ where $$\kappa = \dfrac{\sum\limits_{n = 0}^{N-1}\mu_n}{\max\limits_{n = 0,\ldots,N-1}\mu_n}.$$
\end{lemma}

We now apply this lemma to derive concentration bounds for $\|\mA\vx\|_2^2$, where $\vx$ is a vector of Gaussian random variables and $\mA$ is a matrix.
\begin{lemma}
\label{lem:GaussianConcentration}
Let $\vx \sim \mathcal{CN}(\vct{0},\mR)$ for some positive semidefinite $\mR \in \C^{N \times N}$. Also, let $\mA \in \C^{K \times N}$. Then, the random variable $\left\|\mA\vx\right\|_2^2$ satisfies $$\P\left\{\left\|\mA\vx\right\|_2^2 \ge \beta \E\left[\left\|\mA\vx\right\|_2^2\right]\right\} \le \beta^{-1}e^{-\kappa(\beta-1-\ln\beta)} \quad \text{for} \quad \beta > 1,$$ and $$\P\left\{\left\|\mA\vx\right\|_2^2 \le \beta \E\left[\left\|\mA\vx\right\|_2^2\right]\right\} \le e^{-\kappa(\beta-1-\ln\beta)} \quad \text{for} \quad 0 < \beta < 1,$$ where $$\kappa = \dfrac{\tr\left[\mA\mR\mA^*\right]}{\left\|\mA\mR\mA^*\right\|}.$$
\end{lemma}

\begin{proof}
Since $\vx \sim \mathcal{CN}(\vct{0},\mR)$, we can write $\vx = \mR^{1/2}\vy$ where $\vy \sim \mathcal{CN}(\vct{0},\mId)$ and $\mR^{1/2}$ is the unique positive semidefinite squareroot of $\mR$. Then, using eigendecomposition, we can write $\mR^{1/2}\mA^*\mA\mR^{1/2} = \mW\mD\mW^*$ where $\mW$ is unitary and $\mD = \diag(d_1,\ldots,\d_n)$. Since $\mW$ is unitary, $\vz := \mW^*\vy \sim \mathcal{CN}(\vct{0},\mId_N)$. Then, we have: $$\left\|\mA\vx\right\|_2^2 = \vx^*\mA^*\mA\vx = \vy^*\mR^{1/2}\mA^*\mA\mR^{1/2}\vy = \vy^*\mW\mD\mW^*\vy = \vz^*\mD\vz = \sum_{n = 0}^{N-1}d_n\left|\vz[n]\right|^2.$$ Since $\vz \sim \mathcal{CN}(\vct{0},\mId_N)$, we have that $\vz[0],\ldots,\vz[N-1]$ are i.i.d. $\mathcal{CN}(0,1)$. Hence, $d_n|\vz[n]|^2 \sim \text{Exp}(d_n)$, and are independent. If we apply Lemma~\ref{lem:ExpChernoff}, the fact that the trace and operator norm are invariant under unitary similarity transforms, and the matrix identities $\tr[\mX\mX^*] = \tr[\mX^*\mX]$ and $\|\mX\mX^*\| = \|\mX^*\mX\|$, we obtain $$\displaystyle \P\left\{\left\|\mA\vx\right\|_2^2 \ge \beta \E\left[\left\|\mA\vx\right\|_2^2\right]\right\} \le \beta^{-1}e^{-\kappa(\beta-1-\ln\beta)} \quad \text{for} \quad \beta > 1,$$ and $$\displaystyle \P\left\{\left\|\mA\vx\right\|_2^2 \le \beta \E\left[\left\|\mA\vx\right\|_2^2\right]\right\} \le e^{-\kappa(\beta-1-\ln\beta)} \quad \text{for} \quad 0 < \beta < 1,$$ where $$\kappa = \dfrac{\sum\limits_{n = 0}^{N-1}d_n}{\max\limits_{n = 0,\ldots,N-1} d_n} = \dfrac{\tr[\mD]}{\|\mD\|} = \dfrac{\tr[\mW\mD\mW^*]}{\|\mW\mD\mW^*\|} = \dfrac{\tr\left[\mR^{1/2}\mA^*\mA\mR^{1/2}\right]}{\left\|\mR^{1/2}\mA^*\mA\mR^{1/2}\right\|} = \dfrac{\tr\left[\mA\mR\mA^*\right]}{\left\|\mA\mR\mA^*\right\|}.$$
\end{proof}

We note that similar bounds can be obtained by applying the Hanson-Wright Inequality \cite{Hanson71,Wright73}.  

\subsection{Intermediate results}
We continue by presenting a lemma showing that certain matrices can be represented as an integral of a frequency dependent rank-$1$ matrix.
\begin{lemma}
For any frequency $f \in \R$, define a complex sinusoid $\ve_f \in \C^N$ by $\ve_f[n] = e^{j2\pi fn}$ for $n = 0,\ldots,N-1$. Then, we have: $$\mB = \int_{-W}^{W}\ve_f\ve_f^*\,df, \quad \mId = \int_{-1/2}^{1/2}\ve_f\ve_f^*\,df, \quad \text{and} \quad \int_{\Omega}\ve_f\ve_f^*\,df = \mId-\mB,$$ where $\Omega = [-\tfrac{1}{2},\tfrac{1}{2}] \setminus [-W,W]$. Furthermore, if $x(t)$ is a stationary, ergodic, zero-mean, Gaussian stochastic process $x(t)$ with power spectral density $S(f)$, and $\vx \in \C^N$ is a vector of equispaced samples $\vx[n] = x(n)$ for $n = 0,\ldots,N-1$, then the covariance matrix of $\vx$ can be written as $$\mR := \E[\vx\vx^*] = \int_{-1/2}^{1/2}S(f)\ve_f\ve_f^*\,df.$$
\end{lemma}
\begin{proof}
For any $m,n = 0,\ldots,N-1$, we have $$\int_{-W}^{W}\ve_f[m]\overline{\ve_f[n]}\,df = \int_{-W}^{W}e^{j2\pi f(m-n)}\,df = \dfrac{\sin[2\pi W(m-n)]}{\pi(m-n)} = \mB[m,n],$$ and $$\int_{-1/2}^{1/2}\ve_f[m]\overline{\ve_f[n]}\,df = \int_{-1/2}^{1/2}e^{j2\pi f(m-n)}\,df = \begin{cases}1 & \text{if} \ m = n \\ 0 & \text{if} \ m \neq n\end{cases} = \mId[m,n].$$ We can put these into matrix form as $$\int_{-W}^{W}\ve_f\ve_f^*\,df = \mB \quad \text{and} \quad \int_{-1/2}^{1/2}\ve_f\ve_f^*\,df = \mId.$$ From this, it follows that $$\int_{\Omega}\ve_f\ve_f^*\,df = \int_{-1/2}^{1/2}\ve_f\ve_f^*\,df - \int_{-W}^{W}\ve_f\ve_f^*\,df = \mId-\mB.$$ Finally, using the definition of the power spectral density, we have $$\mR[m,n] = \E[\vx[m]\overline{\vx[n]}] = \int_{-1/2}^{1/2}S(f)e^{j2\pi f(m-n)}\,df = \int_{-1/2}^{1/2}S(f)\ve_f[m]\overline{\ve_f[n]}\,df$$ for $m,n = 0,\ldots,N-1$. Again, we can put this into matrix form as $$\mR := \E[\vx\vx^*] = \int_{-1/2}^{1/2}S(f)\ve_f\ve_f^*\,df.$$
\end{proof}

Next, we show that the expectation of the multitaper spectral estimate is the convolution of the power spectral density $S(f)$ with the spectral window $\psi(f)$.

\begin{lemma}
\label{lem:MultitaperExpectation}
The expectation of the multitaper spectral estimate can be written as $$\E\left[\hatSmt_K(f)\right] = \int_{-1/2}^{1/2}S(f-f')\psi(f')\,df'$$ where $$\psi(f) := \displaystyle \dfrac{1}{K}\tr\left[\mS_K^*\ve_{-f}\ve_{-f}^*\mS_K\right] = \dfrac{1}{K}\left\|\mS_K^*\ve_{-f}\right\|_2^2 = \dfrac{1}{K}\sum_{k = 0}^{K-1}\left|\sum_{n = 0}^{N-1}\vs_k[n]e^{-j2\pi fn}\right|^2$$ is the spectral window of the multitaper estimate. 
\end{lemma}

\begin{proof}
Since $\hatSmt_K(f) = \dfrac{1}{K}\left\|\mS_K^*\mE_f^*\vx\right\|_2^2$ where $\vx \sim \mathcal{CN}(\vct{0},\mR)$, by \cref{lem:ExpCovNormGaussians}, we have $$\E\left[\hatSmt_K(f)\right] = \dfrac{1}{K}\tr\left[\mS_K^*\mE_f^*\mR\mE_f\mS_K\right].$$ We can rewrite this expression as follows:

\begin{align*}
\E\left[\hatSmt_K(f)\right] &= \dfrac{1}{K}\tr\left[\mS_K^*\mE_f^*\mR\mE_f\mS_K\right]
\\
&= \dfrac{1}{K}\tr\left[\mS_K^*\mE_f^*\left(\int_{-1/2}^{1/2}S(f')\ve_{f'}\ve_{f'}^*\,df'\right)\mE_f\mS_K\right]
\\
&= \dfrac{1}{K}\int_{-1/2}^{1/2}S(f')\tr\left[\mS_K^*\mE_f^*\ve_{f'}\ve_{f'}^*\mE_f\mS_K\right]\,df'
\\
&= \dfrac{1}{K}\int_{-1/2}^{1/2}S(f')\tr\left[\mS_K^*\ve_{f'-f}\ve_{f'-f}^*\mS_K\right]\,df'
\\
&= \int_{-1/2}^{1/2}S(f')\psi(f-f')\,df',
\\
&= \int_{f-1/2}^{f+1/2}S(f-f')\psi(f')\,df'
\\
&= \int_{-1/2}^{1/2}S(f-f')\psi(f')\,df',
\end{align*}
\\
where the last line follows since $S(f)$ is $1$-periodic (by definition) and $\psi(f) = \dfrac{1}{K}\|\mS_K^*\ve_{-f}\|_2^2$ is $1$-periodic (because $\ve_{-f}[n] = e^{-2\pi fn}$ is $1$-periodic for all $n$).

Finally, it's easy to check that $\psi(f)$ can be written in the following alternate forms: $$\psi(f) = \dfrac{1}{K}\tr\left[\mS_K^*\ve_{-f}\ve_{-f}^*\mS_K\right] = \dfrac{1}{K}\ve_{-f}^*\mS_K\mS_K^*\ve_{-f} = \dfrac{1}{K}\|\mS_K^*\ve_{-f}\|_2^2 = \dfrac{1}{K}\sum_{k = 0}^{K-1}\left|\sum_{n = 0}^{N-1}\vs_k[n]e^{-j2\pi fn}\right|^2.$$
\end{proof}

\begin{lemma}
\label{lem:SpectralWindow}
The spectral window $\psi(f)$ defined in Lemma~\ref{lem:MultitaperExpectation}, satisfies the following properties:
\begin{itemize}
\item $\psi(f) = \psi(-f)$ for all $f \in \R$
\item $\displaystyle\int_{-W}^{W}\psi(f)\,df = 1-\Sigma^{(1)}_K$ and $\displaystyle\int_{\Omega}\psi(f)\,df = \Sigma^{(1)}_K$ where $\Omega = [-\tfrac{1}{2},\tfrac{1}{2}] \setminus [-W,W]$
\item $0 \le \psi(f) \le \dfrac{N}{K}$
\end{itemize}
\end{lemma}
\begin{proof}
First, the spectral window is an even function since $$\psi(-f) = \dfrac{1}{K}\sum_{k = 0}^{K-1}\left|\sum_{n = 0}^{N-1}\vs_k[n]e^{j2\pi fn}\right|^2 = \dfrac{1}{K}\sum_{k = 0}^{K-1}\left|\sum_{n = 0}^{N-1}\vs_k[n]e^{-j2\pi fn}\right|^2 = \psi(f)$$ for all $f \in \R$. 

As a consequence, we can set $\mS_K = \begin{bmatrix}\vs_0 & \vs_1 & \cdots & \vs_{K-1}\end{bmatrix} \in \R^{N \times K}$ and write $\psi(f) = \psi(-f) = \dfrac{1}{K}\tr\left[\mS_K^*\ve_f\ve_f^*\mS_K\right]  = \dfrac{1}{K}\left\|\mS_K^*\ve_f\right\|_2^2$. Then, the integral of the spectral window $\psi(f)$ over $[-W,W]$ is
\begin{align*}
\int_{-W}^{W}\psi(f)\,df &= \int_{-W}^{W}\dfrac{1}{K}\tr\left[\mS_K^*\ve_f\ve_f^*\mS_K\right]\,df 
\\
&= \dfrac{1}{K}\tr\left[\mS_K^*\left(\int_{-W}^{W}\ve_f\ve_f^*\,df\right)\mS_K\right]
\\
&= \dfrac{1}{K}\tr\left[\mS_K^*\mB\mS_K\right] 
\\
&= \dfrac{1}{K}\tr\left[\mLambda_K\right]
\\
&= \dfrac{1}{K}\sum_{k = 0}^{K-1}\lambda_k 
\\
&= 1-\dfrac{1}{K}\sum_{k = 0}^{K-1}(1-\lambda_k) 
\\
&= 1 - \Sigma_K^{(1)},
\end{align*}
where we have used the notation $\mLambda_K = \diag(\lambda_0,\ldots,\lambda_{K-1})$. Similarly, the integral of the spectral window $\psi(f)$ over $\Omega = [-\tfrac{1}{2},-W) \cup (W,\tfrac{1}{2}]$ is
\begin{align*}
\int_{\Omega}\psi(f)\,df &= \int_{\Omega}\dfrac{1}{K}\tr\left[\mS_K^*\ve_f\ve_f^*\mS_K\right]\,df 
\\
&= \dfrac{1}{K}\tr\left[\mS_K^*\left(\int_{\Omega}\ve_f\ve_f^*\,df\right)\mS_K\right]
\\
&= \dfrac{1}{K}\tr\left[\mS_K^*(\mId-\mB)\mS_K\right] 
\\
&= \dfrac{1}{K}\tr\left[\mId-\mLambda_K\right] 
\\
&= \dfrac{1}{K}\sum_{k = 0}^{K-1}(1-\lambda_k) 
\\
&= \Sigma_K^{(1)}.
\end{align*}

Finally, the spectral window is bounded by $$0 \le \psi(f) = \dfrac{1}{K}\left\|\mS_K^*\ve_f\right\|_2^2 \le \dfrac{1}{K}\left\|\mS_K^*\right\|^2\left\|\ve_f\right\|_2^2 = \dfrac{N}{K},$$ where we have used the facts that $\|\mS_K^*\| = 1$ (because $\mS_K$ is orthonormal) and $\left\|\ve_f\right\|_2^2 = N$. 
\end{proof}

\subsection{Proof of Theorem~\ref{thm:Bias2Diff}}
By Lemma~\ref{lem:MultitaperExpectation}, we have $$\E\left[\hatSmt_K(f)\right] = \int_{-1/2}^{1/2}S(f-f')\psi(f')\,df'.$$ We now split the expression for the bias into two pieces as follows:
\begin{align*}
\Bias\left[\hatSmt_K(f)\right] &= \left|\E\hatSmt_K(f) - S(f)\right|
\\
&= \left|\int_{-1/2}^{1/2}S(f-f')\psi(f')\,df' - S(f)\right|
\\
&= \left|\int_{-W}^{W}S(f-f')\psi(f')\,df' + \int_{\Omega}S(f-f')\psi(f')\,df' - S(f)\right|
\\
&\le \underbrace{\left|\int_{-W}^{W}S(f-f')\psi(f')\,df' - S(f)\right|}_{\text{local bias}} + \underbrace{\left|\int_{\Omega}S(f-f')\psi(f')\,df'\right|}_{\text{broadband bias}}.
\end{align*}

Since $S(f)$ is twice continuously differentiable, for any $f' \in [-W,W]$, there exists a $\xi_{f'}$ between $f-f'$ and $f$ such that $$S(f-f') = S(f) - S'(f)f' + \dfrac{1}{2}S''(\xi_{f'})f'^2.$$ Then, since $\left|S''(\xi_{f'})\right| \le \displaystyle\max_{\xi \in [f-W,f+W]}|S''(\xi)| = M''_f$, $\int_{-W}^{W}\psi(f)\,df = 1-\Sigma^{(1)}_K$, and $0 \le \psi(f') \le \tfrac{N}{K}$ for all $f' \in \R$, we can bound the local bias as follows:
\begin{align*}
\left|\int_{-W}^{W}S(f-f')\psi(f')\,df' - S(f)\right| &= \left|\int_{-W}^{W}\left(S(f)-S'(f)f'+\dfrac{1}{2}S''(\xi_{f'})f'^2\right)\psi(f')\,df' - S(f)\right|
\\
&= \left|\int_{-W}^{W}S(f)\psi(f')\,df'-S(f) - \int_{-W}^{W}\underbrace{S'(f)f'\psi(f')}_{\text{odd w.r.t.} \ f'}\,df' +\dfrac{1}{2}\int_{-W}^{W}S''(\xi_{f'})f'^2\psi(f')\,df'\right|
\\
&= \left|\int_{-W}^{W}S(f)\psi(f')\,df'-S(f) - 0 +\dfrac{1}{2}\int_{-W}^{W}S''(\xi_{f'})f'^2\psi(f')\,df'\right|
\\
&\le \left|\int_{-W}^{W}S(f)\psi(f')\,df'-S(f)\right| + \left|\dfrac{1}{2}\int_{-W}^{W}S''(\xi_{f'})f'^2\psi(f')\,df'\right|
\\
&= \left|S(f)(1-\Sigma_K^{(1)})-S(f)\right| + \left|\dfrac{1}{2}\int_{-W}^{W}S''(\xi_{f'})f'^2\psi(f')\,df'\right|
\\
&\le S(f)\Sigma_K^{(1)} + \dfrac{1}{2}\int_{-W}^{W}|S''(\xi_{f'})||f'|^2\psi(f')\,df'
\\
&\le S(f)\Sigma_K^{(1)} + \dfrac{1}{2}\int_{-W}^{W}M''_f|f'|^2\dfrac{N}{K}\,df'
\\
&= S(f)\Sigma_K^{(1)} + \dfrac{M''_fNW^3}{3K}
\\
&\le M_f\Sigma_K^{(1)} + \dfrac{M''_fNW^3}{3K}
\end{align*}

Since $S(f') \le \displaystyle\max_{f' \in \R}S(f') = M$, we can bound the broadband bias as follows: 
\begin{align*}
\left|\int_{\Omega}S(f-f')\psi(f')\,df'\right| &= \int_{\Omega}S(f-f')\psi(f')\,df'
\\
&\le \int_{\Omega}M\psi(f')\,df'
\\
&= M\Sigma_K^{(1)}
\end{align*}

Combining the bounds on the local bias and broadband bias yields \begin{align*}\Bias\left[\hatSmt_K(f)\right] &\le \underbrace{\left|\int_{-W}^{W}S(f-f')\psi(f')\,df' - S(f)\right|}_{\text{local bias}} + \underbrace{\left|\int_{\Omega}S(f-f')\psi(f')\,df'\right|}_{\text{broadband bias}}
\\
&\le M_f\Sigma_K^{(1)}+\dfrac{M''_fNW^3}{3K} + M\Sigma_K^{(1)}
\\
&= \dfrac{M''_fNW^3}{3K} + (M+M_f)\Sigma_K^{(1)},
\end{align*}
which establishes Theorem~\ref{thm:Bias2Diff}.

\subsection{Proof of Theorem~\ref{thm:BiasNotDiff}}
Without the assumption that $S(f)$ is twice differentiable, we can still obtain a bound on the bias. Using the bounds $m_f = \displaystyle\min_{\xi \in [f-W,f+W]}S(\xi) \le S(f') \le \max_{\xi \in [f-W,f+W]}S(\xi) = M_f$ and $0 \le S(f') \le \max_{\xi \in \R}S(\xi) = M$ along with the integrals $\int_{-W}^{W}\psi(f)\,df = 1-\Sigma^{(1)}_K$ and $\int_{\Omega}\psi(f)\,df = \Sigma^{(1)}_K$, we can obtain the following upper bound on $\E\hatSmt_K(f)-S(f)$:
\vspace{0.00 in}
\begin{align*}
\E\hatSmt_K(f)-S(f) &= \int_{-1/2}^{1/2}S(f-f')\psi(f')\,df' - S(f') 
\\
&= \int_{-W}^{W}S(f-f')\psi(f')\,df' + \int_{\Omega}S(f-f')\psi(f')\,df' - S(f')
\\
&\le \int_{-W}^{W}M_f\psi(f')\,df' + \int_{\Omega}M\psi(f')\,df' - m_f
\\
&= M_f(1-\Sigma^{(1)}_K) + M\Sigma^{(1)}_K - m_f
\\
&= (M_f-m_f)(1-\Sigma^{(1)}_K) + (M-m_f)\Sigma^{(1)}_K
\\
&\le (M_f-m_f)(1-\Sigma^{(1)}_K) + M\Sigma^{(1)}_K,
\end{align*} 
\\
Similarly, we can obtain the following lower bound on $\E\hatSmt_K(f)-S(f)$:
\vspace{0.00 in}
\begin{align*}
\E\hatSmt_K(f)-S(f) &= \int_{-1/2}^{1/2}S(f-f')\psi(f')\,df' - S(f') 
\\
&= \int_{-W}^{W}S(f-f')\psi(f')\,df' + \int_{\Omega}S(f-f')\psi(f')\,df' - S(f')
\\
&\ge \int_{-W}^{W}m_f\psi(f')\,df' + \int_{\Omega}0\psi(f')\,df' - M_f
\\
&= m_f(1-\Sigma^{(1)}_K) + 0 - M_f
\\
&= -(M_f-m_f)(1-\Sigma^{(1)}_K) - M_f\Sigma^{(1)}_K
\\
&\ge -(M_f-m_f)(1-\Sigma^{(1)}_K) - M\Sigma^{(1)}_K.
\end{align*} 
\\
From the above two bounds, we have $$\Bias\left[\hatSmt_K(f)\right] = \left|\E\hatSmt_K(f)-S(f)\right| \le (M_f-m_f)(1-\Sigma_K^{(1)})+M\Sigma_K^{(1)},$$ which establishes Theorem~\ref{thm:BiasNotDiff}.

\subsection{Proof of Theorem~\ref{thm:Variance}}
Since $\hatSmt_K(f) = \dfrac{1}{K}\left\|\mS_K^*\mE_f^*\vx\right\|_2^2$ where $\vx \sim \mathcal{CN}(\vct{0},\mR)$, by Lemma~\ref{lem:ExpCovNormGaussians}, we have $$\Var\left[\hatSmt_K(f)\right] = \Var\left[\dfrac{1}{K}\left\|\mS_K^*\mE_f^*\vx\right\|_2^2\right] = \dfrac{1}{K^2}\Cov\left[\left\|\mS_K^*\mE_f^*\vx\right\|_2^2,\left\|\mS_K^*\mE_f^*\vx\right\|_2^2\right] = \dfrac{1}{K^2}\left\|\mS_K^*\mE_f^*\mR\mE_f\mS_K\right\|_F^2.$$ 
\\
We focus on bounding the Frobenius norm of $\mS_K^*\mE_f^*\mR\mE_f\mS_K$. To do this, we first split it into two pieces - an integral over $[-W,W]$ and an integral over $\Omega = [-\tfrac{1}{2},\tfrac{1}{2}]\setminus [-W,W]$: 
\\
\begin{align*}
\mS_K^*\mE_f^*\mR\mE_f\mS_K &= \mS_K^*\mE_f^*\left(\int_{-1/2}^{1/2}S(f')\ve_{f'}\ve_{f'}^*\,df'\right)\mE_f\mS_K
\\
&= \mS_K^*\left(\int_{-1/2}^{1/2}S(f')\mE_f^*\ve_{f'}\ve_{f'}^*\mE_f\,df'\right)\mS_K
\\
&= \mS_K^*\left(\int_{-1/2}^{1/2}S(f')\ve_{f'-f}\ve_{f'-f}^*\,df'\right)\mS_K
\\
&= \mS_K^*\left(\int_{f-1/2}^{f+1/2}S(f+f')\ve_{f'}\ve_{f'}^*\,df'\right)\mS_K
\\
&= \mS_K^*\left(\int_{-1/2}^{1/2}S(f+f')\ve_{f'}\ve_{f'}^*\,df'\right)\mS_K
\\
&= \mS_K^*\left(\int_{-W}^{W}S(f+f')\ve_{f'}\ve_{f'}^*\,df' + \int_{\Omega}S(f+f')\ve_{f'}\ve_{f'}^*\,df'\right)\mS_K
\\
&= \mS_K^*\left(\int_{-W}^{W}S(f+f')\ve_{f'}\ve_{f'}^*\,df'\right)\mS_K + \mS_K^*\left(\int_{\Omega}S(f+f')\ve_{f'}\ve_{f'}^*\,df'\right)\mS_K.
\end{align*}

We will proceed by bounding the Frobenius norm of the two pieces above, and then applying the triangle inequality. Since $S(f) \le \displaystyle\max_{f \in \R}S(f) = M$ for all $f \in \R$, we trivially have $$\mS_K^*\left(\int_{\Omega}S(f+f')\ve_{f'}\ve_{f'}^*\,df'\right)\mS_K \preceq \mS_K^*\left(\int_{\Omega}M\ve_{f'}\ve_{f'}^*\,df'\right)\mS_K = \mS_K^*[M(\mId-\mB)]\mS_K = M(\mId-\mLambda_K).$$
\\
Then, since $\mP \preceq \mQ$ implies $\|\mP\|_F \le \|\mQ\|_F$, we have $$\left\|\mS_K^*\left(\int_{\Omega}S(f+f')\ve_{f'}\ve_{f'}^*\,df'\right)\mS_K\right\|_F \le \left\|M(\mId-\mLambda_K)\right\|_F = \sqrt{\sum_{k = 0}^{K-1}M^2(1-\lambda_k)^2} = M\sqrt{K}\Sigma^{(2)}_K.$$

Obtaining a good bound on the first piece requires a more intricate argument. We define $\mathbbm{1}_W(f) = 1$ if $f \in [-W,W]$ and $\mathbbm{1}_W(f) = 0$ if $f \in [-
\tfrac{1}{2},\tfrac{1}{2}]\setminus[-W,W]$. For convenience, we also extend $\mathbbm{1}_W(f)$ to $f \in \R$ such that $\mathbbm{1}_W(f)$ is $1$-periodic. With this definition, we can write
\begin{align*}
\int_{-W}^{W}S(f+f')\ve_{f'}\ve_{f'}^*\,df' &= \int_{-1/2}^{1/2}S(f+f')\mathbbm{1}_W(f')\ve_{f'}\ve_{f'}^*\,df'
\\
&= \int_{0}^{1}S(f+f')\mathbbm{1}_W(f')\ve_{f'}\ve_{f'}^*\,df'
\\
&= \sum_{\ell = 0}^{N-1}\int_{\tfrac{\ell}{N}}^{\tfrac{\ell+1}{N}}S(f+f')\mathbbm{1}_W(f')\ve_{f'}\ve_{f'}^*\,df'
\\
&= \sum_{\ell = 0}^{N-1}\int_{0}^{\tfrac{1}{N}}S(f+f'+\tfrac{\ell}{N})\mathbbm{1}_W(f'+\tfrac{\ell}{N})\ve_{f'+\tfrac{\ell}{N}}\ve_{f'+\tfrac{\ell}{N}}^*\,df'
\\
&= \int_{0}^{\tfrac{1}{N}}\sum_{\ell = 0}^{N-1}S(f+f'+\tfrac{\ell}{N})\mathbbm{1}_W(f'+\tfrac{\ell}{N})\ve_{f'+\tfrac{\ell}{N}}\ve_{f'+\tfrac{\ell}{N}}^*\,df',
\end{align*}
where the second line holds since $S(f)$, $\mathbbm{1}_W(f)$, and $\ve_f$ are $1$-periodic. Now, for any $f' \in \R$, the vectors $\left\{\tfrac{1}{\sqrt{N}}\ve_{f'+\tfrac{\ell}{N}}\right\}_{\ell = 0}^{N-1}$ form an orthonormal basis of $\C^N$. Hence, we have $$\left\|\sum_{\ell = 0}^{N-1}a_{\ell}\ve_{f'+\tfrac{\ell}{N}}\ve_{f'+\tfrac{\ell}{N}}^*\right\|_F^2 = N^2\sum_{\ell = 0}^{N-1}|a_{\ell}|^2$$ for any choice of coefficients $\{a_{\ell}\}_{\ell = 0}^{N-1}$ and offset frequency $f' \in \R$. By applying this formula, along with the triangle inequality and the Cauchy-Shwarz Integral inequality, we obtain 
\begin{align*}
\left\|\int_{-W}^{W}S(f+f')\ve_{f'}\ve_{f'}^*\,df'\right\|_F^2 &= \left\|\int_{0}^{\tfrac{1}{N}}\sum_{\ell = 0}^{N-1}S(f+f'+\tfrac{\ell}{N})\mathbbm{1}_W(f'+\tfrac{\ell}{N})\ve_{f'+\tfrac{\ell}{N}}\ve_{f'+\tfrac{\ell}{N}}^*\,df'\right\|_F^2
\\
&\le \left(\int_{0}^{\tfrac{1}{N}}\left\|\sum_{\ell = 0}^{N-1}S(f+f'+\tfrac{\ell}{N})\mathbbm{1}_W(f'+\tfrac{\ell}{N})\ve_{f'+\tfrac{\ell}{N}}\ve_{f'+\tfrac{\ell}{N}}^*\right\|_F\,df'\right)^2
\\
&= \left(\int_{0}^{\tfrac{1}{N}}N\left(\sum_{\ell = 0}^{N-1}S(f+f'+\tfrac{\ell}{N})^2\mathbbm{1}_W(f'+\tfrac{\ell}{N})^2\right)^{1/2}\,df'\right)^2
\\
&\le \left(\int_{0}^{\tfrac{1}{N}}N^2\,df'\right)\left(\int_{0}^{\tfrac{1}{N}}\sum_{\ell = 0}^{N-1}S(f+f'+\tfrac{\ell}{N})^2\mathbbm{1}_W(f'+\tfrac{\ell}{N})^2\,df'\right)
\\
&= N\sum_{\ell = 0}^{N-1}\int_{0}^{\tfrac{1}{N}}S(f+f'+\tfrac{\ell}{N})^2\mathbbm{1}_W(f'+\tfrac{\ell}{N})^2\,df'
\\
&= N\sum_{\ell = 0}^{N-1}\int_{\tfrac{\ell}{N}}^{\tfrac{\ell+1}{N}}S(f+f')^2\mathbbm{1}_W(f')^2\,df'
\\
&= N\int_{0}^{1}S(f+f')^2\mathbbm{1}_W(f')^2\,df'
\\
&= N\int_{-1/2}^{1/2}S(f+f')^2\mathbbm{1}_W(f')^2\,df'
\\
&= N\int_{-W}^{W}S(f+f')^2\,df'
\\
&= 2NWR_f^2
\end{align*}
\\
Since $\mS_K \in \R^{N \times K}$ is orthonormal, $\|\mS_K^*\mX\mS_K\|_F \le \|\mX\|_F$ for any Hermitian matrix $\mX \in \C^{N \times N}$. Hence, $$\left\|\mS_K^*\int_{-W}^{W}S(f+f')\ve_{f'}\ve_{f'}^*\,df'\mS_K\right\|_F \le \left\|\int_{-W}^{W}S(f+f')\ve_{f'}\ve_{f'}^*\,df'\right\|_F \le R_f\sqrt{2NW}.$$

Finally, by applying the two bounds we've derived, we obtain
\begin{align*}
\left\|\mS_K^*\mE_f^*\mR\mE_f\mS_K\right\|_F &= \left\|\mS_K^*\left(\int_{-W}^{W}S(f+f')\ve_{f'}\ve_{f'}^*\,df'\right)\mS_K + \mS_K^*\left(\int_{\Omega}S(f+f')\ve_{f'}\ve_{f'}^*\,df'\right)\mS_K\right\|_F
\\
&\le \left\|\mS_K^*\left(\int_{-W}^{W}S(f+f')\ve_{f'}\ve_{f'}^*\,df'\right)\mS_K\right\|_F + \left\|\mS_K^*\left(\int_{\Omega}S(f+f')\ve_{f'}\ve_{f'}^*\,df'\right)\mS_K\right\|_F
\\
&\le R_f\sqrt{2NW} + M\sqrt{K}\Sigma^{(2)}_K,
\end{align*}
and thus,
$$\Var\left[\hatSmt_K(f)\right] = \dfrac{1}{K^2}\left\|\mS_K^*\mE_f^*\mR\mE_f\mS_K\right\|_F^2 \le \dfrac{1}{K}\left(R_f\sqrt{\dfrac{2NW}{K}} + M\Sigma^{(2)}_K\right)^2.$$

\subsection{Proof of Theorem~\ref{thm:Covariance}}
Since $\hatSmt_K(f) = \dfrac{1}{K}\left\|\mS_K^*\mE_f^*\vx\right\|_2^2$ where $\vx \sim \mathcal{CN}(\vct{0},\mR)$, by Lemma~\ref{lem:ExpCovNormGaussians}, we have $$\Cov\left[\hatSmt_K(f_1),\hatSmt_K(f_2)\right] = \Cov\left[\dfrac{1}{K}\left\|\mS_K^*\mE_{f_1}^*\vx\right\|_2^2 , \dfrac{1}{K}\left\|\mS_K^*\mE_{f_2}^*\vx\right\|_2^2\right] = \dfrac{1}{K^2}\left\|\mS_K^*\mE_{f_1}^*\mR\mE_{f_2}\mS_K\right\|_F^2.$$
We focus on bounding the Frobenius norm of $\mS_K^*\mE_{f_1}^*\mR\mE_{f_2}\mS_K$. To do this, we first split it into three pieces - an integral over $[f_1-W,f_1+W]$, an integral over $[f_2-W,f_2+W]$, and an integral over $\Omega' = [-\tfrac{1}{2},\tfrac{1}{2}]\setminus([f_1-W,f_1+W] \cup [f_2-W,f_2+W])$: 
\begin{align*}
\mS_K^*\mE_{f_1}^*\mR\mE_{f_2}\mS_K &= \mS_K^*\mE_{f_1}^*\left(\int_{-1/2}^{1/2}S(f)\ve_f\ve_f^*\,df\right)\mE_{f_2}\mS_K
\\
&= \int_{-1/2}^{1/2}S(f)\mS_K^*\mE_{f_1}^*\ve_f\ve_f^*\mE_{f_2}\mS_K\,df
\\
&= \int_{-1/2}^{1/2}S(f)\mS_K^*\ve_{f-f_1}\ve_{f-f_2}^*\mS_K\,df
\\
&= \int_{f_1-W}^{f_1+W}S(f)\mS_K^*\ve_{f-f_1}\ve_{f-f_2}^*\mS_K\,df + \int_{f_2-W}^{f_2+W}S(f)\mS_K^*\ve_{f-f_1}\ve_{f-f_2}^*\mS_K\,df 
\\
& \ \ \ \ \ + \int_{\Omega'}S(f)\mS_K^*\ve_{f-f_1}\ve_{f-f_2}^*\mS_K\,df.
\end{align*}

By using the triangle inequality, the identity $\|\vx\vy^*\|_F = \|\vx\|_2\|\vy\|_2$ for vectors $\vx,\vy$, the Cauchy-Shwarz Inequality, and the facts that $\psi(f) \le \tfrac{N}{K}$ and $\int_{\Omega}\psi(f)\,df = \Sigma^{(1)}_K$, we can bound the Frobenius norm of the first piece by
\begin{align*}
\left\|\int_{f_1-W}^{f_1+W}S(f)\mS_K^*\ve_{f-f_1}\ve_{f-f_2}^*\mS_K\,df\right\|_F^2 &\le \left(\int_{f_1-W}^{f_1+W}\left\|S(f)\mS_K^*\ve_{f-f_1}\ve_{f-f_2}^*\mS_K\right\|_F\,df\right)^2
\\
&\le \left(\int_{f_1-W}^{f_1+W}S(f)\left\|\mS_K^*\ve_{f-f_1}\right\|_2\left\|\mS_K^*\ve_{f-f_2}\right\|_2\,df\right)^2
\\
&\le \left(\int_{f_1-W}^{f_1+W}S(f)^2\left\|\mS_K^*\ve_{f-f_1}\right\|_2^2\,df \right)\left(\int_{f_1-W}^{f_1+W}\left\|\mS_K^*\ve_{f-f_2}\right\|_2^2\,df\right)
\\
&= \left(\int_{-W}^{W}S(f+f_1)^2\left\|\mS_K^*\ve_f\right\|_2^2\,df \right)\left(\int_{f_1-f_2-W}^{f_1-f_2+W}\left\|\mS_K^*\ve_f\right\|_2^2\,df\right)
\\
&= \left(\int_{-W}^{W}S(f+f_1)^2K\psi(f)\,df \right)\left(\int_{f_1-f_2-W}^{f_1-f_2+W}K\psi(f)\,df\right)
\\
&\le \left(\int_{-W}^{W}S(f+f_1)^2 \cdot N\,df \right)\left(\int_{\Omega}K\psi(f)\,df\right)
\\
&= 2NWR_{f_1}^2 \cdot K\Sigma^{(1)}_K
\\
&= R_{f_1}^2 \cdot 2NWK\Sigma^{(1)}_K
\end{align*}
In a nearly identical manner, we can bound the second piece by $$\left\|\int_{f_2-W}^{f_2+W}S(f)\mS_K^*\ve_{f-f_1}\ve_{f-f_2}^*\mS_K\,df\right\|_F^2 \le R_{f_2}^2 \cdot 2NWK\Sigma^{(1)}_K.$$

The third piece can also be bounded in a similar manner, but the details are noticeably different, so we show the derivation:
\begin{align*}
\left\|\int_{\Omega'}S(f)\mS_K^*\ve_{f-f_1}\ve_{f-f_2}^*\mS_K\,df\right\|_F^2 &\le \left(\int_{\Omega'}\left\|S(f)\mS_K^*\ve_{f-f_1}\ve_{f-f_2}^*\mS_K\right\|_F\,df\right)^2
\\
&\le \left(\int_{\Omega'}S(f)\left\|\mS_K^*\ve_{f-f_1}\right\|_2\left\|\mS_K^*\ve_{f-f_2}\right\|_2\,df\right)^2
\\
&\le \left(\int_{\Omega'}M\left\|\mS_K^*\ve_{f-f_1}\right\|_2\left\|\mS_K^*\ve_{f-f_2}\right\|_2\,df\right)^2
\\
&\le M^2\left(\int_{\Omega'}\left\|\mS_K^*\ve_{f-f_1}\right\|_2^2\,df\right)\left(\int_{\Omega'}\left\|\mS_K^*\ve_{f-f_2}\right\|_2^2\,df\right)
\\
&= M^2\left(\int_{\Omega'}K\psi(f-f_1)\,df\right)\left(\int_{\Omega'}K\psi(f-f_2)\,df\right)
\\
&\le M^2\left(\int_{\Omega'_1}K\psi(f-f_1)\,df\right)\left(\int_{\Omega'_2}K\psi(f-f_2)\,df\right)
\\
&= M^2\left(\int_{\Omega}K\psi(f)\,df\right)\left(\int_{\Omega}K\psi(f)\,df\right)
\\
&= M^2\left(K\Sigma^{(1)}_K\right)^2
\end{align*}
where $\Omega'_1 = [-\tfrac{1}{2},\tfrac{1}{2}]\setminus[f_1-W,f_1+W]$ and $\Omega'_2 = [-\tfrac{1}{2},\tfrac{1}{2}]\setminus[f_2-W,f_2+W]$.

Finally, by applying the three bounds we've derived, we obtain
\begin{align*}
\left\|\mS_K^*\mE_{f_1}^*\mR\mE_{f_2}\mS_K\right\|_F &\le \left\|\int_{f_1-W}^{f_1+W}S(f)\mS_K^*\ve_{f-f_1}\ve_{f-f_2}^*\mS_K\,df\right\|_F + \left\|\int_{f_2-W}^{f_2+W}S(f)\mS_K^*\ve_{f-f_1}\ve_{f-f_2}^*\mS_K\,df\right\|_F
\\
& \ \ \ \ \ + \left\|\int_{\Omega'}S(f)\mS_K^*\ve_{f-f_1}\ve_{f-f_2}^*\mS_K\,df\right\|_F
\\
&\le R_{f_1}\sqrt{2NWK\Sigma^{(1)}_K} + R_{f_2}\sqrt{2NWK\Sigma^{(1)}_K} + MK\Sigma^{(1)}_K,
\end{align*}
and thus, 
$$0 \le \Cov\left[\hatSmt_K(f_1),\hatSmt_K(f_2)\right] = \dfrac{1}{K^2}\left\|\mS_K^*\mE_{f_1}^*\mR\mE_{f_2}\mS_K\right\|_F^2 \le \left((R_{f_1}+R_{f_2})\sqrt{\dfrac{2NW}{K}\Sigma^{(1)}_K} + M\Sigma^{(1)}_K\right)^2.$$

\subsection{Proof of Theorem~\ref{thm:Concentration}}
Since $\hatSmt_K(f) = \dfrac{1}{K}\left\|\mS_K^*\mE_f^*\vx\right\|_2^2$ where $\vx \sim \mathcal{CN}(\vct{0},\mR)$, by Lemma~\ref{lem:GaussianConcentration}, we have $$\P\left\{\hatSmt_K(f) \ge \beta\E\hatSmt_K(f)\right\} \le \beta^{-1}\e^{-\kappa_f(\beta-1-\ln \beta)} \quad \text{for} \quad \beta > 1,$$ and $$\P\left\{\hatSmt_K(f) \le \beta\E\hatSmt_K(f)\right\} \le \e^{-\kappa_f(\beta-1-\ln \beta)} \quad \text{for} \quad 0 < \beta < 1,$$ where $$\kappa_f = \dfrac{\tr\left[\mS_K^*\mE_f^*\mR\mE_f\mS_K\right]}{\left\|\mS_K^*\mE_f^*\mR\mE_f\mS_K\right\|}.$$

We can get an upper bound on $\mS_K^*\mE_f^*\mR\mE_f\mS_K$ in the Loewner ordering by splitting it into two pieces as done in the proof of Theorem~\ref{thm:Variance}, and then bounding each piece:
\begin{align*}
\mS_K^*\mE_f^*\mR\mE_f\mS_K &= \mS_K^*\left(\int_{-W}^{W}S(f+f')\ve_{f'}\ve_{f'}^*\,df'\right)\mS_K + \mS_K^*\left(\int_{\Omega}S(f+f')\ve_{f'}\ve_{f'}^*\,df'\right)\mS_K
\\
&\preceq \mS_K^*\left(\int_{-W}^{W}M_f\ve_{f'}\ve_{f'}^*\,df'\right)\mS_K + \mS_K^*\left(\int_{\Omega}M\ve_{f'}\ve_{f'}^*\,df'\right)\mS_K
\\
&= \mS_K^*\left(M_f\mB\right)\mS_K + \mS_K^*\left(M(\mId-\mB)\right)\mS_K
\\
&= M_f\mLambda_K + M(\mId-\mLambda_K)
\\
&= M_f\mId + (M-M_f)(\mId-\mLambda_K).
\end{align*}
Then, by using the fact that $\mP \preceq \mQ \implies \|\mP\| \le \|\mQ\|$ for PSD matrices $\mP$ and $\mQ$, we can bound, $$\left\|\mS_K^*\mE_f^*\mR\mE_f\mS_K\right\| \le \left\|M_f\mId + (M-M_f)(\mId-\mLambda_K)\right\| = M_f+(M-M_f)(1-\lambda_{K-1}).$$ 

We can also get a lower bound on $\tr[\mS_K^*\mE_f^*\mR\mE_f\mS_K] = K\E\left[\hatSmt_K(f)\right]$ by using the formula for $E\left[\hatSmt_K(f)\right]$ from Lemma~\ref{lem:MultitaperExpectation} along with the properties of the spectral window derived in Lemma~\ref{lem:SpectralWindow} as follows:
\begin{align*}
\tr[\mS_K^*\mE_f^*\mR\mE_f\mS_K] &= K\E\left[\hatSmt_K(f)\right]
\\
&= K\int_{-1/2}^{1/2}S(f-f')\psi(f')\,df'
\\
&\ge K\int_{-W}^{W}S(f-f')\psi(f')\,df'
\\
&= K\int_{-W}^{W}S(f-f')\left[\dfrac{N}{K}-\left(\dfrac{N}{K}-\psi(f')\right)\right]\,df'
\\
&= K\int_{-W}^{W}S(f-f')\dfrac{N}{K}\,df' - K\int_{-W}^{W}S(f-f')\left(\dfrac{N}{K}-\psi(f')\right)\,df'
\\
&= N\int_{f-W}^{f+W}S(f')\,df' - \int_{-W}^{W}S(f-f')(N-K\psi(f'))\,df'
\\
&\ge N\int_{f-W}^{f+W}S(f')\,df' - \int_{-W}^{W}M_f(N-K\psi(f'))\,df'
\\
&= 2NWA_f - \left(2NW-K\left(1-\Sigma^{(1)}_K\right)\right)M_f
\\
&= K\left(1-\Sigma^{(1)}_K\right)M_f - 2NW(M_f-A_f)
\end{align*}

Combining the upper bound on $\left\|\mS_K^*\mE_f^*\mR\mE_f\mS_K\right\|$ with the lower bound on $\tr[\mS_K^*\mE_f^*\mR\mE_f\mS_K]$, yields $$\kappa_f = \dfrac{\tr\left[\mS_K^*\mE_f^*\mR\mE_f\mS_K\right]}{\left\|\mS_K^*\mE_f^*\mR\mE_f\mS_K\right\|} \ge \dfrac{K\left(1-\Sigma^{(1)}_K\right)M_f - 2NW(M_f-A_f)}{M_f+\left(M-M_f\right)(1-\lambda_{K-1})}.$$

\section{Proof of Results in Section~\ref{sec:FastAlgorithms}}
\label{sec:FastAlgorithmsProofs}

\subsection{Fast algorithm for computing $\Psi(f)$ at grid frequencies}
\label{sec:FastPsi}
To begin developing our fast approximations for $\hatS^{\mt}_{K}(f)$, we first show that an eigenvalue weighted sum of $N$ tapered spectral estimates can be evaluated at a grid of frequencies $f \in [L]/L$ where $L \ge 2N$ in $O(L \log L)$ operations and using $O(L)$ memory. 
\begin{lemma}
\label{lem:FastPsi}
For any vector $\vx \in \C^N$ and any integer $L \ge 2N$, the quantity $$\Psi(f) := \sum_{k = 0}^{N-1}\lambda_k\hatS_k(f) \quad \text{where} \quad \hatS_k(f) = \left|\sum_{n = 0}^{N-1}\vs_k[n]\vx[n]e^{-j2\pi fn}\right|^2$$ can be evaluated at the grid frequencies $f \in [L]/L$ in $O(L \log L)$ operations and using $O(L)$ memory. 
\end{lemma}
\begin{proof}
Using eigendecomposition, we can write $\mB = \mS\mLambda\mS^*$, where $$\mS =  \begin{bmatrix}\vs_0 & \vs_1 & \cdots & \vs_{N-1}\end{bmatrix}$$ and $$\mLambda = \diag(\lambda_0,\lambda_1,\ldots,\lambda_{N-1}).$$ For any $f \in \R$, we let $\mE_f \in \C^{N \times N}$ be a diagonal matrix with diagonal entries $$\mE_f[n,n] = e^{j2\pi fn} \quad \text{for} \quad n \in [N].$$ With this definition, $\Psi(f)$ can be written as
\begin{align*}
\Psi(f) &= \sum_{k = 0}^{N-1}\lambda_k\hatS_k(f) 
\\
&= \sum_{k = 0}^{N-1}\lambda_k\left|\sum_{n = 0}^{N-1}\vs_k[n]\vx[n]e^{-j2\pi fn}\right|^2 
\\
&= \sum_{k = 0}^{N-1}\lambda_k\left|\vs_k^*\mE_f^*\vx\right|^2
\\
&= \sum_{k = 0}^{N-1}\mLambda[k,k]\left|(\mS^*\mE_f^*\vx)[k]\right|^2 
\\
&=  \vx^*\mE_f\mS\mLambda\mS^*\mE_f^*\vx 
\\
&= \vx^*\mE_f\mB\mE_f^*\vx
\end{align*}
\\
This gives us a a formula for $\Psi(f) = \sum_{k = 0}^{N-1}\lambda_k\hatS_k(f)$ which does not require computing any of the Slepian tapers. We will now use the fact that $\mB$ is a Toeplitz matrix to efficiently compute $\Psi(\tfrac{\ell}{L})$ for all $\ell \in [L]/L$.  

First, note that we can ``extend'' $\mB$ to a larger circulant matrix, which is diagonalized by a Fourier Transform matrix. Specifically, define a vector of sinc samples $\vb \in \R^{L}$ by $$\vb[\ell] = \begin{cases}\dfrac{\sin[2\pi W\ell]}{\pi \ell} & \text{if} \ \ell \in \{0,\ldots,N-1\} \\ 0 & \text{if} \ \ell \in \{N,\ldots,L-N\} \\ \dfrac{\sin[2\pi W(L-\ell)]}{\pi(L-\ell)} & \text{if} \ \ell \in \{L-N+1,\ldots,L-1\}\end{cases},$$ and let $\mBext \in \R^{L \times L}$ be defined by $$\mBext[m,n] = \vb[m-n \pmod{L}] \ \text{for} \ m,n \in [L].$$ 
\\
It is easy to check that $\mBext[m,n] = \mB[m,n] \ \text{for all} \ m,n \in [N]$, i.e., the upper-left $N \times N$ submatrix of $\mBext$ is $\mB$. Hence, we can write $$\mB = \mZ^*\mBext\mZ,$$ where $$\mZ = \begin{bmatrix}\mId_{N \times N} \\ \mtx{0}_{(L-N) \times N}\end{bmatrix} \in \R^{L \times N}$$ is a zeropadding matrix. Since $\mBext$ is a circulant matrix whose first column is $\vb$, we can write $$\mBext = \mF^{-1}\diag(\mF\vb)\mF$$ where $\mF \in \C^{L \times L}$ is an FFT matrix, i.e., $$\mF[m,n] = e^{-j2\pi mn/L} \ \text{for} \ m,n \in [L].$$ Note that with this normalization, the inverse FFT satisfies $$\mF^{-1} = \dfrac{1}{L}\mF^*,$$ as well as the conjugation identity $$\mF^{-1}\vy = \dfrac{1}{L}\overline{\mF\overline{\vy}} \quad \text{for all} \quad \vy \in \C^N.$$
\\
Next, for any $f \in \R$, let $\mD_f \in \C^{L \times L}$ be a diagonal matrix with diagonal entries $$\mD_f[m,m] = e^{j2\pi fm} \quad \text{for} \quad m \in [L],$$ and for each $\ell \in [L]$, let $\mC_{\ell} \in \C^{L \times L}$ be a cyclic shift matrix, i.e., $$\mC_{\ell}[m,n] = \begin{cases}1 & \text{if} \ n-m \equiv \ell \pmod{L} \\ 0 & \text{otherwise}\end{cases}.$$
\\
Since $\mE_f$ and $\mD_f$ are both diagonal matrices and $\mE_f[n,n] = \mD_f[n,n]$ for $n \in [N]$, we have 
$$\mZ\mE_f^* = \mD_f^*\mZ \quad \text{and} \quad \mE_f\mZ^* = \mZ^*\mD_f$$
for all $f \in \R$. Also, it is easy to check that cyclically shifting each column of $\mF$ by $\ell$ indices is equivalent to modulating the rows of $\mF$, or more specifically 
$$\mF\mD_{\ell/L}^* = \mC_{\ell}\mF \quad \text{and} \quad \mD_{\ell/L}\mF^* = \mF^*\mC_{\ell}^*$$ for all $\ell \in [L]$. Additionally, for any vectors $\vp, \vq \in \C^{L}$, we will denote $\vp \circ \vq \in \C^{L}$ to be the pointwise multiplication of $\vp$ and $\vq$, i.e., $$(\vp \circ \vq)[\ell] = \vp[\ell]\vq[\ell] \quad \text{for} \quad \ell \in [L],$$ and $\vp \circledast \vq \in \C^{L}$ to be the circluar cross-correlation of $\vp$ and $\vq$, i.e., $$(\vp \circledast \vq)[\ell] = \sum_{\ell' = 0}^{L-1}\overline{\vp[\ell']}\vq[\ell'+\ell \pmod{L}] \quad \text{for} \quad \ell \in [L].$$ Note that the circular cross-correlation of $\vp$ and $\vq$ can be computed using FFTs via the formula $$\vp \circledast \vq = \mF^{-1}(\overline{\mF\vp} \circ \mF\vq).$$ We will also use the notation $|\vp|^2 = \overline{\vp} \circ \vp$ for convenience.

We can now manipulate our formula for $\Psi(\tfrac{\ell}{L})$ as follows
\begin{align*}
\Psi(\tfrac{\ell}{L}) &:= \vx^*\mE_{\ell/L}\mB\mE_{\ell/L}^*\vx
\\
&= \vx^*\mE_{\ell/L}\mZ^*\mBext\mZ\mE_{\ell/L}^*\vx
\\
&= \vx^*\mE_{\ell/L}\mZ^*\mF^{-1}\diag(\mF\vb)\mF\mZ\mE_{\ell/L}^*\vx
\\
&= \dfrac{1}{L}\vx^*\mE_{\ell/L}\mZ^*\mF^*\diag(\mF\vb)\mF\mZ\mE_{\ell/L}^*\vx
\\
&= \dfrac{1}{L}\vx^*\mZ^*\mD_{\ell/L}\mF^*\diag(\mF\vb)\mF\mD_{\ell/L}^*\mZ\vx 
\\
&= \dfrac{1}{L}\vx^*\mZ^*\mF^*\mC_{\ell}^*\diag(\mF\vb)\mC_{\ell}\mF\mZ\vx
\\
&= \dfrac{1}{L}\sum_{\ell' = 0}^{L-1} (\mF\vb)[\ell'] \cdot \left|\left(\mC_{\ell}\mF\mZ\vx\right)[\ell']\right|^2
\\
&= \dfrac{1}{L}\sum_{\ell' = 0}^{L-1} (\mF\vb)[\ell'] \cdot \left|\left(\mF\mZ\vx\right)[\ell'+\ell \pmod{L}]\right|^2
\\
&= \left(\dfrac{1}{L}\overline{\mF\vb} \circledast \left|\mF\mZ\vx\right|^2\right)[\ell].
\\
&= \left(\mF^{-1}\left(\dfrac{1}{L}\overline{\mF\overline{\mF\vb}} \circ \mF\left|\mF\mZ\vx\right|^2\right)\right)[\ell]
\\
&= \left(\mF^{-1}\left(\mF^{-1}\mF\vb \circ \mF\left|\mF\mZ\vx\right|^2\right)\right)[\ell]
\\
&= \left(\mF^{-1}\left(\vb \circ \mF\left|\mF\mZ\vx\right|^2\right)\right)[\ell].
\end{align*} 
\\
Therefore, $$\begin{bmatrix}\Psi(\tfrac{0}{L}) & \Psi(\tfrac{1}{L}) & \cdots & \Psi(\tfrac{L-2}{L}) & \Psi(\tfrac{L-1}{L}) \end{bmatrix}^T = \mF^{-1}\left(\vb \circ \mF\left|\mF\mZ\vx\right|^2\right).$$ 

So to compute $\Psi(f)$ at all grid frequencies $f \in [L]/L$, we only need to compute $\mF^{-1}\left(\vb \circ \mF\left|\mF\mZ\vx\right|^2\right)$, which can be done in $O(L \log L)$ operations using $O(L)$ memory via three length-$L$ FFTs/IFFTs and a few pointwise operations on vectors of length $L$. 
\end{proof}
\subsection{Approximations for general multitaper spectral estimates}
\label{sec:FastApproxLemma}
Next, we present a lemma regarding approximations to spectral estimates which use orthonormal tapers. 

\begin{lemma}
\label{lem:Approx}
Let $\vx \in \C^N$ be a vector of $N$ equispaced samples, and let $\{\vv_k\}_{k = 0}^{N-1}$ be any orthonormal set of tapers in $\C^N$. For each $k \in [N]$, define a tapered spectral estimate $$V_k(f) = \left|\sum_{n = 0}^{N-1}\vv_k[n]\vx[n]e^{-j2\pi fn}\right|^2.$$ 
Also, let $\left\{\gamma_k\right\}_{k = 0}^{N-1}$ and $\left\{\tildegamma_k\right\}_{k = 0}^{N-1}$ be real coefficients, and then define a multitaper spectral estimate $\widehat{V}(f)$ and an approximation $\widetilde{V}(f)$ by $$\widehat{V}(f) = \sum_{k = 0}^{N-1}\gamma_k V_k(f) \quad \text{and} \quad  \widetilde{V}(f) = \sum_{k = 0}^{N-1}\widetilde{\gamma}_k V_k(f).$$
Then, for any frequency $f \in \R$, we have $$\left|\widehat{V}(f) - \widetilde{V}(f)\right| \le \left(\max_{k}\left|\gamma_k-\tildegamma_k\right|\right)\|x\|_2^2.$$
\end{lemma}

\begin{proof}
Let $\mV = \begin{bmatrix}\vv_0 & \cdots & \vv_{N-1}\end{bmatrix}$, and let $\mGamma$, $\widetilde{\mGamma} \in \R^{N \times N}$, and $\mE_f \in \C^{N \times N}$ be diagonal matrices whose diagonal entries are $\mGamma[n,n] = \gamma_n$, $\widetilde{\mGamma}[n,n] = \widetilde{\gamma}_n$, and $\mE_f[n,n] = e^{j2\pi fn}$ for $n \in [N]$. Then,
\begin{align*}
\widehat{V}(f) &= \sum_{k = 0}^{N-1}\gamma_k V_k(f) \\
& = \sum_{k = 0}^{N-1}\gamma_k\left|\sum_{n = 0}^{N-1}\vv_k[n]\vx[n]e^{-j2\pi fn}\right|^2\\
&= \sum_{k = 0}^{N-1}\gamma_k\left|\vv_k^*\mE_f^*\vx\right|^2 \\
&= \sum_{k = 0}^{N-1}\mGamma[k,k]\left|(\mV^*\mE_f^*\vx)[k]\right|^2 \\
&= \vx^*\mE_f\mV\mGamma\mV^*\mE_f^*\vx.
\end{align*}
In a nearly identical manner, $$\widetilde{V}(f) = \vx^*\mE_f\mV\widetilde{\mGamma}\mV^*\mE_f^*\vx.$$
Since $\mV$ is orthonormal, $\|\mV\| = \|\mV^*\| = 1$. Since $\mE_f$ is diagonal, and all the diagonal entries have modulus $1$, $\|\mE_f\| = \|\mE_f^*\| = 1$. Hence, for any $f \in \R$, we can bound
\begin{align*}
\left|\widehat{V}(f) - \widetilde{V}(f)\right| &= \left|\vx^*\mE_f\mV\left(\mGamma-\widetilde{\mGamma}\right)\mV^*\mE_f^*\vx\right| 
\\
&\le \|\vx\|_2 \|\mE_f\| \|\mV\|  \|\mGamma-\widetilde{\mGamma}\| \|\mV^*\| \|\mE_f^*\| \|\vx\|_2
\\
&= \|\mGamma-\widetilde{\mGamma}\|  \|\vx\|_2^2
\\
&= \left(\max_{k}|\gamma_k-\widetilde{\gamma}_k|\right)\|x\|_2^2,
\end{align*}
as desired.
\end{proof}

\subsection{Proof of Theorem~\ref{thm:ApproxMultitaper}}
\label{sec:ApproxMultitaperProof}
Recall that the indices $[N]$ are partitioned as follows: 
\begin{align*}
\setI_1 &= \{k \in [K] : \lambda_k \ge 1-\eps\}
\\
\setI_2 &= \{k \in [K] : \eps < \lambda_k < 1-\eps\}
\\
\setI_3 &= \{k \in [N] \setminus [K] : \eps < \lambda_k < 1-\eps\}
\\
\setI_4 &= \{k \in [N] \setminus [K] : \lambda_k \le \eps\}.
\end{align*}

Using the partitioning above, the unweighted multitaper spectral estimate $\hatSmt_K(f)$ can be written as 
\begin{align*}
\hatSmt_K(f) &= \dfrac{1}{K}\sum_{k = 0}^{K-1}\hatS_k(f) 
\\
&= \sum_{k \in \setI_1 \cup \setI_2}\dfrac{1}{K}\hatS_k(f),
\end{align*}
\\
and the approximate estimator $\tildeSmt_K(f)$ can be written as
\begin{align*}
\tildeSmt_K(f) &= \dfrac{1}{K}\Psi(f) + \dfrac{1}{K}\sum_{k \in \setI_2}(1-\lambda_k)\hatS_k(f) - \dfrac{1}{K}\sum_{k \in \setI_3}\lambda_k\hatS_k(f)
\\
&= \sum_{k = 0}^{N-1}\dfrac{\lambda_k}{K}\hatS_k(f) + \sum_{k \in \setI_2}\dfrac{1-\lambda_k}{K}\hatS_k(f) - \sum_{k \in \setI_3}\dfrac{\lambda_k}{K}\hatS_k(f)
\\
&= \sum_{k \in \setI_1\cup\setI_4}\dfrac{\lambda_k}{K}\hatS_k(f) + \sum_{k \in \setI_2}\dfrac{1}{K}\hatS_k(f)
\end{align*}
Thus, $\hatSmt_K(f)$ and $\tildeSmt_K(f)$ can be written as 
$$
\hatSmt_K(f) = \sum_{k = 0}^{N-1}\gamma_k\hatS_k(f) \quad \text{and} \quad \tildeSmt_K(f) = \sum_{k = 0}^{N-1}\tildegamma_k\hatS_k(f)
$$ 
where  
$$
\gamma_k = \begin{cases}1/K & k \in \setI_1 \cup \setI_2, \\ 0 & k \in \setI_3 \cup \setI_4, \end{cases} \quad \text{and} \quad \tildegamma_k = \begin{cases}\lambda_k/K & k \in \setI_1 \cup \setI_4, \\ 1/K & k \in \setI_2, \\ 0 & k \in \setI_3. \end{cases}
$$

We now bound $\left|\gamma_k - \tildegamma_k\right|$ for all $k \in [N]$. For $k \in \setI_1$, we have $\lambda_k \ge 1-\eps$, and thus, $$\left|\gamma_k - \tildegamma_k\right| = \left|\dfrac{1}{K} - \dfrac{\lambda_k}{K}\right| = \dfrac{1-\lambda_k}{K} \le \dfrac{\eps}{K}.$$ 
For $k \in \setI_2 \cup \setI_3$ we have $\gamma_k = \tildegamma_k$, i.e., $\left|\gamma_k - \tildegamma_k\right| = 0$. For $k \in \setI_4$, we have $\lambda_k \le \eps$, and thus, $$\left|\gamma_k - \tildegamma_k\right| = \left|0 - \dfrac{\lambda_k}{K}\right| = \dfrac{\lambda_k}{K} \le \dfrac{\eps}{K}.$$ 
Hence, $\left|\gamma_k - \tildegamma_k\right| \le \frac{\eps}{K}$ for all $k \in [N]$, and thus by Lemma~\ref{lem:Approx}, 
$$
\left|\hatSmt_K(f) - \tildeSmt_K(f)\right| \le \dfrac{\eps}{K}\|\vx\|_2^2$$ for all frequencies $f \in \R$.
 
\subsection{Proof of Theorem~\ref{thm:FastMultitaper}}
\label{sec:ApproxMultitaperProof}
To evaluate the approximate multitaper estimate $$\tildeSmt_K(f) = \dfrac{1}{K}\Psi(f) + \dfrac{1}{K}\sum_{k \in \setI_2}(1-\lambda_k)\hatS_k(f) - \dfrac{1}{K}\sum_{k \in \setI_3}\lambda_k\hatS_k(f)$$ at the $L$ grid frequencies $f \in [L]/L$ one needs to:
\begin{itemize}
\item For each $k \in \setI_2 \cup \setI_3$, precompute the Slepian basis vectors $\vs_k$ and eigenvalues $\lambda_k$.

Computing the Slepian basis vector $\vs_k$ and the corresponding eigenvalue $\lambda_k$ for a single index $k$ can be done in $O(N \log N)$ operations and $O(N)$ memory via the method described in \cite{Gruenbacher94}. This needs to be done for $\#(\setI_2 \cup \setI_3) = \#\{k : \eps < \lambda_k < 1-\eps\} = O(\log(NW)\log\tfrac{1}{\eps})$ values of $k$, so the total cost of this step is $O(N\log(N)\log(NW)\log\tfrac{1}{\eps})$ operations and $O(N\log(NW)\log\tfrac{1}{\eps})$ memory.

\item For $\ell \in [L]$, evaluate $\Psi(\tfrac{\ell}{L})$.

If $L \ge 2N$, then evaluating $\Psi(\tfrac{\ell}{L})$ for $\ell \in [L]$ can be done in $O(L \log L)$ operations and $O(L)$ memory as shown in Lemma~\ref{lem:FastPsi}. If $N \le L < 2N$, then $2L \ge 2N$, so by Lemma~\ref{lem:FastPsi}, we can evaluate $\Psi(\tfrac{\ell}{2L})$ for $\ell \in [2L]$ in $O(2L \log 2L) = O(L \log L)$ operations and $O(2L) = O(L)$ memory and then simply downsample the result to obtain $\Psi(\tfrac{\ell}{L})$ for $\ell \in [L]$.

\item For each $k \in \setI_2 \cup \setI_3$ and each $\ell \in [L]$, evaluate $\hatS_k(\tfrac{\ell}{L})$.

Evaluating $\hatS_k(\tfrac{\ell}{L}) = \left|\sum_{n = 0}^{N-1}\vs_k[n]\vx[n]e^{-j2\pi n\ell/L}\right|^2$ for all $\ell \in [L]$ can be done by pointwise multiplying $\vs_k$ and $\vx$, zeropadding this vector to length $L$, computing a length-$L$ FFT, and then computing the squared magnitude of each FFT coefficient. This takes $O(L \log L)$ operations and $O(L)$ memory. This needs to be done for $\#(\setI_2 \cup \setI_3) = O(\log(NW)\log\tfrac{1}{\eps})$ values of $k$, so the total cost of this step is $O(L\log L\log(NW)\log \tfrac{1}{\eps})$ operations and $O(L\log(NW)\log\tfrac{1}{\eps})$ memory.

\item For each $\ell \in [L]$, evaluate the weighted sum above for $\tildeSmt_K(\tfrac{\ell}{L})$. 

Once $\Psi(\tfrac{\ell}{L})$ and $\hatS_k(\tfrac{\ell}{L})$ for $k \in \setI_2 \cup \setI_3$ are computed, evaluating $\tildeSmt_K(\tfrac{\ell}{L})$ requires $O(\#(\setI_2 \cup \setI_3)) = O(\log(NW)\log\tfrac{1}{\eps})$ multiplications and additions. This has to be done for each $\ell \in [L]$, so the total cost is $O(L\log(NW)\log\tfrac{1}{\eps})$ operations.
\end{itemize}

Since $L \ge N$, the total cost of evaluating the approximate multitaper estimate $\tildeSmt_K(f)$ at the $L$ grid frequencies $f \in [L]/L$ is $O(L\log L \log(NW)\log\tfrac{1}{\eps})$ operations and $O(L\log(NW)\log\tfrac{1}{\eps})$ memory.


\end{document}